\documentclass[11pt,a4paper]{article}
\usepackage{amssymb}
\usepackage{amsmath, amscd}
\usepackage{amsfonts}
\usepackage{enumerate}
\usepackage{graphicx}
\usepackage{a4wide}
\usepackage{microtype}
\usepackage{bbm}
\usepackage{slashed}
\usepackage{times}
\usepackage{amsthm}
\usepackage{mathrsfs}
\usepackage{hyperref}
\usepackage{color}
\usepackage{mathtools}
\usepackage{enumitem}
\usepackage{soul}
\setul{.5pt}{.2pt}

\definecolor{darkred}{rgb}{0.8,0.1,0.1}
\hypersetup{
     colorlinks=true,         
     linkcolor=darkred,
     citecolor=blue,
}

\usepackage{comment}
\usepackage[margin=2.0cm]{geometry}
\usepackage[all,cmtip]{xy}

\usepackage{marvosym}

\DeclareMathAlphabet{\mathbfsf}{\encodingdefault}{\sfdefault}{bx}{n}

\theoremstyle{plain}

\theoremstyle{definition}

\definecolor{NiColor}{RGB}{77,77,255}
\definecolor{NiColoRed}{RGB}{255,77,77}
\definecolor{NiCitation}{RGB}{0,181,26}

\newtheoremstyle{TheoremStyle}
{3pt}
{3pt}
{\slshape}
{}
{\bf}
{:}
{.5em}
{}
\newtheoremstyle{ExampleAndRemarkStyle}
{3pt}
{3pt}
{\slshape}
{}
{\bf}
{:}
{.5em}
{}
\newtheoremstyle{ProofStyle}
{3pt}
{3pt}
{}
{}
{\bf}
{:}
{.5em}
{}

\theoremstyle{TheoremStyle}
\newtheorem{theorem}{Theorem}
\newtheorem{corollary}[theorem]{Corollary}
\newtheorem{proposition}[theorem]{Proposition}
\newtheorem{lemma}[theorem]{Lemma}
\newtheorem{Definition}[theorem]{Definition}
\theoremstyle{ExampleAndRemarkStyle}
\newtheorem{remark}[theorem]{Remark}
\newtheorem{Example}[theorem]{Example} 
\theoremstyle{ProofStyle}

\newcommand{\calg}{\mathfrak{Alg}_c}
\newcommand{\alg}{\mathfrak{Alg}_c}
\newcommand{\bkgg}{\mathfrak{BkgG}}
\newcommand{\vect}{\mathfrak{Vec}}
\newcommand{\arr}{\mathsf{Arr}}
\newcommand{\obj}{\mathsf{Obj}}
\newcommand{\A}{\mathcal{A}} 

\newcommand{\ie}{\emph{i.e.}}

\newcommand{\wick}[1]{:\!{#1}\!:}

\newcommand{\C}{\mathbb{C}}
\newcommand{\R}{\mathbb{R}}
\newcommand{\N}{\mathbb{N}}

\title{%
	The algebra of Wick polynomials of a scalar field\\ on a Riemannian manifold
}

\author{%
	Claudio Dappiaggi$^{1,2,3,a}$, Nicol\`o Drago$^{4,5,b}$ 
	and Paolo Rinaldi$^{1,2,3,c}$\vspace{4mm}\\
	{\small $^1$ Dipartimento di Fisica -- Universit{\`a} di Pavia, Via Bassi 6, 27100 Pavia, Italy.}\vspace{2mm}\\
	{\small $^2$ INFN, Sezione di Pavia -- Via Bassi 6, 27100 Pavia, Italy.}\vspace{2mm}\\
	{\small $^3$ Istituto Nazionale di Alta Matematica -- Sezione di Pavia, Via Ferrata, 5, 27100 Pavia, Italy.}\vspace{2mm}\\
	{\small $^4$ Dipartimento di Matematica -- Universit{\`a} di Trento, via Sommarive 15, I-38123 Povo (Trento), Italy.}\vspace{2mm}\\
	{\small $^5$ INFN, TIFPA -- via Sommarive 15, I-38123 Povo (Trento), Italy.}\vspace{4mm}\\
	{\footnotesize ~$^a$ claudio.dappiaggi@unipv.it~,~$^b$ 
		nicolo.drago89@gmail.com~,~$^c$ paolo.rinaldi01@universitadipavia.it}
}

\date{\today}

\begin{document}
	
	\maketitle
	
	\begin{abstract}
		On a connected, oriented, smooth Riemannian manifold without boundary we consider a real scalar field whose dynamics is ruled by $E$, a second order elliptic partial differential operator of Laplace type. Using the functional formalism and working within the framework of algebraic quantum field theory and of the principle of general local covariance, first we construct the algebra of locally covariant observables in terms of equivariant sections of a bundle of smooth, regular polynomial functionals over the affine space of the parametrices associated to $E$. Subsequently, adapting to the case in hand a strategy first introduced by Hollands and Wald in a Lorentzian setting, we prove the existence of Wick powers of the underlying field, extending the procedure to smooth, local and polynomial functionals and discussing in the process the regularization ambiguities of such procedure. Subsequently we endow the space of Wick powers with an algebra structure, dubbed E-product, which plays in a Riemannian setting the same r\^ole of the time ordered product for field theories on globally hyperbolic spacetimes. In particular we prove the existence of the E-product and we discuss both its properties and the renormalization ambiguities in the underlying procedure. As last step we extend the whole analysis to observables admitting derivatives of the field configurations and we discuss the quantum M\o ller operator which is used to investigate interacting models at a perturbative level.
	\end{abstract}
	

	
	\section{Introduction}\label{Section: Introduction}
	Algebraic quantum field theory is an axiomatic, mathematically rigorous framework which can be summarized as a two step approach \cite{Haag:1963dh}. In the first, one assigns to a physical system a $\ast$-algebra $\mathcal{A}$, whose elements are interpreted as observables, encompassing structural properties such as causality and the canonical commutation relations, see for example \cite{Benini:2013fia} for a review. In the second, one assigns to $\mathcal{A}$ a state, that is a positive and normalized linear functional, which allows via the GNS theorem to recover the standard probabilistic interpretation proper of quantum systems. This viewpoint has been very successful especially in the analysis of models of quantum field theories living on a globally hyperbolic Lorentzian spacetime, see for example \cite{Brunetti:2015vmh} for a recent collection of some notable results. In particular the algebraic approach has clarified and extended to curved backgrounds the analysis of interactions by means of perturbation theory and the associated renormalization ambiguities \cite{Brunetti-Duetsch-Fredenhagen-09,Rejzner:2016hdj}. The whole procedure is based on a few key ingredients. At the level of states, one needs to consider only those enjoying the so-called Hadamard condition, see for example \cite{Khavkine:2014mta}. This is a prescription on the form of the wavefront set of the two-point correlation functions of the underlying free field theory. It guarantees both that the ultraviolet behaviour of the quantum state coincides with that of the Poincar\'e vacuum and that the quantum fluctuations of all observables are finite. In addition one extends the collection of all observables to include also the Wick polynomials of the underlying fields endowed with a time-ordered product defining an algebra structure. This problem has been studied by several authors starting from \cite{Brunetti:1995rf} and particularly relevant are the seminal papers written by Hollands and Wald \cite{Hollands-Wald-01,Hollands-Wald-02}.
	We remark that, in these papers,
suitable analytic properties of the underlying structures have been assumed -- cf. \cite[Sec. 4.2]{Hollands-Wald-01,Hollands-Wald-02} --
	but such constraints have been recently weakened by Khavkine and Moretti in \cite{Khavkine:2014zsa, Khavkine-Melati-Moretti-17}.
	
	In almost all the analyses present in the literature, the problem of discussing interactions at a perturbative level in terms of Wick ordered quantum field has always been tackled under the assumptions that the underlying background is Lorentzian. Yet, in several instances it turns out that, if one considers models built on Riemannian manifolds, explicit calculations are often easier since one can use several tools and techniques coming from quantum statistical mechanics. In all these cases, these so-called Euclidean quantum field theories play only an auxiliary r\^ole and it is implicitly taken for granted that all results should be translated to a Lorentzian framework via a Wick rotation. This procedure is technically very delicate and it works only under very specific hypotheses, which have been investigated first by Osterwalder and Schrader \cite{Osterwalder:1973dx,Osterwalder:1974tc}. A further notable analysis in the algebraic framework can be found in \cite{Schlingemann:1998cw,Wald-79}.
	
	While the attitude of considering Euclidean quantum field theories only as an auxiliary tool is certainly justified in many instances, we are strongly advocating that this viewpoint is highly reductive. There exists a plethora of physically relevant models in quantum statistical mechanics, which are nothing but quantum field theories intrinsically defined on a Riemannian manifold. There are several examples ranging from Landau-Ginzburg theory to non-linear sigma models. The latter were recently studied within the framework of algebraic quantum field theory in connection to the derivation of Ricci flow \cite{CDDR18}.
    In all these cases there is no physical or mathematical reason to consider a Wick rotated version in a Lorentzian setting and therefore one needs to adopt an intrinsic viewpoint in which Euclidean field theories are studied independently from any Lorentzian counterpart.
    
    In this paper we adopt this perspective and we use the framework of algebraic quantum field theory considering a real scalar field on a smooth, oriented and connected Riemannian manifold of arbitrary dimension greater than $1$ and constructing the associated algebra of Wick polynomials. In our analysis we will be mainly inspired by \cite{Hollands-Wald-01,Hollands-Wald-02,Khavkine:2014zsa} who have solved completely this problem under the hypothesis that the underlying background is Lorentzian and globally hyperbolic. While we are strongly influenced by these papers, we stress that the problem that we are tackling is not a simple rewriting of these works in an Euclidean signature. Working on a Riemannian manifold leads to several structural and technical notable differences in comparison to the Lorentzian framework which we now highlight.
    
    As a matter of fact, let $(M,g)$ be a not necessarily compact, Riemannian, oriented, connected, smooth manifold of dimension $\dim M=D\geq 2$, such that $\partial M=\emptyset$.
    We consider on top of it a real scalar field $\phi:M\to\mathbb{R}$ whose dynamics is ruled by $E$, a second order, elliptic differential operator. Our first goal is to construct an algebra of observables associated to this system. To this end we employ the functional formalism, which has been successfully used in many instances in algebraic quantum field theory \cite{Brunetti-Duetsch-Fredenhagen-09,Rejzner:2016hdj}.
    Yet, contrary to the Lorentzian scenario, we do not consider the space of on-shell configurations and observables as functionals defined on this space, but we work only off-shell. The reason is two-fold. From the physical viewpoint, the lessons we learn from quantum statistical mechanics and from the state sum approach is that one needs to consider all accessible configurations and not only those selected by the equations of motion. From a mathematical and structural perspective, instead, information on $E$ is encoded in the associated fundamental solution $G$. Yet, working directly with it is problematic, since neither its existence nor its uniqueness are guaranteed, which is parameterized by the kernel of $E$, see \cite{Li_Tam}. For this reason one needs to consider in place of $G$ the collection of all parametrices associated to $E$, see {\it e.g.} \cite{Shubin,Wells}, which always exist yielding an inverse of $E$ up to smoothing operators. In sharp contrast with the Lorentzian framework, where an algebra of observables is constructed using the distinguished, uniquely defined, advanced and retarded fundamental solutions associated to a symmetric hyperbolic partial differential equation, our observables are constructed as equivariant sections of an affine bundle whose base space is the collection of all parametrices while the typical fiber is a vector space of regular and polynomial functionals. These are endowed with a fiberwise algebra structure induced by the parametrix of the operator $E$.
    
    The ensuing $\ast$-algebra, dubbed $\mathcal{A}_{\mathrm{reg}}(M;g)$ enjoys notable properties. Contrary to the Lorentzian counterpart, it is commutative as a consequences of the parametrices being symmetric. In addition the construction is functorial. Hence, following the same ideas of \cite{Brunetti:2001dx}, the assignment of $\mathcal{A}_{\mathrm{reg}}[M;h]$ to $(M;g)$ is local and covariant, thus allowing us to identify it as an Euclidean locally covariant quantum field theory. As a byproduct we can introduce the notion of a locally covariant observable, which includes as a special sub case that of a locally covariant quantum field.
    
    The subsequent goal of our investigation is to enlarge $\mathcal{A}_{\mathrm{reg}}[M;g]$ so as to include Wick polynomials while keeping the property that the construction is local and covariant. To this end we consider a larger class of functionals, namely those which are polynomial and local. The problem that we need to face is that the product defined on the algebra of regular functionals is not well-defined on this new class on account of the singular structure of the parametrices of $E$. In order to bypass this hurdle, we divide our analysis in two main steps. In the first one we focus on the so-called Wick powers, which are, roughly speaking, an integer power of a single quantum field configuration. We generalize the procedure outlined in the seminal papers \cite{Hollands-Wald-01,Hollands-Wald-02}, though our underlying framework follows that of \cite{Khavkine:2014zsa} in which it has been shown that one can consider smooth manifolds rather than those analytic, thanks to an application of the Peetre-Slov\'ak theorem, see \cite{Navarro-Sancho-14} and references therein. We prove existence of Wick ordered powers and we discuss and characterize the ambiguities in their construction. This part of our work generalizes to the case of a real scalar field the analysis for non-linear sigma models which appeared in \cite{CDDR18}. We mention that in \cite{Dang:2019byu,Dang-18} one can find a complementary analysis of the Wick squared scalar field on a compact Riemannian manifold.
    
    At this stage we can realize the second notable difference from the Lorentzian counterpart. In discussing the quantization of a field theory on a globally hyperbolic spacetime, one needs to deal with two distinguished algebra structures, the one induced by the so-called $\star$-product and that associated to the time-ordered product. The latter is the one relevant for endowing Wick polynomials with an algebra structure. In a Riemannian setting one deals with a single commutative product which is well defined on regular functionals while it needs to be extended also to Wick ordered powers, giving rise to what we refer to as E-product. Even in this case we prove existence of such product and we characterize the non-uniqueness of its definition which is the source of the renowned regularization ambiguities in the case in hand.  
    
    In the second main step of our paper, we extend our construction to account also for Wick polynomials containing derivatives of the field configurations, while keeping track of the covariance of the construction.
    In this procedure, following \cite{Hollands-Wald-05}, we need to add two further requirements in comparison to those needed to construct Wick powers.
    These are the Leibnitz rule and the principle of perturbative agreement (PPA), see also \cite{Drago-Hack-Pinamonti-2016} for a generalization.
    Both can be read as necessary consistency conditions and the second one entails heuristically that, in an interacting theory, every linear contribution to the equation of motion can be equivalently considered as part of the free theory or of the interaction without affecting the overall construction.
   	It is important to observe that, while implementing the Leibnitz rule appears to be harmless, in \cite{Hollands-Wald-05}, it has been shown that the PPA can fail for parity violating Lorentzian field theories -- actually, in the Lorentzian framework, it will always fail in two-dimensions.
    Such failure can be interpreted as an unavoidable ``anomaly" in the quantization procedure.
    Yet it is known that there exist instances where the PPA can be coherently implemented, {\it e.g.} charged Dirac fields, \cite{Zahn:2013ywa}.
    Finally it is worth mentioning that our results are complementary to those obtained by Keller in \cite{Keller-09,Keller:2010xq} in the analysis of Epstein-Glaser renormalization in an Euclidean framework.
    We remark that the net of algebras that we obtain at the end of our construction seems to bear similarities with factorization algebras and, in our opinion, it would be worth making a detailed comparison along the same lines of \cite{Gwilliam-Rejzner-17,Benini_Perin_Schenkel} in the Lorentzian setting.
    
    As last step, we investigate the structure of the $\ast$-algebra of observables when the underlying Lagrangian is not only quadratic in the field configurations but it contains also an interacting term. This is codified in a local perturbation, so that it can be analysed in the framework of pAQFT as described in \cite{Brunetti:2015vmh, Rejzner:2016hdj}. The key point of this approach is the possibility to realize every local and covariant observable of the interacting theory as a formal power series in the algebra of the underlying free field theory. This is encoded in a linear and covariant map $\mathsf{R}_V$, dubbed quantum M\o ller operator, whose construction is analysed in the framework of Euclidean locally covariant theories.
    The outcome is that $\mathsf{R}_V$ is both local and covariant only if one selects a fundamental solution $G$ of the underlying elliptic operator $E$ -- \textit{cf.} Section \ref{Section: Interacting models} for a more detailed discussion.
    Contrary to the Lorentzian scenario, where such selection is locally covariant when working with the category of globally hyperbolic spacetimes, this is not the case in the Euclidean regime. Hence, to bypass this hurdle, one must encode the choice of $G$ as a background datum in the underlying category in order to restore local covariance. This procedure generalizes a similar strategy followed in \cite[Sec. 6]{Benini:2013ita} when dealing with the failure of isotony in the analysis of the interplay between the principle of general local covariance and the quantization of Abelian gauge theories.
    
    \vskip .3cm
    
    \noindent The paper is organized as follows: in Section \ref{Section: General Setting} we fix notation and conventions, while in Section \ref{Section: Euclidean Locally Covariant Field Theories} we introduce the notion of an Euclidean locally covariant field theory, proving that a real scalar field on a smooth, connected, oriented Riemannian manifold, whose dynamics is ruled by a second order, elliptic differential operator can be described within this framework. In Section \ref{Section: Locally Covariant Observables and Quantum Fields} we introduce the notion of locally covariant observables as a preliminary step to discuss Wick ordered powers of quantum fields. This is the core of Section \ref{Section: Wick Ordered Powers of Quantum Fields} in which we discuss Wick powers, their existence and the ambiguities in their construction. Subsequently we investigate how to endow Wick powers with an algebra structure. In Section \ref{Section: E-Product of Wick Ordered Powers of Quantum Fields} we discuss the so-called E-product which is a local and covariant extension of the one introduced in Section \ref{Section: Euclidean Locally Covariant Field Theories} for regular functionals. Also in this case we prove existence of the E-product and we discuss the ambiguities in its construction. In Section \ref{Section: The Case of a Scalar Field with Derivatives} we extend our analysis to account also for Wick polynomials including derivatives of the field configurations. This forces us to introduce two new requirements, the Leibnitz rule and the principle of perturbative agreement which are discussed in detail. In Section \ref{Section: Interacting models} we discuss the $\ast$-algebra of observables of interacting field theories in the framework of perturbative algebraic quantum field theory. In particular we study the quantum M\o ller operator and its interplay with locality and covariance. Finally in Appendix \ref{Appendix: Peetre-Slovak} we recall one of the main results that we use, namely the Peetre-Slov\'ak theorem. 
    
\subsection{General Setting}\label{Section: General Setting}
	
Goal of this section is to fix notations and conventions, introducing the key geometric and analytic structures, which play a r\^ole in this work. With $(M,g)$ we denote a connected, oriented and smooth Riemannian manifold, with $\dim M=\mathrm{D}\geq 2$. In addition, for simplicity we assume that $M$ has empty boundary, {\it i.e.} $\partial M=\emptyset$.
Notice that we are not assuming that $M$ is compact.
On top of $M$, we consider a real scalar field $\varphi:M\to\mathbb{R}$, whose associated space of real-valued kinematic configurations is $\mathcal{E}(M)\equiv C^\infty(M;\mathbb{R})$. In this paper we shall adopt the notation $C^\infty(M)\equiv C^\infty(M;\mathbb{R})$.
Borrowing the nomenclature from the Lorentzian realm, dynamical configurations are the extrema of the Lagrangian density,
	\begin{align}\label{Eqn: Lagrangian Density}
		\mathcal{L}_E[\varphi]\vcentcolon=\langle\varphi,E\varphi\rangle\mu_g, \qquad\varphi\in \mathcal{E}(M),
	\end{align}
	where $\mu_g$ is the metric induced volume form, while $\langle,\rangle$ stands for the pointwise product between smooth functions.
	In addition $E:\mathcal{E}(M)\to\mathcal{E}(M)$ is a generic operator of Laplace type, that is a formally self-adjoint second order elliptic partial differential operator whose principal symbol is $g_{ij}(x)k^ik^j$ for every $x\in M$ and for every $k\in T^*_xM$.
	Hence, in every local chart, such operator reads
	\begin{equation}\label{Eqn: local_form_E}
		E\vcentcolon
		=-(\nabla_j-A_j)g^{jk}(\nabla_k+A_k)+c\,,
	\end{equation}
	where $\nabla$ stands for the covariant derivative, $A\in\Gamma(T^*M)$ while  $c\in C^\infty(M)$. 
	If both $A$ and $c$ vanish, $E$ coincides with the Laplace-Beltrami operator built out of $g$. In the following we shall consider $(g,A,c)$ as background structures and it is important to evaluate their so-called \emph{engineering dimensions} $d_\varphi, d_A,d_c\in\R$. These coefficients are determined by considering the scaling transformations  
	\begin{align*}
		g\mapsto g_\lambda\vcentcolon=\lambda^{-2}g,\quad\varphi\mapsto\varphi_\lambda\vcentcolon=\lambda^{d_\varphi}\varphi,\quad A\mapsto A_\lambda\vcentcolon=\lambda^{d_A}A\quad c\mapsto c_\lambda\vcentcolon=\lambda^{d_c}c,
	\end{align*}
	and requiring the Lagrangian density to be invariant under such transformations, namely
	\begin{align*}
		\mathcal{L}_{E_\lambda}[\varphi_\lambda]\equiv\mathcal{L}[\varphi_\lambda, g_\lambda, A_\lambda, c_\lambda]=\mathcal{L}[\varphi, g, A, c]\equiv\mathcal{L}_E[\varphi]. 
	\end{align*}
	Recalling Equation \eqref{Eqn: Lagrangian Density} and the scaling behaviour of the volume measure $\mu_{\lambda^{-2}g}=\lambda^{-\mathrm{D}}\mu_g$, a straightforward computation yields
	\begin{align}\label{Eqn: Engineering Dimension}
		d_\varphi=\frac{\mathrm{D}-2}{2},\qquad d_A=0,\qquad d_c=2.
	\end{align}
\begin{remark}\label{Remark:general_framework}
		We observe that an equivalent framework, which we could have considered, consists of picking as basic data a connected, oriented, smooth manifold $M$, still for simplicity with empty boundary, together with a generic second order elliptic differential operator $K$ acting on scalar function. In this context one can endow $M$ with a smooth Riemannian metric defined directly out of the principal symbol of $K$. Hence, while, on the one hand, opting for $E$ as in \eqref{Eqn: local_form_E} does not entail a loss of generality, on the other hand, it is a more convenient setting to emphasize and to analyze the r\^ole of general local covariance in the next sections.
\end{remark}
	
	\section{Euclidean Locally Covariant Field Theories}\label{Section: Euclidean Locally Covariant Field Theories}
	In this section, we have a twofold goal. First of all we define the notion of an Euclidean locally covariant field theory and secondly we prove that the model of a real scalar field as per \eqref{Eqn: Lagrangian Density} and \eqref{Eqn: local_form_E} fits in this scheme. To this end, we shall make use of the language of categories following the same ideas developed for the first time in the Lorentzian setting in the seminal work \cite{Brunetti:2001dx}. In this endeavour we follow in spirit and we extend partly the framework of \cite{CDDR18}. Hence we start by defining the basic ingredients:	

\begin{enumerate}
\item $\bkgg$ denotes the category of \emph{background geometries}, such that
	\begin{itemize}
		\item $\obj(\bkgg)$ is the collection of pairs $(M;h)$, where $M$ denotes a smooth, connected and oriented manifold with empty boundary and with $\dim M=\mathrm{D}\geq 2$, whereas $h\equiv (g, A, c)$ identifies the background data, that is $A\in\Gamma(T^*M)$, $c\in C^\infty(M)$ while $g\in \Gamma(S^2T^*M)$ is a Riemannian metric;
		\item $\arr(\bkgg)$ is the collection of morphisms between $(M;h),(M^\prime;h^\prime)\in\obj(\bkgg)$ which are specified by an orientation preserving isometric embedding between $\chi:M\to M^\prime$ such that $h=\chi^*h^\prime$ where $h^\prime\equiv(g^\prime,A^\prime,c^\prime)$.
	\end{itemize}
\item $\calg$ is the category whose objects are unital, commutative $^*$-algebras while the arrows are unit preserving, injective $^*$-homomorphisms.
\item $\vect$ is the category whose objects are real vector spaces whereas whose arrows are injective linear morphisms.
\end{enumerate}

\begin{remark}\label{Remark: Scaling transformation}
Notice that, similarly to \cite{Khavkine:2014zsa} and to \cite{CDDR18}, $\bkgg$ enjoys the property of being \emph{dimensionful}, \emph{i.e.}, in view of \eqref{Eqn: Engineering Dimension}, it is endowed with an action of $\R_+\vcentcolon=(0,\infty)$ on $\obj(\bkgg)$ 
\begin{align}\label{Eqn: Scaling Transformation}
(M;h)=(M;g,A,c)\mapsto(M;h_\lambda)\vcentcolon=(M; g_\lambda,A_\lambda,c_\lambda)\vcentcolon=(M; \lambda^{-2}g, A,\lambda^2c),
\end{align}
which is preserved by the arrows of $\bkgg$.
\end{remark}

\begin{Definition}\label{Def: Euclidean locally covariant theory}
A (scalar) \emph{Euclidean locally covariant field theory} is a pair $(\mathcal{A},\lbrace\varsigma_{\lambda,\mu}\rbrace_{\lambda,\mu\in(0,+\infty)})$ made of the following data:
\begin{enumerate}
	\item
	$\mathcal{A}$ is a covariant functor $\A\vcentcolon\bkgg\to\calg$ .
	For any $\lambda,\mu>0$, let $\A_\lambda\vcentcolon=\A\circ\rho_\lambda\vcentcolon\bkgg\to\calg$ be the covariant functor where $\rho_\lambda\vcentcolon\bkgg\to\bkgg$ is the functor acting as the identity on $\arr(\bkgg)$ and according to \eqref{Eqn: Scaling Transformation} on $\obj(\bkgg)$.
	\item Then for all $\lambda,\mu\in(0,+\infty)$, $\varsigma_{\lambda,\mu}$ is a natural isomorphism $\varsigma_{\lambda,\mu}\colon\colon\A_\mu\to\A_\lambda$ such that, for every $\lambda,\mu,\rho\in\mathbb{R}_+$, 
		\begin{align}
			\varsigma_{\lambda,\mu}[M;h]=
			\varsigma_{\lambda,\rho}[M;h]\circ\varsigma_{\rho,\mu}[M;h]\,,\qquad
			\varsigma_{\lambda,\lambda}[M;h]=\operatorname{Id}_{\mathcal{A}[M;h]}\,.
		\end{align}
		for any $(M;h)\in\obj(\bkgg)$.
		For the sake of brevity, in the following we shall write $\varsigma_\lambda[M;h]:=\varsigma_{1,\lambda}[M;h]$.
\end{enumerate}

\end{Definition}

\begin{remark}
	Observe that, in comparison to \cite{CDDR18}, we have strengthened the definition of an Euclidean locally covariant theory by hard coding the requirement that the $\ast$-algebra associated to each background geometry is commutative. As we will show, in the model in hand this requirement is a natural byproduct of the structural property of the elliptic operator $E$ defined in \eqref{Eqn: local_form_E}.
\end{remark}

\begin{remark}
Notice that $\varsigma_\lambda$ is such that the scaling transformation of Equation \eqref{Eqn: Scaling Transformation} is implemented coherently in the theory described by the functor $\mathcal{A}$, hence entailing that $\mathcal{A}_\lambda$ can be interpreted as the functor describing the theory $\mathcal{A}$ at the scale $\lambda$.
\end{remark}

\subsection{The scalar field as an Euclidean locally covariant field theory}

We are now in position to reformulate the model ruled by the Lagrangian density \eqref{Eqn: Lagrangian Density} as an Euclidean locally covariant theory. To this end, we start by considering an arbitrary but fixed background geometry $(M;h)\in\obj(\bkgg)$, showing how to build a unital, commutative $\ast$-algebra $\A[M;h]$ associated with the Lagrangian density \eqref{Eqn: Lagrangian Density} -- \textit{cf.} definition \ref{Def: locally covariant algebra of regular observables} and proposition \ref{Prop: local and covariant algebra of local polynomials}.

Hence, let $(M;h)\in\obj(\bkgg)$ and let $E$ be the operator \eqref{Eqn: Lagrangian Density}. Being elliptic and formally self-adjoint it admits \cite[Th. 4.4]{Wells} a symmetric operator $\widetilde{P}\vcentcolon \mathcal{D}(M)\to \mathcal{E}(M)$ which is unique up to smoothing operators such that 
\begin{align}\label{Eqn: defining property of parametrix}
\widetilde{P}E-\mathrm{Id}_{\mathcal{D}(M)}\in
\mathcal{E}(M\times M)\,,\qquad
E\widetilde{P}-\mathrm{Id}_{\mathcal{D}(M)}\in
\mathcal{E}(M\times M)\,.
\end{align}
In addition, observe that each operator $\widetilde{P}$ identifies an associated parametrix, that is a bi-distribution $P\in\mathcal{D}^\prime(M\times M)$, such that, for all $f,f^\prime\in \mathcal{D}(M)$, $P(f\otimes f^\prime)=P(f^\prime\otimes f)=\langle \widetilde{P}(f),f^\prime\rangle_g$ where $\langle,\rangle_g$ is the metric induced pairing between $\mathcal{E}(M)$ and $\mathcal{D}(M)$.
The singularities of $P$ are codified in its wavefront set, which, as a consequence of \cite[Corol. 8.3.2]{Hormander-83}, reads
\begin{align}\label{Eqn: Parametrix wave front set}
\mathrm{WF}(P)=\lbrace(x, k_1;x,k_2)\in T^*(M\times M)\setminus\lbrace0\rbrace\;|\; k_1+k_2=0\rbrace.
\end{align}
In the following, we will denote with $\mathrm{Par}[M;h]$ the set of symmetric parametrices associated with the theory on $(M;h)$. In view of Equation \eqref{Eqn: defining property of parametrix}, $\mathrm{Par}[M;h]$ is an affine space modeled on $\mathcal{E}(M\times M)$.

\begin{remark}\label{Remark: Hadamard representation}
	Recall that each parametrix $P\in\mathcal{D}^\prime(M\times M)$ associated to the elliptic operator $E$ admits a Hadamard representation \cite[Chap. 5]{Garabedian-98}.	For an arbitrary but fixed $x_0\in M$ let $O$ be a convex geodesic neighbourhood centered at $x_0$.
	Then for all $x,y\in O$, the associated integral kernel reads
	\begin{align}\label{Eqn: Hadamard representation}
		P(x,y)=H(x,y)+W_P(x,y), \quad H(x,y)=\frac{U(x,y)}{\sigma^\frac{\mathrm{D}-2}{2}(x,y)}+V(x,y)\log\left(\frac{\sigma(x,y)}{\nu^2}\right),
	\end{align}
	where $\nu\in\mathbb{R}$ is an arbitrary reference length, $\sigma(x,y)$ is the halved squared geodesic distance between $x$ and $y$ while $U,V,W_P\in \mathcal{E}(O\times O)$ are symmetric functions with $V=0$, if $\mathrm{D}$ is odd.
	The coefficients $U,V$ in \eqref{Eqn: Hadamard representation} are defined in terms of a formal power series in $\sigma$, that is, $V(x,y)=\sum_nV_n(x,y)\sigma(x,y)^n$, $U(x,y)=\sum_nU_n(x,y)\sigma(x,y)^n$.
	The functions $V_n(x,y), U_n(x,y)$ satisfy a hierarchical system of transport equations, built only out of the background geometric data $(M;h)$ and of the elliptic operator $E$.
	The series defining $U,V$ can be made convergent locally by introducing suitable cut-off functions which do not alter the singular behaviour in the limit $x\to y$ -- \textit{cf.} \cite[Sec. 5.2]{Hollands-Wald-01}.
	$H$ is also known as the Hadamard parametrix and it codifies locally the singular structure of $P$.
	Moreover notice that, although $W_P(x,y)$ is well-defined only for $x,y\in\mathcal{O}$, its coinciding point limit $[W_P](x):=W_P(x,x)$ can be extended, via a partition of unity argument, to a globally well-defined function $[W_P]\in\mathcal{E}(M)$. The procedure does not depend on the chosen partition of unity.
\end{remark}

Having introduced the key structures, our strategy is to consider an arbitrary but fixed parametrix $P\in\mathrm{Par}[M;h]$ building a unital $^*$-algebra associated to the theory ruled by the operator $E$ as in \eqref{Eqn: local_form_E}. At a later stage, we will show how to remove the dependence from the parametrix chosen. Therefore we need to define suitable classes of functionals -- see {\it e.g.} \cite{Brouder-Dang-Gengoux-Rejzner-18},

\begin{Definition}\label{Def: functionals}
A functional $F\vcentcolon \mathcal{E}(M)\to\C$ is called:
\begin{itemize}
\item \emph{smooth} if, for any $\varphi,\varphi_1,\dots\varphi_k\in \mathcal{E}(M)$, with $k\geqslant1$, the $k$-th functional derivative $F^{(k)}[\varphi]$, defined as
\begin{align*}
	\big\langle F^{(k)}[\varphi],\varphi_1\otimes\ldots\otimes\varphi_k\big\rangle\vcentcolon=
			\frac{\partial^k}{\partial s_1\dots\partial s_k}F\left(\varphi+\sum_{i=1}^{k}s_i\varphi_i\right)\bigg|_{s_1=\ldots s_k=0}\,,
\end{align*}
identifies a symmetric and compactly supported distribution, namely $F^{(k)}[\varphi]\in \mathcal{E}^\prime(M^k)$ where $M^k\vcentcolon=\underbrace{M\times...\times M}_k$,
\item \emph{regular} if,  for any $\varphi,\varphi_1,\dots,\varphi_k\in \mathcal{E}(M)$, with $k\geqslant1$,  $F^{(k)}[\varphi]\in \mathcal{D}(M^k)$;
\item \emph{polynomial} if it has only a finite number of non-vanishing functional derivatives;
\item \emph{compactly supported} if $\overline{\bigcup\limits_{\varphi\in\mathcal{E}(M)}\mathrm{supp}(F^{(1)}[\varphi])}$ is compact;
\item \emph{local} if, for all $k\in\N$,
	\begin{itemize}
	\item $\mathrm{supp}(F^{(k)}[\varphi])\subset \mathrm{Diag}(M^k)\vcentcolon=\lbrace(x_1,\dots, x_k)\in M^k\,|\,x_1=\dots=x_k\rbrace$, for all $\varphi\in \mathcal{E}(M)$;
	\item for all $\varphi\in\mathcal{E}(M)$, the wave front set $\operatorname{WF}(F^{(k)}[\varphi])$ is contained in $D_k$, the conormal of the thin diagonal, that is $D_k\vcentcolon=\left\{(x_1,\zeta_1,...,x_k,\zeta_k)\in T^*(M^k)\;|\;x_1=\ldots=x_k\,\mathrm{and}\,\sum\limits_{i=1}^k\zeta_i=0\right\}$.
	\end{itemize}
\end{itemize}
We denote with $\mathcal{P}_{\mathrm{reg}}[M;h]$ (resp. $\mathcal{P}_{\operatorname{loc}}[M;h]$) the set of polynomial and regular (resp. polynomial and local), compactly supported functionals on $M$.
We also denote with $\mathcal{P}[M;h]$ the commutative and associative, unital $\ast$-algebra generated by $\mathcal{P}_{\operatorname{loc}}[M;h]$ with respect to the pointwise product. The $\ast$-involution is induced by complex conjugation.
\end{Definition}

\begin{remark}\label{Remark: functional without dependence on derivative of the field configuration}
	In this section we are implicitly assuming that all functionals $F$ are such that $F(\varphi)$ does not depend on the derivatives of $\varphi$, being in addition polynomial.
	For example we are excluding functionals such as  $F(\varphi):=\int_M\mu_g(x)\mu_g(y)\,\omega^{ab}(x,y)\varphi(x)\partial_a\varphi(y)\partial_b\varphi(y)$ where $\omega\in\Gamma_{\mathrm{c}}(TM\boxtimes TM)$.
	We shall remove this limitation in Section \ref{Section: The Case of a Scalar Field with Derivatives}.
\end{remark}

\begin{proposition}\label{Prop: regular algebra}
The vector space $\mathcal{P}_{\mathrm{reg}}[M;h]$ of smooth, regular and polynomial functionals is an associative and commutative $^*$-algebra if endowed with the following product: for any $F,G\in\mathcal{P}_{\mathrm{reg}}[M;h]$,
\begin{align}\label{Eqn: Algebra product}
	(F\cdot_PG)(\varphi)=
	F(\varphi)G(\varphi)+\sum_{k\geqslant1}\frac{1}{k!}\langle F^{(k)}[\varphi], P^{\otimes k}G^{(k)}[\varphi]\rangle,
\end{align}
where
$P^{\otimes k}G^{(k)}[\varphi]\in \mathcal{E}(M^k)$ is the extension of $\underbrace{P\otimes...\otimes P}_k$ to $G^{(k)}[\varphi]$ according to \cite[Thm. 8.2.13]{Hormander-83}.
The $^*$-involution on $\mathcal{P}_{\mathrm{reg}}[M;h]$ is completely fixed by $F^*(\varphi)=\overline{F(\varphi)}$ for all $F\in\mathcal{P}_{\mathrm{reg}}[M;h]$. We denote with $\mathcal{F}_{\operatorname{reg},P}[M;h]$ the $^*$-algebra $(\mathcal{P}_{\mathrm{reg}}[M;h], \cdot_P,*)$.
\end{proposition}
\begin{proof}
	First of all, notice that \eqref{Eqn: Algebra product} is well defined. On the one hand, the functional $F$ and $G$ being regular, their derivatives $F^{(k)}[\varphi]$ and $G^{(k)}[\varphi]$ identify smooth and compactly supported functions and thus every term $\langle F^{(k)}[\varphi], P^{\otimes k}G^{(k)}[\varphi]\rangle$ in the sum is well defined. On the other hand, $F$ and $G$ being polynomial, only a finite number of non vanishing terms appear in the sum, guaranteeing convergence. Finally, associativity holds per construction whereas commutativity is a by product of each parametrix of $E$ being symmetric.
\end{proof}

\noindent Notice that $(M;h)\mapsto\mathcal{F}_{\operatorname{reg},P}[M;h]$ falls short from identifying an Euclidean locally covariant field theory in the sense of Definition \ref{Def: Euclidean locally covariant theory} since this construction requires the choice of an arbitrary parametrix $P\in\operatorname{Par}[M;h]$. Our next goal is the removal of this arbitrariness. The first step consists of proving that different choices of parametrix lead to algebras which are $\ast$-isomorphic. The next proposition makes this statement precise and since its proof is identical, mutatis mutandis, to that of \cite[Prop. 1.4.7]{Lindner:2013ila}, \cite[Prop. II.4]{Keller-09}, we omit it.

\begin{proposition}\label{Prop: star-isomorphism between different algebras}
Consider an arbitrary but fixed $(M;h)\in\obj(\bkgg)$ and let $P,Q\in\operatorname{Par}[M;h]$. Then the $^*$-algebras $\mathcal{F}_{\operatorname{reg},P}[M;h]$ and $\mathcal{F}_{\operatorname{reg},Q}[M;h]$ are $^*$-isomorphic, the $^*$-isomorphism being 
\begin{align}\label{Eqn: star-isomorphism between different algebras}
\alpha_P^Q\colon\mathcal{F}_{\operatorname{reg},Q}[M;h]\to\mathcal{F}_{\operatorname{reg},P}[M;h]\,,\quad
(\alpha_P^Q F)(\varphi)\vcentcolon=\bigg[\exp\big[\Upsilon_{P-Q}\big]F\bigg](\varphi)\,,
\end{align}
where 
\begin{align}\label{Eqn: Gamma contraction operator}
\bigg[\exp\big[\Upsilon_{P-Q}\big]F\bigg](\varphi)&=\sum\limits_{n=0}^\infty\frac{1}{2^n n!}\langle(P-Q)^{\otimes n},F^{(2n)}[\varphi]\rangle
\end{align}
and where $\Upsilon_{P-Q}$ is such that
\begin{equation*}
(\Upsilon_{P-Q}F)(\varphi)\vcentcolon=\frac{1}{2}\big\langle P-Q,F^{(2)}[\varphi]\big\rangle\,.
\end{equation*}
\end{proposition}

\noindent The second and last step consists of recollecting all $^*$-algebras of Proposition \ref{Prop: star-isomorphism between different algebras} in a single structure.

\begin{Definition}\label{Def: bundle over parametrices}
	We call $\mathsf{E}_{\mathrm{reg}}[M;h]$ and $\mathsf{E}[M;h]$ the bundles
	\begin{align}\label{Eqn: definition of the bundle of fiber algebras}
		\mathsf{E}_{\mathrm{reg}}[M;h]\vcentcolon=\bigcup_{P\in\operatorname{Par}[M;h]}\mathcal{F}_{\operatorname{reg},P}[M;h],\quad\mathrm{and}\quad \mathsf{E}[M;h]\vcentcolon=\bigcup_{P\in\operatorname{Par}[M;h]}\mathcal{P}[M;h]
	\end{align}
	both with base space $\operatorname{Par}[M;h]$ and projection maps
		$\pi_{\mathsf{E}_{\mathrm{reg}}[M;h]}(F_P)\vcentcolon=P$ (\textit{resp}. $\pi_{\mathsf{E}[M;h]}((P,F))=P$) for all $F_P\in\mathcal{F}_{\operatorname{reg},P}[M;h]$ (\textit{resp}. $F\in\mathcal{P}[M;h]$).
\end{Definition}

\noindent Observe that each fibre of $\mathsf{E}[M;h]$ can be considered an algebra only with respect to the pointwise product and not with respect to $\cdot_P$ since Equation \eqref{Eqn: Algebra product} is in general ill-defined over $\mathcal{P}_{\operatorname{loc}}[M;h]$ on account of the singular structure of the parametrices of $E$.

\begin{Definition}\label{Def: locally covariant algebra of regular observables}
We call $\Gamma_{\mathrm{eq}}(\mathsf{E}_{\mathrm{reg}}[M;h])$ the complex vector space of equivariant sections of $\mathsf{E}_{\mathrm{reg}}[M;h]$, \ie,
\begin{align}\label{Eqn: definition of locally covariant algebra of interest}
	\Gamma_{\mathrm{eq}}(\mathsf{E}_{\mathrm{reg}}[M;h])\vcentcolon=\big\lbrace F\in\Gamma(\mathsf{E}_{\mathrm{reg}}[M;h])\;|\;F(P)=\alpha_P^Q F(Q)\quad\forall P,Q\in\mathrm{Par}[M;h]\big\rbrace\,.
	\end{align}
and we denote with $\mathcal{A}_{\mathrm{reg}}[M;h]\equiv(\Gamma_{\mathrm{eq}}(\mathsf{E}_{\mathrm{reg}}[M;h]), \cdot, ^*)$ the unital $^*$-algebra whose product and involution are the following: for all $F,G\in\Gamma_{\mathrm{eq}}(\mathsf{E}_{\mathrm{reg}}[M;h])$
\begin{align}\label{Eqn: product and involution on the locally covariant algebra of interest}
	(F\cdot G)(P)\vcentcolon=F(P)\cdot_P G(P)\,,\qquad
	F^*(P)\vcentcolon=F(P)^*\,,
	\end{align}
where $\cdot_P$ is defined in \eqref{Eqn: Algebra product}, while $F[P]^*$ is the $\ast$-operation introduced in Proposition \ref{Prop: regular algebra}. Similarly we define $\Gamma_{\mathrm{eq}}(\mathsf{E}[M;h])$ as the complex vector space of equivariant sections of $\mathsf{E}[M;h]$.
\end{Definition}

\noindent An important consequence of this definition is the following.

\begin{corollary}\label{Corol: Gamma functor}
	Let $\Gamma_{\mathrm{eq},\operatorname{reg}}:\bkgg\to\alg$ and let $\Gamma_{\mathrm{eq}}:\bkgg\to\vect$ be such that
	\footnote{In the following we shall use the symbol $\Gamma_{\mathrm{eq}}$ to refer to either the functor $\Gamma_{\mathrm{eq}}\colon\bkgg\to\alg$ or to the set of equivariant sections over a suitable bundle. There will be no risk of confusion since in the latter case $\Gamma_{\mathrm{eq}}$ will be always followed by the symbol referring to the relevant bundle.}, for every $(M;h)\in\obj(\bkgg)$ and for every $\chi\in\arr(\bkgg)$ mapping from $(M;h)$ to $(M^\prime,h^\prime)$
	$$\Gamma_{\mathrm{eq},\operatorname{reg}}[M;h]=\Gamma_{\mathrm{eq}}(\mathsf{E}_{\mathrm{reg}}[M;h]),\quad\mathrm{and}\quad \Gamma_{\mathrm{eq}}[M;h]=\Gamma_{\mathrm{eq}}(\mathsf{E}[M;h]),$$
	while
	$\Gamma_{\mathrm{eq},\operatorname{reg}}[\chi]$ and $\Gamma_{\mathrm{eq}}[\chi]$ are such that, for all $F\in\Gamma_{\mathrm{eq}}(\mathsf{E}_{\mathrm{reg}}[M;h])$ and for all $G\in\Gamma_{\mathrm{eq}}(\mathsf{E}[M;h])$
	$$(\Gamma_{\mathrm{eq},\operatorname{reg}}[\chi](F))(P^\prime,\varphi^\prime)=F(\chi^*P^\prime,\chi^*\varphi^\prime),\quad\mathrm{and}\quad (\Gamma_{\mathrm{eq}}[\chi](G))(P^\prime,\varphi^\prime)=G(\chi^*P^\prime,\chi^*\varphi^\prime),$$
	where $P^\prime\in\operatorname{Par}[M^\prime;h^\prime]$ and $\varphi^\prime\in\mathcal{E}(M^\prime)$. Then both $\Gamma_{\mathrm{eq},\operatorname{reg}}:\bkgg\to\alg$ and $\Gamma_{\mathrm{eq}}:\bkgg\to\vect$ are covariant functors.
\end{corollary}

\begin{proof}
	It suffices to observe that per construction $\Gamma_{\mathrm{eq},\operatorname{reg}}[\operatorname{id}|_{(M;h)}]=\operatorname{id}|_{\Gamma_{\mathrm{eq},\operatorname{reg}}[M;h]}$ while, for any pair of morphisms $\chi:M\to M^\prime$ and $\chi^\prime:M^\prime\to\widetilde{M}$, the properties of the pull-back entail that $\Gamma_{\mathrm{eq},\operatorname{reg}}[\chi^\prime\circ\chi]=\Gamma_{\mathrm{eq},\operatorname{reg}}[\chi^\prime]\circ\Gamma_{\mathrm{eq},\operatorname{reg}}[\chi]$. The same statement holds true when considering $\Gamma_{\mathrm{eq}}$.
\end{proof}

\begin{remark}\label{Remark: scaling map}
	In order to investigate the scaling properties of $\Gamma_{\mathrm{eq},\mathrm{reg}}$, let $(M;h)\in\obj(\bkgg)$ and let, for any $\lambda>0$, $(M;h_\lambda)$ be as in Equation \eqref{Eqn: Scaling Transformation}.
	This will lead to a family of maps $\varsigma_{\lambda,\mu}$ satisfying the properties of Definition \ref{Def: Euclidean locally covariant theory} -- \textit{cf.} Theorem \ref{Theorem: locally covariant algebra of interest}.
	As discussed in Section \ref{Section: General Setting}, the scaling transformations are fully determined by the request of leaving the Lagrangian density invariant, namely $\mathcal{L}_\lambda=\mathcal{L}$.
		As a matter of fact, for all $\lambda>0$ the space of parametrices $\operatorname{Par}[M;h]$, $\operatorname{Par}[M;h_\lambda]$ are in bijection since $P_\lambda:=\lambda^{-2}P\in\operatorname{Par}[M;h_\lambda]$ for all $P\in\operatorname{Par}[M;h]$. This is a by product of $E\to \lambda^2E$ under scaling $g_{ab}\to\lambda^{-2}g_{ab}$.
		Moreover notice that also the local Hadamard representation -- \textit{cf.} Remark \ref{Remark: Hadamard representation} -- changes under scaling.
	Therefore we may define a linear map $\varsigma_\lambda\colon\Gamma_{\mathrm{eq}}[M;h_\lambda]\to\Gamma_{\mathrm{eq}}[M;h]$
	\begin{align}\label{Eqn: definition of scaling map}
		(\varsigma_\lambda F)(P,\varphi)\vcentcolon=F(\lambda^{-2}P,\lambda^{\frac{\mathrm{D}-2}{2}}\varphi)\,,
	\end{align}
	for all $F\in\Gamma_{\mathrm{eq}}[M;h]$.

	Notice that the engineering dimension of $\varphi$ has been inserted to match with the scaling dimension of the integral kernel with respect to the volume measure $\mu_g$ of the parametrix $P(x,y)$.
\end{remark}

\noindent We conclude the section proving a key result.

\begin{theorem}\label{Theorem: locally covariant algebra of interest}
Let $\mathcal{A}_{\mathrm{reg}}:\bkgg\to\alg$ be the covariant functor such that
\begin{itemize}
\item for any $(M;h)\in\obj(\bkgg)$, $\mathcal{A}_{\mathrm{reg}}[M;h]$ is the unital $^*$-algebra of Definition \ref{Def: locally covariant algebra of regular observables};
\item for any arrow $\chi:M\to\widetilde{M}$ and for any $F\in\Gamma_{\mathrm{eq}}(\mathsf{E}_{\mathrm{reg}}[M;h])$, $P\in\operatorname{Par}[\widetilde{M};\widetilde{h}]$ and $\varphi\in \mathcal{E}(\widetilde{M})$, $(\mathcal{A}_{\mathrm{reg}}(\chi)F)(P,\varphi)=F(\chi^*P,\chi^*\varphi)$.
\item for any $\lambda>0$ the scaling $\varsigma_\lambda$ is defined as in Remark \ref{Remark: scaling map}.
\end{itemize}
Then $\mathcal{A}_{\mathrm{reg}}:\bkgg\to\alg$ is a (scalar) Euclidean locally covariant field theory as per Definition \ref{Def: Euclidean locally covariant theory}.
\end{theorem}

\begin{proof}
First of all, notice that $\mathcal{A}_{\mathrm{reg}}$ is well defined since $\mathcal{A}_{\mathrm{reg}}[M;h]$ is per construction a unital, associative and commutative $^*$-algebra. In addition, for every arrow $\chi$ from $[M;h]$ to $[\widetilde{M};\widetilde{h}]$, it holds that $\chi^*P\in\operatorname{Par}[M;h]$ and thus $\mathcal{A}_{\mathrm{reg}}(\chi)F$ is still an equivariant section. In addition, for all $\chi,\widetilde{\chi}\in\arr\bkgg$, the properties of the pull-back entail that $(\chi\circ\widetilde{\chi})^*=\widetilde{\chi}^*\circ\chi^*$. It descends that $\mathcal{A}_{\mathrm{reg}}(\chi)\circ\mathcal{A}_{\mathrm{reg}}(\widetilde{\chi})=\mathcal{A}(\widetilde{\chi}\circ\chi)$ and $\mathcal{A}_{\mathrm{reg}}(\mathrm{Id}_M)=\mathrm{Id}_{\mathcal{A}_{\mathrm{reg}}[M;h]}$. This entails that $\mathcal{A}_{\mathrm{reg}}$ is a covariant functor.

Finally, a direct computation shows that the linear map $\varsigma_\lambda$ defined in Remark \ref{Remark: Scaling transformation} is a $\ast$-isomorphism between $\mathcal{A}_{\mathrm{reg}}[M;h]$ and $\mathcal{A}_{\mathrm{reg}}[M;h_\lambda]$. Notice the crucial r\^ole played by the engineer dimension of $\varphi$ in Equation \eqref{Eqn: definition of scaling map}.
This has been inserted to match with the scaling dimension of the integral kernel with respect to the volume measure $\mu_g$ of the parametrix $P(x,y)$.
\end{proof}
From now on, in this paper, with $\mathcal{A}_{\mathrm{reg}}$ we denote the Euclidean locally covariant field theory as per Definition \ref{Def: locally covariant algebra of regular observables} and Theorem \ref{Theorem: locally covariant algebra of interest}.

\begin{remark}\label{Remark: on-shell algebra}
	In the Lorentzian framework it is common to consider off-shell and on-shell algebras, the latter being obtained as a quotient between the first one and a suitable $\ast$-ideal encoding dynamically trivial observables. A similar procedure has no straightforward counterpart in the Riemannian setting due to the equations of motion being ruled by an elliptic operator. Nevertheless, we may identify a ``distinguished" algebra by considering the one constructed out of equivariant sections over a sub-bundle of $\mathsf{E}[M;h]$ whose base space is that of fundamental solutions $G\in\operatorname{SolFond}[M;h]$ of $E$.
	These are exact inverses of $E$ and, according to \cite{Li_Tam}, their existence is not guaranteed in general.
	Yet, assuming the space $\operatorname{SolFond}[M;h]\subset\operatorname{Par}[M;h]$ to be non-trivial, we may consider the algebra $\mathcal{A}_{\mathrm{reg},\mathrm{ex}}[M;h]\vcentcolon=\Gamma_{\mathrm{eq}}(\mathsf{E}_{\mathrm{ex}}[M;h])$, where $\mathsf{E}_{\mathrm{ex}}[M;h]\vcentcolon=\bigcup_{G\in\operatorname{SolFond}[M;h]}\mathcal{F}_{\operatorname{reg},G}[M;h]$.
	This algebra may be considered as an ``exact" version of $\mathcal{A}[M;h]$.
\end{remark}

\begin{remark}
	On account of Definition \ref{Def: locally covariant algebra of regular observables} and of Proposition	\ref{Prop: star-isomorphism between different algebras} it can proved that $\mathcal{A}_{\mathrm{reg}}[M;h]$ is in fact $*$-isomorphic to the algebra $\mathcal{P}_{\mathrm{reg}}[M;h]$ equipped with pointwise product -- the same holds true for the subsequent algebra $\mathcal{A}[M;h]$ introduced in Proposition \ref{Prop: local and covariant algebra of local polynomials}.
	It may appear more useful to deal directly with $\mathcal{P}_{\mathrm{reg}}[M;h]$, however, one should remember that the scaling map $\varsigma_\lambda$ introduced in Remark \ref{Remark: Scaling transformation} leads to a non-trivial scaling behaviour of elements $F\in\mathcal{A}[M;h]$ -- \textit{cf.} Definition \ref{Def: rescaled locally covariant observable} and subsequent discussion.
	This anomalous scaling is due to the scaling behaviour of the Hadamard parametrix $H$ introduced in remark \ref{Remark: Hadamard representation} and it is best seen when dealing with $\mathcal{A}_{\mathrm{reg}}[M;h]$.
\end{remark}

\section{Locally Covariant Observables and Quantum Fields}\label{Section: Locally Covariant Observables and Quantum Fields}
In this section we introduce the notion of locally covariant observables, as 
distinguished classes of natural transformation with value in $\alg$.
\begin{Definition}\label{Def: functor Gamma}
	For all $\ell\in\mathbb{N}$ we define $\Gamma_{\mathrm{c}}^\ell:\bkgg\to\vect$ as the functor such that, for any $(M;h)\in\obj(\bkgg)$ and $\chi\in\arr(\bkgg)$,
	\begin{align}\label{Eqn: definition of the functor Gamma order k}
	\Gamma_{\mathrm{c}}^\ell[M;h]\vcentcolon=
	\mathrm{S}C^\infty_{\mathrm{c}}(M^\ell)
	\,,\qquad
	\Gamma_{\mathrm{c}}^\ell[\chi]\vcentcolon=\chi_*\,.
	\end{align} 
	Similarly, we call $\Gamma_{\mathrm{c}}^\bullet:\bkgg\to\alg$ the covariant functor such that, for any $(M;h)\in\obj(\bkgg)$ and $\chi\in\arr(\bkgg)$,
	\begin{align}\label{Eqn: definition of the functor Gamma}
	\Gamma_{\mathrm{c}}^{\bullet}[M;h]\vcentcolon=\bigoplus_{\ell=0}^\infty
	\mathrm{S}C^\infty_{\mathrm{c}}(M^\ell)
	\,,\qquad
	\Gamma_{\mathrm{c}}^{\bullet}[\chi]\vcentcolon=\chi_*\,,
	\end{align}                         
	where $\mathrm{S}C^\infty_{\mathrm{c}}(M^\ell)$ denotes the collection of all smooth, compactly supported, complex valued functions symmetric in their argument, with the convention that $\mathrm{S}C^\infty_{\mathrm{c}}(M^0)\vcentcolon=\mathbb{C}$, while $\mathrm{S}C^\infty_{\mathrm{c}}(M)\equiv C^\infty_{\mathrm{c}}(M)$.
	The product structure on $\Gamma_{\mathrm{c}}^\bullet[M;h]$ is induced by the symmetric tensor product, namely if $f_\ell\in \mathrm{S}C^\infty_{\mathrm{c}}(M^\ell)$ and $f_n\in \mathrm{S}C^\infty_{\mathrm{c}}(M^n)$ then $f\cdot_{\Gamma_{\mathrm{c}}^\bullet[M;h]}g:=f\otimes g\in \mathrm{S}C^\infty_{\mathrm{c}}(M^{\ell+n})$.
	We shall denote with $\Gamma^{\bullet}\colon\bkgg\to\alg$ and $\Gamma^\ell\colon\bkgg\to\vect$ the analogous contravariant functors defined by dropping the subscript $_{\mathrm{c}}$ -- notice that $\Gamma^{\bullet}(\chi)\vcentcolon=\chi^*$.
\end{Definition}
\noindent In the spirit of \cite{Brunetti:2001dx}, we introduce \emph{locally covariant observables} as follows
\begin{Definition}\label{Def: locally covariant observable}
	Let $\mathcal{A}_{\mathrm{reg}}:\bkgg\to\alg$ and $\Gamma^{\bullet}_{\mathrm{c}}:\bkgg\to\alg$ be the functors respectively as per Definition \ref{Def: locally covariant algebra of regular observables} and \ref{Def: functor Gamma}.
	We define a \emph{locally covariant observable} as a natural transformation $\mathcal{O}:\Gamma_{\mathrm{c}}^{\bullet}\to\mathcal{A}_{\mathrm{reg}}$,  that is, $\forall(M;h)\in\obj(\bkgg)$, $\mathcal{O}[M;h]:\Gamma_{\mathrm{c}}^{\bullet}[M;h]\to\mathcal{A}_{\mathrm{reg}}[M;h]$ is an arrow in $\alg$ and it holds that, for any $\chi\in\arr(\bkgg)$ mapping $(M;h)$ to $(\widetilde{M};\widetilde{h})$, 
	\begin{align}\label{eq: comp_nat_trans}
		\mathcal{O}[\widetilde{M};\widetilde{h}]\circ\Gamma_{\mathrm{c}}^{\bullet}[\chi]=\mathcal{A}_{\mathrm{reg}}[\chi]\circ\mathcal{O}[M;h]\,.
	\end{align}
\end{Definition}

\begin{remark}\label{Remark: generalized locally covariant observable}
	The previous definition -- see also Definition \ref{Def: rescaled locally covariant observable} -- generalizes to any Euclidean locally covariant theory $\mathcal{A}\colon\bkgg\to\alg$ as per Definition \ref{Def: Euclidean locally covariant theory}, identifying the most general notion of locally covariant observable as a natural transformation $\mathcal{O}\colon\Gamma_{\operatorname{c}}^\bullet\to\mathcal{A}$.
	Notice that in \cite{Brunetti-Duetsch-Fredenhagen-09,Khavkine-Melati-Moretti-17,Hollands-Wald-01,Hollands-Wald-02,Hollands-Wald-03,Hollands-Wald-03,Hollands-Wald-05} local and covariant observables are defined as natural transformations $\mathcal{O}\colon\Gamma_{c}^{1}\to\mathcal{A}$.
	From this point of view, Definition \ref{Def: locally covariant observable} identifies a multilocal covariant observable, by incorporating also the structure of natural algebra homomorphism.
	This is useful for keeping track of the algebraic properties carried by local and covariant observables -- \textit{cf.} \ref{Prop: local and covariant algebra of local polynomials}.
\end{remark}

\begin{remark}
	Notice that, given any $(M;h)\in\obj(\bkgg)$, $\mathcal{O}[M;h]$ can be seen as an algebra-valued distribution, \ie, for any $P\in\operatorname{Par}[M;h]$ and for any $\varphi\in \mathcal{E}(M)$
	\begin{align*}
	\mathcal{O}[M;h](\cdot, P, \varphi):C_\mathrm{c}^\infty(M^\ell)\ni f\mapsto\mathcal{O}[M;h](f, P, \varphi)\in\C,
	\end{align*}
	is required to be a distribution, namely $\mathcal{O}[M;h](\cdot, P, \varphi)\in\mathcal{D}^\prime(M^\ell)$.
\end{remark}

\begin{remark}
	Notice that, being $\mathcal{O}[M;h]\in\mathsf{Arr}(\alg)$, for all $\ell\in\N$ and for all $f_1,\dots, f_\ell\in \mathcal{D}(M)$, it holds 
	\begin{align*}
			\mathcal{O}[M;h](f_1\otimes\dots\otimes f_\ell)=\mathcal{O}[M;h](f_1)\cdots\mathcal{O}[M;h](f_\ell)\,,
		\end{align*}
	where the product on the right-hand side is in $\mathcal{A}_{\mathrm{reg}}[M;h]$.
	This observation and the assumed continuity of $\mathcal{O}[M;h](\cdot,P,\varphi)$ imply
	that a locally covariant observable as per Definition \ref{Def: locally covariant observable} is completely determined on $\Gamma_\mathrm{c}^{\bullet}$ once its action on the degree $\ell=0,1$ and the product $\cdot$ of the algebra $\mathcal{A}_\mathrm{reg}$ are known.
	Notice furthermore that, since we are dealing with \emph{regular} functionals, the products involved in the previous equation are all well-defined.
	This will not be the case when dealing with \emph{local} functionals.
\end{remark}

\begin{Example}\label{Example: quantum fields}
Let $(M;h)\in\obj(\bkgg)$ and let $\Phi[M;h]:\Gamma_\mathrm{c}^{\bullet}[M;h]\to\mathcal{A}_\mathrm{reg}[M;h]$ be such that, given $f\in \mathcal{D}(M)$,  for any $P\in\operatorname{Par}[M;h]$ and for any $\varphi\in \mathcal{E}(M)$, $\Phi[M;h](f)$ is the linear functional 
\begin{align}
		\Phi[M;h](f,P,\varphi)\vcentcolon=\int_M\,\mu_g\,
		\textrm{$f\varphi$}\,,
	\end{align}
extended according to the equation in the preceding remark.
Consider now a morphism $\chi\in\arr(\bkgg)$ mapping $(M;h)$ to $(\widetilde{M};\widetilde{h})$. In order to prove that $\Phi$ is a locally covariant observable, we need to show that $\Phi[\widetilde{M};\widetilde{h}]\circ\Gamma_{\mathrm{c}}^{\bullet}[\chi]=\mathcal{A}_{\mathrm{reg}}[\chi]\circ\Phi[M;h]$. This follows from the definition since, for every $f\in C^\infty_\mathrm{c}(M)$ and for all $\widetilde{P}\in\mathrm{Par}(\widetilde{M};\widetilde{h})$, $\widetilde{\varphi}\in C^\infty(\widetilde{M})$,
	\begin{align*}
		\bigg[\mathcal{A}_{\operatorname{reg}}(\chi)\Phi[M;h](f)\bigg](\widetilde{P},\widetilde{\varphi})&=
		\Phi[M;h](f,\chi^*\widetilde{P},\chi^*\widetilde{\varphi})=
		\int_{\widetilde{M}}\,\mu_{\widetilde{g}}\,
		\textrm{$\widetilde{\varphi}\chi_*f$}
		\\&=
		\Phi[\widetilde{M},\widetilde{h}](\Gamma^\bullet_{\mathrm{c}}(\chi)f,\widetilde{P},\widetilde{\varphi})\,.
	\end{align*}
We conclude that $\Phi$ is a locally covariant observable, to which we will refer to as \emph{locally covariant quantum field}.
\end{Example}
Since we will be interested in the scaling behavior of locally covariant observables, we introduce the notion of \emph{rescaled locally covariant observable}.
\begin{Definition}\label{Def: rescaled locally covariant observable}
Let $\mathcal{O}:\Gamma_{\mathrm{c}}^{\bullet}\to\mathcal{A}_{\mathrm{reg}}$ be a locally covariant observable. For any $(M;h)\in\obj(\bkgg)$ we call $S_\lambda\mathcal{O}$ the \emph{rescaled locally covariant observable} at scale $\lambda>0$, defined by
\begin{align}\label{Eqn: rescaled locally covariant observable}
	\big(S_\lambda\mathcal{O}\big)[M;h](f)\vcentcolon=
	\varsigma_\lambda\big[\mathcal{O}[M;h_\lambda](\lambda^{\mathrm{D}}f)\big]\,,
\end{align}
for all $f\in C^\infty_\mathrm{c}(M)$ and where $(M;h_\lambda)$ is defined as per Equation \eqref{Eqn: Scaling Transformation}, while $\varsigma_\lambda$ has been defined in remark \ref{Remark: scaling map}.
Furthermore, on the one hand, we say that $\mathcal{O}$ has engineering dimension $\mathrm{d}_\mathcal{O}\in\R$ if, for any  $(M;h)\in\obj(\bkgg)$ and for any $f\in C^\infty_\mathrm{c}(M)$, it satisfies 
\begin{align}\label{Eqn: defining equation for engineer dimension}
		(S_\lambda\mathcal{O})[M;h](f)=\lambda^{\mathrm{d}_{\mathcal{O}}}\mathcal{O}[M;h](f)\,.
	\end{align}
On the other hand, we say that $\mathcal{O}$ scales \emph{almost homogeneously} with dimension $\kappa\in\R$ and order $m\in\N$ if 
	\begin{align}\label{Eq: scaling alm_hom}
		S_\lambda\mathcal{O}[M;h](f)=\lambda^{\kappa}\mathcal{O}[M;h](f)+
		\lambda^{\kappa}\sum_{j\leq m}\log(\lambda)^j\mathcal{O}_j[M;h](f)\,,
	\end{align}
for any $(M;h)\in\obj(\bkgg)$, $f\in C^\infty_\mathrm{c}(M)$ and where $\mathcal{O}_j$, for all $j\in\lbrace1,\dots,m\rbrace$, are locally covariant observables which scale almost homogeneously with degree $\kappa$ and order $m-j$. The definition is inductive in the order $m$ and a locally covariant observable which scales almost homogeneously with dimension $\kappa$ and order $m=0$ scales homogeneously with dimension $\kappa$.
\end{Definition}

\begin{remark}
	Notice that the scaling of the test-function $f$ is chosen in such a way that the density $f\mu_g$ is scale invariant.
\end{remark}

\section{Wick Ordered Powers of Quantum Fields}\label{Section: Wick Ordered Powers of Quantum Fields}
Our next goal is to bypass the limitation of $\Gamma_{\mathrm{eq}}(\mathsf{E}[M;h])$ not being an algebra, since \eqref{Eqn: Algebra product} is ill-defined on local polynomial functionals. To overcome this hurdle we introduce Wick monomials, which play the r\^ole of a non-linear generalization of local and covariant observables as per Example \ref{Example: quantum fields}, proving their existence and classifying the ambiguities in their definition.

In this endeavor we adapt to the Riemannian case the approach taken by \cite{Khavkine:2014zsa}, which, in turn, is a generalization of the seminal papers \cite{Hollands-Wald-01, Hollands-Wald-02} in which the condition of the underlying manifold being analytic is dropped.
This is achieved applying the Peetre-Slov\'ak theorem, which is recalled succinctly in Appendix \ref{Appendix: Peetre-Slovak}. 

We divide the analysis in two steps, focusing first on \textit{Wick powers} and subsequently on \textit{Wick monomials}. The former identify, roughly speaking, an integer power of a single fundamental field $\Phi$ -- \textit{cf.} Example \ref{Example: quantum fields}. The latter codify the product of finitely many Wick powers, leading to the algebraic structure which we refer to as $\mathrm{E}$-product -- \textit{cf.} Proposition \ref{Prop: local and covariant algebra of local polynomials}.

In this section we discuss in detail the first step, following a procedure similar to the one employed in \cite{CDDR18} in the study of non linear sigma models. Observe that, in the following, $\Phi$ will always denote the locally covariant observable defined in Example \ref{Example: quantum fields}.
\begin{Definition}\label{Def: Wick polynomials}
	Let $\Gamma_{\operatorname{eq}}$ and $\Gamma^1_\mathrm{c}$ be the functors defined respectively in Corollary \ref{Corol: Gamma functor} and Definition \ref{Def: functor Gamma}. We call family of \emph{Wick powers}, associated to $\Phi$, a collection of natural transformations $\lbrace\Phi^k\rbrace_{k\in\N\cup\{0\}}$ with $\Phi^k\vcentcolon\Gamma^1_\mathrm{c}\to\Gamma_{\operatorname{eq}}$ such that the following conditions are met:
	\begin{enumerate}
		\item $\forall\,k\in\N\cup\lbrace0\rbrace$, $\Phi^k$
		is a natural transformation -- here we are regarding $\Gamma_{\mathrm{eq}}$ as a $\vect$-valued functor -- which scales almost homogeneously with dimension $k\mathrm{D}_\varphi=k\big(\frac{\mathrm{D}-2}{2}\big)$ and order at most $k$, where $\mathrm{D}=\dim(M)$ and where we have considered the natural generalization of Definition \ref{Def: rescaled locally covariant observable} to this setting;
		
		\item if $k=1$, $\Phi^1\equiv\Phi$ while, if $k=0$, $\Phi^0=\mathrm{Id}_{\Gamma_{\operatorname{eq}}}$, where, for any $(M;h)\in\obj(\bkgg)$, $\mathrm{Id}_{\Gamma_{\operatorname{eq}}[M;h]}$ denotes the identity functional such that for any $f_1\in \Gamma^{1}_{\mathrm{c}}[M;h]$, $\mathrm{Id}_{\Gamma_{\operatorname{eq}}[M;h]}(f_1,P;\varphi)\vcentcolon=\int_M \mu_g\,f_1$;
		
		\item $\forall k\in\N\cup\lbrace0\rbrace$, $(M;h)\in\obj(\bkgg)$, $f_1\in \Gamma^{1}_{\mathrm{c}}[M;h]$, $P\in\operatorname{Par}[M;h]$ and $\varphi_1,\varphi_2\in \mathcal{E}(M)$,
		\begin{align}\label{Eqn: derivative of Wick powers}
			\langle\Phi^k[M;h](f_1,P)^{(1)}[\varphi_1],\varphi_2\rangle=k\,\Phi^{k-1}[M;h](\varphi_2 f_1,P,\varphi_1),
		\end{align}
		where the superscript $^{(1)}$ on the left hand side denotes the first order functional derivative;
		
		\item let $n\in\N$ and let $(M;h_s)\in\obj(\bkgg)$, with $\lbrace h_s\rbrace_{s\in\R^n}$ a smooth and compactly supported $n$-dimensional family of variations of $h$ -- see Definition \ref{Def: smooth compactly supported d-dimensional family of variations} and let $L(M)$ denote the trivial line bundle $M\times\mathbb{C}$. For any smooth family $\lbrace P_s\rbrace_{s\in\R^n}$ with $P_s\in\mathrm{Par}(M;h_s)$ and for any $s\in\R^n$, let $\mathcal{U}_{k}\in\mathcal{D}^\prime(\pi^*_nL(M))$ be the distribution on the pull-back bundle $\pi^*_nL(M)$ -- here $\pi_n:\R^n\times M\to M$ denotes the canonical projection -- such that, for any $f_1\in \Gamma^{1}_{\mathrm{c}}[M;h]$, 
		\begin{align}\label{Eqn: distribution for microlocal spectrum condition}
			\mathcal{U}_{k}(\chi,f_1)\vcentcolon=\int_{\R^n}\mathrm{d}s\;\Phi^k[M;h_s](f_1,P_s,0)\chi(s)\,,\quad
			\forall\chi\in \mathcal{D}(\R^n).
		\end{align}
		It holds that $\forall k\in\N$,
		\begin{align}\label{Eqn: microlocal spectrum condition}
		\mathrm{WF}(\mathcal{U}_{k})=\emptyset,
		\end{align}
		with $\mathrm{WF}(\mathcal{U}_{k})$ denoting the wave front set of the distribution $\mathcal{U}_{k}$ \cite[Def. 8.1.2]{Hormander-83};
		
		\item for any $k\in\N$, $(M;h)\in\obj(\bkgg)$ and $f\in \Gamma_{\mathrm{c}}^{1}[M;h]$,
		\begin{align}\label{Eqn: hermiticity of Wick powers}
		\Phi^k[M;h](f)^*=\Phi^k[M;h](\overline{f}).
		\end{align}
	\end{enumerate}
\end{Definition}

\begin{remark}\label{Remark: dependence on s of the parametrix}
	Notice that the family $\lbrace P_s\rbrace_{s\in\R^n}$ is associated to a unique $P\in\mathcal{D}'(\pi_n^*L(M\times M))$ and the existence of such family is a consequence both of the smooth dependence on $s\in\R^n$ of the elliptic operator $E_s$, associated with the background geometries $(M;h_s)$ and of the construction of $P_s$ as a pseudodifferential operator \cite[Thm. 5.1]{Shubin}.
\end{remark}

\subsection{Existence of Wick Ordered Powers of Quantum Fields}\label{Section: Existence of Wick powers}
In this short section we exhibit an explicit construction of Wick powers abiding by the axioms of Definition \ref{Def: Wick polynomials}. Let $k\in\N$, $(M;h)\in\obj(\bkgg)$ and let $f\in \Gamma_{\mathrm{c}}^{1}[M;h]$ while $\varphi\in\mathcal{E}(M)$. Starting from the polynomial local functional
\begin{align}\label{Eqn: k-th power polynomial functional}
		\phi^k[M;h](f,\varphi)\vcentcolon=\int_M\,\mu_g\varphi^k f\,,
\end{align}
we construct an equivariant counterpart with respect to the choice of a parametrix $P\in\operatorname{Par}[M;h]$. Recalling that $\operatorname{Par}[M;h]$ is an affine space modeled over $C^\infty(M\times M)$, we set
\begin{align}\label{Eqn: definition of Wick powers}
	\wick{\Phi^k}_H[M;h](f,P,\varphi)\vcentcolon=\bigg[\exp\big[\Upsilon_{W_P}]\phi^k[M;h](f)\bigg](\varphi)\,,
\end{align}
where $\Upsilon_{W_P}$ is defined as in \eqref{Eqn: Gamma contraction operator}, while
$W_P$ has been introduced in Remark \ref{Remark: Hadamard representation}.

Observe that \eqref{Eqn: definition of Wick powers} is well-defined on account of the support properties of the functional derivatives of local functionals -- \textit{cf.} Definition \ref{Def: functionals} -- which ensure that only the coinciding point limit $[W_P](x)=W_P(x,x)$ is needed in the evaluation of $\exp[\Upsilon_{W_P}]F$.

This prescription fulfills all requirements of Definition \ref{Def: Wick polynomials}.
The proof is very similar to the one outlined in \cite{Hollands-Wald-01}. For this reason here we shall give only a brief sketch.
As a matter of fact $\wick{\Phi^k}_H$ is a locally covariant observable which scales almost homogeneously with dimension $k\big(\frac{\mathrm{D}-2}{2}\big)$ and order at most $k$ as a consequence of the engineering dimension of $\varphi$, see Equation \eqref{Eqn: Engineering Dimension}, and of the presence in even dimensions of the logarithmic term in the Hadamard expansion of the parametrix, {\it cf.} Equation \eqref{Eqn: Hadamard representation}.
The second, the third and the fifth condition of Definition \ref{Def: Wick polynomials} hold true per construction, while the fourth one is a by product of the identities
\begin{align*}
	\wick{\Phi^{2k+1}}_H[M;h](f,P,0)=0\,,\qquad	\wick{\Phi^{2k}}_H[M;h](f,P,0)=\int_M\,\mu_g [W_P]^k f\,,\qquad\forall\,k\in\N\,.
\end{align*}
Since for any smooth family of parametrices $\{P_s\}_{s\in\mathbb{R}}$ it holds that $[W_{P_s}](x)$ is also a smooth in $(s,x)\in\mathbb{R}\times M$, the previous identity entails that the associated distribution $\mathcal{U}_{k}$ has empty wave front set.

\subsection{Non-Uniqueness of Wick Ordered Powers of Scalar Quantum Fields}\label{Section: Uniqueness of Wick Ordered Powers of Quantum Fields}
In this section we investigate whether there exist ambiguities in the prescription of Wick polynomials outlined in Section \ref{Section: Existence of Wick powers}. In the Lorentzian setting, this is an overkilled topic \cite{Hollands-Wald-01, Khavkine-Melati-Moretti-17,Khavkine:2014zsa} and, for our purposes, we adopt the same strategy of \cite{Khavkine:2014zsa}. We split the main result of this section in two theorems, namely, in the first, we prove a general formula \eqref{Eqn: ambiguities for Wick powers} relating two arbitrary prescriptions for Wick powers by means of a family of suitable coefficients, whose structural properties are proven in the second theorem.

\begin{theorem}\label{Theorem: Unicity of Wick powers I}
	Let $\{\widehat{\Phi}^k\}_k$ and $\{\Phi^k\}_k$ be two families of Wick powers associated to $\Phi$ as per Definition \ref{Def: Wick polynomials}.
	Then for $(M,h)\in\mathsf{Obj}(\bkgg)$ and for all $k\geq 2$ there exists a family $\{c_j[M;h]\}_{2\leq j\leq k}$ of smooth functions $c_j[M;h]\in\Gamma^{1}[M;h]$ such that for all $f\in\Gamma_{\mathrm{c}}^{1}[M;h]$
	\begin{align}\label{Eqn: ambiguities for Wick powers}
		\widehat{\Phi}^k[M;h](f)=\Phi^k[M;h](f)+\sum_{j=0}^{k-2}\binom{k}{j}\Phi^j[M;h](c_{k-j}[M;h]\lrcorner f)\,,
	\end{align}
	where $c_{k-j}[M;h]\lrcorner f\in\Gamma_{\mathrm{c}}^{1}[M;h]$ denotes the pointwise multiplication\footnote{We stick with this notation in view of section \ref{Section: The Case of a Scalar Field with Derivatives}.} between $c_{k-j}[M;h]$ and $f$.
	The tensor $c_{k-j}[M;h]$ is weakly regular as per definition \ref{Def: weak regularity condition}.
	Moreover by defining
	\begin{align}
		C_{k-j}[M;h](f)\vcentcolon=\int_M\mu_g\;c_{k-j}[M,h]\lrcorner f\operatorname{Id}_{\Gamma_{\mathrm{eq}}[M;h]}\,,
	\end{align}
	we have that $C_{k-j}$ is a local and covariant observable as per definition \ref{Def: locally covariant observable} which scales almost homogeneously with dimension $(k-j)\mathrm{D}_\varphi=(k-j)\big(\frac{\mathrm{D}-2}{2}\big)$ with respect to the transformation $h=(g,A,c)\mapsto h_\lambda=(\lambda^{-2}g,A,\lambda^2c)$ in Equation \eqref{Eqn: Scaling Transformation}.
\end{theorem}

\begin{proof}
	The proof goes per induction with respect to $k\in\N$.
	First of all notice that, since, by definition, $\widehat{\Phi}^1=\Phi=\Phi^1$, the thesis holds true if $k=1$.
	We can now prove the inductive step, \ie, we assume that the thesis holds true up to order $k-1$, namely there exist weakly regular tensors $\{c_j[M;h]\}_{2\leq j\leq k-1}$ such that 
	\begin{align}\label{Eqn: inductive hypothesis}
		\widehat{\Phi}^{k-1}[M;h](f)=\Phi^{k-1}[M;h](f)+\sum_{j=0}^{k-3}\binom{k-1}{j}\Phi^j[M;h](c_{k-1-j}[M;h]\lrcorner f)\,,\\
	\end{align}
	for all $f\in\Gamma_{\mathrm{c}}^{1}[M;h]$.
	Let us introduce
	\begin{align*}
	C_k[M;h](f)\vcentcolon=
	\widehat{\Phi}^{k}[M;h](f)-\Phi^{k}[M;h](f)-\sum_{j=1}^{k-2}\binom{k}{j}\Phi^j[M;h](c_{k-j}[M;h]\lrcorner f)\,.
	\end{align*}
	First of all, notice that $C_k$ is a locally covariant observable which satisfies all the axioms of Wick powers as per Definition \ref{Def: Wick polynomials}, since it is constructed as a linear combination of objects enjoying such properties. In addition, $C_k$ is a $\C$-number field, namely it is proportional to the identity functional. This is a consequence of axiom $(3)$ of Definition \ref{Def: Wick polynomials} and of the inductive hypothesis \eqref{Eqn: inductive hypothesis} which entail
	\begin{align*}
	\langle C_k[M;h]^{(1)}(f,P)[\varphi_1],\varphi_2\rangle=0,\quad\forall\,P\in\operatorname{Par}[M;h],\, f\in \Gamma_{\mathrm{c}}^{1}[M;h],\,\mathrm{and}\, \varphi_1,\varphi_2\in \mathcal{E}(M).
	\end{align*}
	As a consequence, we conclude that $C_k$, seen as an element of $\Gamma_{\operatorname{eq}}[M;h]$, is independent from $P$ and $\varphi$. Therefore 
	\begin{align*}
		C_k[M;h](f)=\int_M \mu_g c_k[M;h]\lrcorner f\;\mathrm{Id}_{\Gamma_{\operatorname{eq}}[M;h]}\,,
	\end{align*}
	where $c_k[M;h]\in\Gamma^1[M;h]$ because of axiom $(4)$ of Definition \ref{Def: Wick polynomials}, which entails
	\begin{align*}
		\mathrm{WF}(c_k[M;h])=\emptyset\,.
	\end{align*}
	To discuss the regularity properties of $c_k[M;h]$, consider an $m$-dimensional family of smooth compactly supported variations $h_s$ of $h$ -- \textit{cf.} definition \ref{Def: smooth compactly supported d-dimensional family of variations}.
	Following the same procedure as above, it descends that $c_k[M;h_s](x)$ is jointly smooth in $(s,x)$.
	Hence $(M;h)\to c_k[M;h]$ is weakly regular. 
\end{proof}

\begin{remark}
	On account of the properties of $C_k[M;h]$ it follows that $c_k[M;h](x)$ depends only on the germ of $h$ at $x$.
	To this end, we focus on the behaviour of $c_k[M;h]$ under pull-back with respect to $\chi\in\arr(\bkgg)$.
	In particular, for any $U_x$, relatively compact open neighbourhood centered at $x$, the inclusion map $\chi_U:U_x\hookrightarrow M$ identifies a morphism of $\bkgg$.
	In addition the locality property of $C_k$ implies $\chi^*c_k[\chi(M);h]=c_k[M;\chi^*h]$ for any  $\chi:U\to M$.
	Hence, for any but fixed $x\in M$, the sought conclusion descends considering a sequence of relatively compact open neighbourhoods centered at $x$, $\{U_{x,i}\}_{i\in\mathbb{N}}$ such that $U_{x,i+1}\subset U_{x,i}$ and $\lim\limits_{i\to\infty}U_{x,i}=\{x\}$.
	The following theorem provides more information on the coefficients $c_k$ \cite{Khavkine-Melati-Moretti-17,Khavkine:2014zsa}.
\end{remark}

\begin{theorem}\label{Theorem: Unicity of Wick powers II}
	Let $\{\widehat{\Phi}^k\}_k$ and $\{\Phi^k\}_k$ be two families of Wick powers associated to $\Phi$ as per Definition \ref{Def: Wick polynomials}.
	With reference to Equation \eqref{Eqn: ambiguities for Wick powers}, it holds that, for any $(M;h)\in\obj(\bkgg)$, the coefficients $c_\ell[M;h]$ are differential operators taking the form
	\begin{align}\label{Eqn: coefficients c_k}
		c_\ell[M;h](x)=c_\ell[g^{ab}(x), \epsilon^{a_1\dots a_n}(x), R_{abcd}(x), \dots, \nabla_{e_1}\dots\nabla_{e_n}R_{abcd}(x),\dots\\
		\dots A(x), \nabla_{e_1}\dots\nabla_{e_n}A_a(x), c(x), \nabla_{e_1}\dots\nabla_{e_n}c(x)]\,,\notag
	\end{align}
	where $\epsilon^{a_1\dots a_n}(x)$ and $R_{abcd}(x)$ denote respectively, the Levi-Civita and the Riemann curvature tensors built out of $g$ at $x\in M$.
	Furthermore, each $c_\ell$ is a polynomial, scalar function, covariantly constructed from of its arguments.
	Finally, every $c_\ell$ scales homogeneously with dimension $\ell\big(\frac{\mathrm{D}-2}{2}\big)$ under the transformation $h=(g,A,c)\mapsto h_\lambda=(\lambda^{-2}g,A, \lambda^2c)$.
\end{theorem}

\noindent We omit the proof of this last statement, which relies in turn on the Peetre-Slov\'ak theorem, since, taking into account Theorem \ref{Theorem: Unicity of Wick powers I}, it is, \emph{mutatis mutandis}, identical to \cite[Theorem 3.1]{Khavkine:2014zsa}.

\section{E-Product of Wick Ordered Powers of Quantum Fields}\label{Section: E-Product of Wick Ordered Powers of Quantum Fields}
In this section we discuss how to endow the Wick powers with a product structure, which we will refer to as {\em E-product}, which can be read as a local and covariant extension of the bilinear map $\cdot_P$ as in Equation \eqref{Eqn: Algebra product}.
As a byproduct $\Gamma_{\operatorname{eq}}$ acquires the structure of an algebra, hence identifying a full fledged Euclidean locally covariant field theory as per Definition \ref{Def: Euclidean locally covariant theory}.
Recall that the counterpart of this analysis in a Lorentzian framework leads to the introduction of the renowned time ordered product (T-product) \cite{Hollands-Wald-02} which is at the heart of perturbation theory. 

In order to define the E-product we first introduce the concept of Wick monomials which can be seen as a natural generalization of the one of Wick power.
In particular Wick monomials provide an extension of the product defined in equation \eqref{Eqn: product and involution on the locally covariant algebra of interest} to a chosen family of Wick power -- \textit{cf.} equation \eqref{Eqn: causal factorization}.
Once a specific choice of Wick monomials has been made, Proposition \ref{Prop: local and covariant algebra of local polynomials} ensures that the algebra $\mathcal{A}_{\mathrm{reg}}[M;h]$ can be extended to a larger one, $\mathcal{A}[M;h]$, generated by the chosen family of Wick powers.
The product over $\mathcal{A}[M;h]$ is induced by Wick monomials, and it is called E-product.

In what follows, $\underline{\mathrm{k}}=(k_n)_{n}$ shall denote a finite sequence of non-negative integers, while with $\ell(\underline{\mathrm{k}})\in\mathbb{N}$ we indicate the number of elements of any such finite sequence $\underline{\mathrm{k}}=(k_1,\ldots,k_{\ell(\underline{\mathrm{k}})})$.

\begin{Definition}\label{Def: E-product}
	Let $\lbrace\Phi^k\rbrace_{k\in\N}$ be a family of Wick powers associated with the quantum field $\Phi$, as per Definition \ref{Def: Wick polynomials} and let $\text{\ul{k}}=(k_n)_{n}$ denote a finite sequence of $\ell(\underline{\mathrm{k}})$ many non-negative natural numbers.
	We call family of Wick monomials $\{\Phi^{\text{\ul{k}}}\}_{\text{\ul{k}}}$ associated to $\{\Phi^k\}_k$ a family of natural transformations
	$\Phi^{\text{\ul{k}}}\colon\Gamma_{\mathrm{c}}^{\ell(\underline{\textrm{k}})}\to\Gamma_{\mathrm{eq}}$ with the following properties:
	\begin{enumerate}
		\item
		for every finite sequence $\text{\ul{k}}$,
		$\Phi^{\text{\ul{k}}}\colon\Gamma_{\operatorname{c}}^{\ell(\underline{\textrm{k}})}\to\Gamma_{\operatorname{eq}}$
		scales almost homogeneously with dimension
		$\sum_{i=1}^{\ell(\underline{\textrm{k}})} k_i\big(\frac{\mathrm{D}-2}{2}\big)$
		and order at most
		$\sum_{i=1}^{\ell(\underline{\textrm{k}})}k_i$,
		with $\mathrm{D}\vcentcolon=\mathrm{dim}(M)$, where $\Gamma_{\mathrm{eq}}$ and
		$\Gamma_{\mathrm{c}}^{\ell(\underline{\textrm{k}})}$
		are the functors introduced respectively in Corollary \ref{Corol: Gamma functor} and in Definition \ref{Def: functor Gamma}; 
		
		\item
		if $\ell(\underline{\textrm{k}})=1$ then $\Phi^{\underline{\textrm{k}}}=\Phi^{k_1}$;
		
		\item let $\text{\ul{k}}$ be an arbitrary sequence,
		and $(M;h)\in\obj(\bkgg)$ and let
		$f_1,\dots,f_{\ell(\underline{\textrm{k}})}\in \Gamma_{\mathrm{c}}^{1}[M;h]$.
		Let $I\subsetneq\lbrace1,\dots,\ell(\underline{\mathrm{k}})\rbrace$ be a proper subset and denote with $I^\mathrm{c}$ the complement of $I$ with respect to $\lbrace1,\dots,\ell(\underline{\mathrm{k}})\rbrace$. If
		\begin{align}\label{Eqn: condition on supports for causal factorization}
			\bigcup_{i\in I}\mathrm{supp}(f_i)\cap\bigcup_{j\in I^\mathrm{c}}\mathrm{supp}(f_j)=\emptyset\,,
		\end{align}
		then
		\begin{align}\label{Eqn: causal factorization}
			\Phi^{\text{\ul{k}}}[M;h](f_1\otimes\dots\otimes f_\ell)=
			\Phi^{\text{\ul{k}}_I}[M;h]\bigg(\bigotimes_{i\in I} f_i\bigg)\cdot\Phi^{\text{\ul{k}}_{I^\mathrm{c}}}[M;h]\bigg(\bigotimes_{j\in I^\mathrm{c}} f_j\bigg)\,,
		\end{align}
		where $\text{\ul{k}}_I$ and $\text{\ul{k}}_{I^\mathrm{c}}$ denote, respectively, the finite sequences associated with the indices of $I$ and $I^\mathrm{c}$ and where $\cdot$ denotes the equivariant product as per Equation \eqref{Eqn: product and involution on the locally covariant algebra of interest}.
		
		\item
		for all sequences $\text{\ul{k}}$ and for any $(M;h)\in\obj(\bkgg)$, $f\in \Gamma_{\mathrm{c}}^{\ell(\underline{\mathrm{k}})}[M;h]$, $P\in\operatorname{Par}[M;h]$ and $\varphi,\psi\in\mathcal{E}(M)$,
		\begin{align}\label{Eqn: derivative of T-product}
			\langle\Phi^{\text{\ul{k}}}[M;h]^{(1)}(f,P)[\varphi],\psi\rangle=
			\sum_{j=1}^{\ell(\underline{\mathrm{k}})} k_j\Phi^{\widehat{\text{\ul{k}}}_j}[M;h](\psi f,P,\varphi),
		\end{align}
		where $(\widehat{\text{\ul{k}}}_j)_n=k_n$ for all $n\neq j$ while $(\widehat{\text{\ul{k}}}_j)_j=k_j-1$ -- we set by definition $\Phi^{\widehat{\underline{\mathrm{k}}}_j}=0$ whenever $k_j-1<0$ for some $j\in\{1,\ldots,\ell(\underline{\mathrm{k}})\}$;
		
		\item
		for all sequences $\text{\ul{k}}$, let $(M;h_s)\in\obj(\bkgg)$ with $\lbrace h_s\rbrace_{s\in\R^n}$ be a smooth and compactly supported family of variations of $h$, as per Definition \ref{Def: smooth compactly supported d-dimensional family of variations}.
		Let $\operatorname{V}^{\ell(\underline{\mathrm{k}})} M$ be the trivial bundle $M^{\ell(\underline{\mathrm{k}})}\times\mathbb{C}$.
		For any smooth family $\lbrace P_s\rbrace_{s\in\R^n}$ with $P_s\in\mathrm{Par}(M;h_s)$ and for any $s\in\R^n$, let $\mathcal{U}_{\text{\ul{k}}}\in\mathcal{D}^\prime(\widetilde{\pi}^*_n(\operatorname{V}^{\ell(\underline{\mathrm{k}})} M))$ be the distribution on the pull-back bundle $\widetilde{\pi}^*_n \operatorname{V}^{\ell(\underline{\mathrm{k}})} M$ -- here $\widetilde{\pi}_n:\R^n\times M^{\ell(\underline{\mathrm{k}})}\to M^{\ell(\underline{\mathrm{k}})}$ denotes the canonical projection -- such that, for any $f\in\Gamma_{\mathrm{c}}^{{\ell(\underline{\mathrm{k}})}}[M;h]$ and $\chi\in C^\infty_\mathrm{c}(\R^n)$,
		\begin{align}\label{Eqn: distribution for msc of T-product}
			\mathcal{U}_{\text{\ul{k}}}(\chi,f)\vcentcolon=\int_{\R^n}\mathrm{d}s\;\Phi^{\text{\ul{k}}}[M;h_s](f,P_s,0)\chi(s)\,.
		\end{align}
		We require that $\mathrm{WF}(\mathcal{U}_{\text{\ul{k}}})\subseteq C_{({\ell(\underline{\mathrm{k}})})}(M;h)$ where
		\begin{align}\label{Eqn: set for microlocal spectrum condition of E-product}
			\nonumber
			C_{({\ell(\underline{\mathrm{k}})})}(M)=\bigg\lbrace(x_1,p_1;\dots;x_{\ell(\underline{\mathrm{k}})},p_{\ell(\underline{\mathrm{k}})};s,\tau)\in T^*(\widetilde{\pi}_n^*V^{\ell(\underline{\mathrm{k}})} M)\setminus\lbrace0\rbrace\,|\,\exists\,I=\lbrace i_1,\dots, i_{|I|}\rbrace\subset\lbrace1,\dots,{\ell(\underline{\mathrm{k}})}\rbrace,\\
			 |I|\geqslant2 : (x_{i_1},\dots, x_{i_{|I|}})\in\mathrm{Diag}(M^{|I|}), \sum_{i\in I}p_i=0\bigg\rbrace;
		\end{align}
	\end{enumerate}

\end{Definition}

\begin{remark}
	Observe that Definition \ref{Def: E-product} coincides with Definition \ref{Def: Wick polynomials} when ${\ell(\underline{\mathrm{k}})}=1$. In particular, in this case, the set $C_{(1)}(M)$ as per Equation \eqref{Eqn: set for microlocal spectrum condition of E-product} is empty.
\end{remark}

\begin{remark}
	It is noteworthy that, for the particular choice $k_j=1$ for all $1\leq j\leq\ell$, axiom (2-3) of Definition \ref{Def: E-product} leads to
	\begin{align*}
		\Phi^{(1,\ldots,1)}[M;h](f_1\otimes\ldots\otimes f_\ell)
		=\Phi[M,h](f_1)\cdot\ldots\cdot\Phi[M;h](f_\ell)\,,
	\end{align*}
	for all $f_1,\ldots f_\ell\in\Gamma_{\mathrm{c}}^1[M;h]$.
	Moreover, for every finite sequence $\underline{\mathrm{k}}=(k_1,\ldots,k_\ell)$ such that $k_j=0$ for some $j\in\{1,\ldots,\ell\}$ we have
	\begin{multline*}
		\Phi^{(k_1,\ldots,k_\ell)}[M;h](f_1\otimes\ldots\otimes f_\ell)
		=\Phi^{(k_1,\ldots,k_{j-1})}[M;h](f_1\otimes\ldots\otimes f_{j-1})\\
		\cdot\operatorname{Id}_{\Gamma_{\mathrm{eq}}[M;h]}(f_j)
		\cdot\Phi^{(k_{j+1},\ldots,k_\ell)}[M;h](f_{j+1}\otimes\ldots\otimes f_\ell)\,,
	\end{multline*}
	for all $f_1,\ldots f_\ell\in\Gamma_{\mathrm{c}}^1[M;h]$, where $\operatorname{Id}_{\Gamma_{\mathrm{eq}}[M;h]}(f_j,P,\varphi)=\int_M f_j\mu_g$.
\end{remark}

\begin{remark}\label{Rmk: causal factorization axiom}
	Notice that, whenever $F,G\in\mathcal{P}_{\mathrm{loc}}(M;h)$ are such that $\operatorname{supp}(F)\cap\operatorname{supp}(G)=\emptyset$, formula \eqref{Eqn: Algebra product} for $F\cdot_PG$ is well-defined for all $P\in\operatorname{Par}[M;h]$ on account of the singular structure of the parametrices, {\it cf} \eqref{Eqn: Parametrix wave front set}. In turn, this entails that the right-hand side of Equation \eqref{Eqn: causal factorization} is well-defined. We will refer to axiom $(3)$ of Definition \ref{Def: E-product} as the \emph{support factorization} axiom. It was first introduced in \cite{Keller-09} under the name of ``causal factorization'' to make a more direct contact with the nomenclature used for quantum fields on globally hyperbolic backgrounds. We prefer to call it differently to emphasize the marked differences between theories built on manifolds with Euclidean and Lorentzian signature.
\end{remark}
\begin{remark}\label{Rmk: particular cases for Wick monomials}
	For each sequence $\text{\ul{k}}$, the transformation $\Phi^{\text{\ul{k}}}\colon\Gamma_{\mathrm{c}}^{\ell(\underline{\mathrm{k}})}\to\Gamma_{\operatorname{eq}}$ should be interpreted as a prescription for the product of finitely many Wick powers (specifically $\Phi^{k_1},\ldots,\Phi^{k_n}$) at different base points $x_1,\ldots,x_n\in M$ -- \textit{cf.} Proposition \ref{Prop: local and covariant algebra of local polynomials}.
	This is also consistent with axiom $(3)$.
\end{remark}

\noindent To conclude the section we show that the E-product allows to identify an Euclidean locally covariant field theory built out of the Wick powers of the underlying scalar field.

\begin{proposition}\label{Prop: local and covariant algebra of local polynomials}
	Let $\{\Phi^{\text{\ul{k}}}\}_{\text{\ul{k}}}$ be a family of Wick monomials associated with an arbitrary but fixed family of Wick powers $\{\Phi^k\}_k$.
	With reference to Definition \ref{Def: locally covariant algebra of regular observables}, for all $(M;h)\in\bkgg$, let $\mathcal{A}\colon\bkgg\to\alg$ be the covariant functor such that
\begin{itemize}
	 \item for every $(M;h)\in\obj(\bkgg)$, $\mathcal{A}[M;h]\subset\Gamma_{\mathrm{eq}}(\mathsf{E}[M;h])$ is the $\ast$-algebra which is generated by $\{\Phi^k(M;h)(f)|\;k\in\mathbb{N}\,,\,f\in\Gamma_{\mathrm{c}}^1[M;h]\}$ where we set
	\begin{align}\label{Eqn: definition of the produt on the local and covariant algebra of local polynomials}
		\Phi^{k_1}[M;h](f_{1})\cdots\Phi^{k_\ell}[M;h](f_{\ell})\vcentcolon=
		\Phi^{(k_1,\ldots,k_\ell)}[M;h](f_{1}\otimes\dots\otimes f_{\ell})\,.
	\end{align}
	(The $\ast$-operation is induced by complex conjugation as in Proposition \ref{Prop: regular algebra}.)
	\item for any arrow $\chi:M\to\widetilde{M}$ and for any $G\in\mathcal{A}[M;h]$, $P\in\operatorname{Par}[\widetilde{M};\widetilde{h}]$ and $\varphi\in \mathcal{E}(\widetilde{M})$, $(\mathcal{A}(\chi)G)(P,\varphi)=G(\chi^*P,\chi^*\varphi)$.
	\item 
	for any $\lambda>0$ the scaling $\varsigma_\lambda\colon\mathcal{A}[M;h]\to\mathcal{A}[M;h_\lambda]$ is defined as in Remark \ref{Remark: scaling map} -- there is no issue in extending its action on $\mathcal{A}[M;h]$.
\end{itemize}
	Then $\mathcal{A}\colon\bkgg\to\alg$ is an Euclidean locally covariant theory as per Definition \ref{Def: Euclidean locally covariant theory}.
	Moreover, for all $k\in\mathbb{N}$, $\Phi^k\colon\Gamma_{\mathrm{c}}^\bullet\to\mathcal{A}$ is a locally covariant observable as per Definition \ref{Def: locally covariant observable}.
\end{proposition}
\begin{proof}	
	The proof follows slavishly that of Theorem \ref{Theorem: locally covariant algebra of interest}, taking into account Definitions \ref{Def: Wick polynomials} and \ref{Def: E-product} which guarantee in particular that Equation \eqref{Eqn: definition of the produt on the local and covariant algebra of local polynomials} is well-posed.
	Notice that the product $\cdot$ defined on $\mathcal{A}[M;h]$ as per equation \eqref{Eqn: definition of the produt on the local and covariant algebra of local polynomials} is commutative and associative since it inherits these properties from those of the symmetrized tensor product between elements in $\Gamma_{\mathrm{c}}^\bullet[M;h]$.
\end{proof}
\subsection{Existence of the E-Product of Wick Ordered Powers of Quantum Fields}\label{Section: Existence E-product of Wick of Wick powers}
Much in the same spirit of the analysis in Section \ref{Section: Existence of Wick powers}, our next goal consists of proving the existence of a prescription for defining an E-product of Wick polynomials satisfying the axioms in Definition \ref{Def: E-product}. To this end, we will follow the same strategy of \cite{Hollands-Wald-02}, to which we also refer for the proofs of some results. Since the construction is rather complicated, we divide it in different steps, to each of which we dedicate a subsection.

\subsubsection{First Step: The inductive hypothesis}

The construction of an E-product proceeds inductively with respect to $\ell={\ell(\underline{\mathrm{k}})}\in\mathbb{N}$. More precisely, for all $\ell\in\mathbb{N}$, $\Phi^{\underline{\mathrm{k}}}$ is constructed for all possible $k_1,\ldots,k_\ell$.
The starting point consists of the observation that, on account of axiom $(2)$ of Definition \ref{Def: E-product}, if $\ell=1$ the Wick monomials coincide with the Wick powers, whose existence has been discussed and proven in Section \ref{Section: Existence of Wick powers}. As a consequence we can make the inductive hypothesis, assuming the existence of a well-defined E-product of Wick powers with $\ell\leqslant n$. To conclude we need to prove the existence of a consistent prescription for $\ell=n+1$.

	The key observation originates from the support factorization axiom $(5)$ in Definition \ref{Def: E-product}, which entails that the E-product of $n+1$ Wick powers is completely determined on $M^{n+1}\setminus\mathrm{Diag}(M^{n+1})$ by its prescription on $n$ factors. This was first observed in \cite{Keller-09,Keller:2010xq} and it is the Riemannian counterpart of the same procedure followed in \emph{causal perturbation theory} on a globally hyperbolic spacetime. 


Concretely let $\underline{\mathrm{k}}=(k_1,\ldots,k_{n+1})$, $I\subsetneq\lbrace1,\dots,n+1\rbrace$ and let $I^\mathrm{c}$ be its complement. We define 
\begin{align*}
C_I\vcentcolon=\lbrace(x_1,\dots, x_{n+1})\in M^{n+1}\,|\,x_i\neq x_j\,\forall\,i\in I\,,j\in I^\mathrm{c}\rbrace\subset M^{n+1},
\end{align*}
observing that, letting $I$ vary, $\lbrace C_I\rbrace_I$ identifies an open cover of $M^{n+1}\setminus\mathrm{Diag}(M^{n+1})$. Let $\lbrace f_I\rbrace$ be a partition of unity subordinated to the open cover $\lbrace C_I\rbrace_I$ and, working at the level of integral kernels on $M^{n+1}\setminus\mathrm{Diag}(M^{n+1})$, we set
\begin{multline}\label{Eqn: outside diagonal T-product}
	\Phi_0^{\text{\ul{k}}}[M;h](P,\varphi)(x_1,\dots, x_{n+1})\vcentcolon
	=\sum_{I\subsetneq\lbrace1,\dots,n+1\rbrace}f_I(x_1,\dots,x_{n+1})
	\Phi^{\text{\ul{k}}_I}[M;h](P,\varphi)(x_1,\dots, x_{|I|})\\
	\cdot\Phi^{\text{\ul{k}}_{I^\mathrm{c}}}[M;h](P,\varphi)(x_1,\dots, x_{|I^\mathrm{c}|})\,.
\end{multline}

\vskip.2cm

\noindent Notice that, on account of the inductive hypothesis and of Definition \ref{Def: E-product},
\begin{enumerate}
	\item $\Phi_0^{\text{\ul{k}}}[M;h](P,\varphi)(x_1\dots, x_{n+1})$ is well-defined since it is a linear combination of E-products between factors of order less or equal to $n$;
	\item $\Phi_0^{\text{\ul{k}}}[M;h](P,\varphi)(x_1\dots, x_{n+1})$ is independent from the chosen partition of unity and any prescription for the E-product of $n+1$ factors must be of the form \eqref{Eqn: outside diagonal T-product} on $M^{n+1}\setminus\mathrm{Diag}(M^{n+1})$.
\end{enumerate}

\subsubsection{Second Step: Local Wick Expansion}
In order to extend $\Phi^{\text{\ul{k}}}_{0}[M;h](P,\varphi)(x_1\dots, x_{n+1})$ to $\mathrm{Diag}(M^{n+1})$ we introduce the \emph{local Wick expansion}. More precisely, consider an open cover of $M$ in terms of convex geodesic neighbourhoods and, for any open set $O$ in such cover and for all $n\in\mathbb{N}$, let $O^{n+1}=\underbrace{O\times\dots\times O}_{n+1}$. Recalling Equation \eqref{Eqn: Hadamard representation}, each parametrix associated to the elliptic operator $E$ in \eqref{Eqn: local_form_E} can be decomposed in $O$ as $P(x,y)=H(x,y)+W_P(x,y)$, where $W_P\in \mathcal{E}(O\times O)$.
Hence, for any $k_1,\ldots,k_{n+1}\in\mathbb{N}$ and for every $(M;h)\in\obj(\bkgg)$, consider the functional
\begin{align*}
	\phi^{\text{\ul{k}}}[M;h](\omega_{n+1},\varphi)\vcentcolon
	=\int_{M^{n+1}}\,\mu_g^{n+1}\langle\varphi^{\text{\ul{k}}},\omega_{n+1}\rangle,
\end{align*}
where $\omega_{n+1}\in\Gamma_{\mathrm{c}}^{n+1}[M;h]$, $\varphi\in\mathcal{E}(M)$, while $\varphi^{\text{\ul{k}}}(x_1,...,x_{n+1})=\prod\limits_{i=1}^{n+1}\varphi^{k_i}(x_i)$ and $\mu_g^{n+1}(x_1,...,x_{n+1})=\prod\limits_{i=1}^{n+1}\mu_g(x_i)$.
Starting from these data and working at the level of integral kernels, for every $P\in\operatorname{Par}[M;h]$, $\varphi\in\mathcal{E}(M)$ we set
\begin{align*}
	\wick{\Phi^{\text{\ul{k}}}}_H[M;h](\omega_{n+1},P,\varphi)\;\vcentcolon=
	\exp[\Upsilon_{W_P}]\big(\phi^{\text{\ul{k}}}[M;h](\omega_{n+1},\varphi)\big)\,,
\end{align*}
for all $\omega_{n+1}\in\Gamma_{\mathrm{c}}^{n+1}[M;h]$ with $\operatorname{supp}(\omega_{n+1})\subseteq O^{n+1}$.
Here $\exp[\Upsilon_{W_P}]$ has been defined in \eqref{Eqn: alpha map} while $W_P$ has been introduced in Remark \ref{Remark: Hadamard representation}.
Notice that, $W_P(x,y)$ is well-defined for $x,y\in O$ after introducing a cut-off in the definition of $H$ -- \textit{cf.} Remark \ref{Remark: Hadamard representation}.
As we are considering a local expansion near the total diagonal, this does not affect the local and covariant behaviour of $\Phi^{\underline{\textrm{k}}}$.
Observe that $\exp[\Upsilon_{W_P}]\big(\phi^{\text{\ul{k}}}[M;h](\omega_{n+1},\varphi)\big)$ is well-defined as a consequence of the support properties of $\omega_{n+1}$.
In what follows we shall denote with $\wick{\Phi^{\text{\ul{k}}}}_H[M;h](P,\varphi)(x_1\dots, x_{n+1})$ the integral kernel associated to $\omega_{n+1}\to \wick{\Phi^{\text{\ul{k}}}}_H[M;h](\omega_{n+1},P,\varphi)$.
The following proposition can be proven \emph{mutatis mutandis} as in \cite[Sect. 3.2]{Hollands-Wald-02}.

\begin{proposition}
Any prescription for the E-product satisfying axioms $(3)$ and $(4)$ of Definition \ref{Def: E-product} admits a \emph{local Wick expansion} of the form, 
\begin{align}\label{Eqn: local Wick expansion}
	\Phi^{\text{\ul{k}}}[M;h](P,\varphi)(x_1\dots, x_{n+1})=
	\sum_{\underline{j}\leq\text{\ul{k}}}
	\binom{\text{\ul{k}}}{\underline{j}}\widetilde{t}_{\underline{j}}[M;h](x_1,\dots,x_{n+1})
	\wick{\Phi^{\underline{k}-\underline{j}}}_H[M;h](P,\varphi)(x_1\dots, x_{n+1})\,,
\end{align}
where $\underline{j}\leq\text{\ul{k}}$ if $0\leq j_i\leq k_i$ for all $i\in\{1,\ldots,n+1\}$ while $\binom{\text{\ul{k}}}{\underline{j}}\vcentcolon=\prod_{i=1}^{n+1}\binom{k_i}{j_i}$.
Here\footnote{Notice that, whenever $j_h=0$ for some $h$, the corresponding integral kernel $t_{\underline{j}}$ does not depend explicitly on $x_h$.
Whenever $j_h=k_h$ for some $h$, Remark \ref{Rmk: particular cases for Wick monomials} applies.
}
$\widetilde{t}_{\underline{j}}[M;h]\in\mathcal{D}^\prime(M^{n+1})$ is such that $\mathrm{WF}(\widetilde{t}_{\underline{j}})\subset C_{(n+1)}(M;h)$, where $C_{(n+1)}(M;h)$ is defined in Equation \eqref{Eqn: set for microlocal spectrum condition of E-product}.
Moreover, each $\widetilde{t}_{\underline{j}}$ is local and covariant, in particular none depends on the parametrix $P\in\operatorname{Par}[M;h]$ appearing in equation \eqref{Eqn: local Wick expansion}.
\end{proposition}

\vskip .2cm

Equation
\eqref{Eqn: local Wick expansion} satisfies axioms $(3)$ and $(4)$ on $M^{n+1}\setminus\mathrm{Diag}(M^{n+1})$ on account of the inductive hypothesis. Therefore, we conclude that, on $M^{n+1}\setminus\mathrm{Diag}(M^{n+1})$, there exists a collection of distributions $t_{\underline{j}}[M;h]\in\mathcal{D}^\prime(M^{n+1}\setminus\mathrm{Diag}(M^{n+1}))$ such that
\begin{align*}
	\Phi_0^{\text{\ul{k}}}[M;h](P,\varphi)(x_1\dots, x_{n+1})=&\\
	\sum_{\underline{j}\leq\text{\ul{k}}}\binom{\text{\ul{k}}}{\underline{j}}t_{\underline{j}}[M;h](x_1,&\dots,x_{n+1})
	\wick{\Phi^{\text{\ul{k}}-\underline{j}}}_H[M;h](P,\varphi)(x_1\dots, x_{n+1})\,.
\end{align*}

As a consequence of this formula, the problem of extending $\Phi^{\text{\ul{k}}}_0[M;h](P,\varphi)$ to the diagonal is reduced to that of extending $t_{\underline{j}}[M;h]$ to $\mathrm{Diag}(M^{n+1})$. To overcome this hurdle we reformulate in terms of the distributions $t_{\underline{j}}$ those axioms of Definition \ref{Def: E-product} which have not been already implemented in the construction above. This yields
\begin{description}
	\item[Axiom $\mathsf{E}$1)]
	Each $t_{\underline{j}}[M;h]$ is local and covariant, namely, given $(N;h_N),(M;h_M)\in\obj(\bkgg)$ and $\chi\in\arr(\bkgg)$ such that $\chi:N\to M$, then 
	\begin{align}\label{Eqn: local covariance of t}
	\chi^*t_{\underline{j}}[M;h_M]=t_{\underline{j}}[N;h_N].
	\end{align}
	Furthermore $t_{\underline{j}}[M;h]$ ought to scale almost homogeneously with dimension $\kappa_{\underline{j}}\vcentcolon=\sum_{i=1}^{n+1}j_i\big(\frac{\mathrm{D}-2}{2}\big)$ and order $m_{\underline{j}}\vcentcolon=\sum_{\ell=1}^{n+1}j_i$, namely
	\begin{align*}
		\lambda^{-\kappa_{\underline{j}}}t_{\underline{j}}[M;h_\lambda]=
		t_{\underline{j}}[M;h]+
		\sum\limits_{\ell\leq m_{\underline{j}}}\frac{\log^{\ell}(\lambda)}{\ell!}v_{\ell}[M;h]\,,
	\end{align*}
	where $v_{\ell}[M;h]$ are local and covariant distributions which scale almost homogeneously with degree $\kappa_{\underline{j}}$ and order $m_{\underline{j}}-\ell$, \textit{cf.} Definition \ref{Def: rescaled locally covariant observable}.

	\item[Axiom $\mathsf{E}$2)]
	for any  multi-index $\underline{j}=(j_1,\dots,j_{n+1})$, let us consider the distribution $\mathcal{T}_{\underline{j}}\in\mathcal{D}'(\mathbb{R}^{d}\times O^{n+1})$ defined by
	\begin{align}\label{Eqn: ancillary distribution T}
		\mathcal{T}_{\underline{j}}(\chi\otimes f)\vcentcolon=
		\int_{\R^n}\mathrm{d}s\;t_{\underline{j}}[M;h_s](f)\chi(s)\quad\chi\in \mathcal{D}(\R^d)\,,\forall f\in\Gamma_{\mathrm{c}}^\ell[O;h_{O}]\,,
	\end{align}
	where $\lbrace h_s\rbrace_{s\in\R^d}$ is a family of smooth and compactly supported variations of $h$ as per Definition \ref{Def: smooth compactly supported d-dimensional family of variations}.
	It must hold
	\begin{align}\label{Eqn: microlocal spectrum condition for t}
		\mathrm{WF}(\mathcal{T}_{\underline{j}})|_{\mathbb{R}^d\times\mathrm{Diag}(M^{n+1})}\perp T(\mathbb{R}^d\times\mathrm{Diag}(M^{n+1}))\,,
	\end{align}
	where $T(\mathbb{R}^d\times\mathrm{Diag}(M^{n+1}))$ denotes the tangent bundle to $\mathbb{R}^d\times \operatorname{Diag}(M^{n+1})$, while the symbol $A \perp B$ means that $\langle a,b\rangle=0$ for all $a\in A$ and $b\in B$, $\langle\;,\;\rangle$ being the standard fiberwise pairing.

	\item[Axiom $\mathsf{E}$3)]
	Each $t_{\underline{j}}[h_s]$ must be symmetric and real valued.
\end{description}

\subsubsection{Scaling Expansion}

The next step consists of investigating the scaling behaviour of the distributions $t_{\underline{j}}[M;h]$ as a preliminary step towards analysing their extension to the diagonal. This part of our analysis follows slavishly that of \cite{Hollands-Wald-02} and the strategy calls for working at the level of integral kernels on $M^{n+1}$ keeping one of the variables fixed, while letting the others vary. Hence, for clarity of the notation, we set
\begin{align*}
x=x_1,\qquad y=(x_2,\dots, x_{n+1}).
\end{align*}

\begin{proposition}
Let $x\in M$ be any fixed point and let $O$ be a geodesically convex normal neighbourhood centred at $x$. Each $t_{\underline{j}}[M;h_s]\in\mathcal{D}^\prime(M^{n+1})$ admits a restriction to $C_x\vcentcolon=\lbrace x\rbrace\times\big(O^n\setminus\underbrace{(x,\dots, x)}_n\big)\subset M^{n+1}$.
\end{proposition}

\begin{proof}
The conormal bundle $N^*C_x$ of $C_x$ is spanned by elements of the form $(x,k;y,\underline{0})$ where $\underline{0}\in T^*_yM^n$. On account of axiom $\mathsf{E}2$, $\operatorname{WF}(t_{\underline{j}}[M;h_s])\cap N^*C_x=\emptyset$ and, thus, on account of \cite[Theorem 8.2.4]{Hormander-83}, the sought statement descends.
\end{proof}

We introduce the notion of \emph{scaling expansion}, namely we consider a geodesically convex normal neighbourhood $O\subset M$ centred at $x$ and choosing any isometric isomorphism $e: T_xM\to\R^{\mathrm{D}}$, with $\mathrm{D}=\mathrm{dim}(M)$, we endow $O$ with the local chart $\alpha_x:O\to\mathbb{R}^\mathrm{D}$ such that, for every $y\in O$
\begin{equation}\label{Eqn: alpha map}
 \alpha_x(y)=e\circ(\exp_x)^{-1}(y),
\end{equation}
where, with a slight abuse of notation, we do not make explicit the dependence of $\alpha$ on $e$.
Any other choice $e^\prime:T_xM\to\R^{\mathrm{D}}$ is related to $e$ by the action of an element $\Lambda\in SO(\mathrm{D})$.

Restricting our attention to $O$, consider thereon a smooth one parameter family of metrics $\{g^{(s)}\}_{s\in\R}$ such that, calling $\chi_s:TO\to TO$ the map $\chi_s(y,\xi)\vcentcolon=(y,s\xi)$ for all $(y,\xi)\in TO$, then $g^{(s)}(\cdot,\cdot)\vcentcolon=s^{-2}g(\chi_s\cdot,\chi_s\cdot)$.
Associated with this structure, we define the one-parameter family of background fields restricted to $O$, $h^{(s)}\vcentcolon=(g^{(s)},A,c)$ in which $A\in\Gamma(T^*O)$ and $c\in C^\infty(O)$ are left fixed.
As a consequence of axiom $\mathsf{E}2$, a partial evaluation of $t[M;h^{(s)}]\in\mathcal{D}^\prime(M^{n+1})$ against a test-function $f\in\mathcal{D}(M^n)$ yields $t[M;h^{(s)}](\delta_x\otimes f)$, a smooth function of $(s,x)$. As a consequence derivatives along the $s$-direction are well-defined and, for any $k\in\N$, we introduce on $O^n\setminus(x,\dots, x)$ the distribution
\begin{align*}
\tau_k[M;h](x,\cdot)\vcentcolon=\frac{\mathrm{d}^k}{\mathrm{d}s^k}t[M;h^{(s)}](x,\cdot)|_{s=0}.
\end{align*}
On account of the smoothness of $t[M;h^{(s)}]$ in $(s,x)$ once tested along the remaining variables, we can apply Taylor expansion theorem writing for every integer $m\geq 0$ 
\begin{align}\label{Eqn: scaling expansion for t outside the diagonal}
t_{\underline{j}}[M;h](x,\cdot)=\sum_{k=0}^m\tau_k[M;h](x,\cdot)+r_m[M;h](x,\cdot),
\end{align}
where, with a slight abuse of notation, we omit the $\underline{j}$-dependence on the right hand side and where the remainder reads
\begin{align*}
r_m[M;h](x,\cdot)=\frac{1}{m!}\int_0^1\,ds\,(1-s)^m\,\frac{d^{m+1}t[M;h^{(s)}]}{ds^{m+1}}(x,\cdot).
\end{align*}

The procedure outlined and Equation \eqref{Eqn: scaling expansion for t outside the diagonal} are referred to as {\it scaling expansion}. This enjoys several notable properties which are summarized in the following theorem whose proof we omit since, mutatis mutandis, it is the same as the one of \cite[Theorem 4.1]{Hollands-Wald-02}.

\begin{theorem}\label{Theorem: properties of the scaling expansion}
With reference to Equation \eqref{Eqn: scaling expansion for t outside the diagonal} it holds that
\begin{enumerate}[label=(\roman*)]
\item $\tau_k[M;h](x,\cdot)$ and $r_m[M;h](x,\cdot)$ lie in $\mathcal{D}^\prime(O^n\setminus\{(x,\dots,x)\})$ and they have a covariant dependence on the metric, \ie, 
for every $(M;h)\in\obj(\bkgg)$ and for every $\psi\in\arr(\bkgg)$ from $(M,\psi^*h)$ to $(M;h)$ such that $\psi(x)=x$ it holds
\begin{align*}
	\psi^*\tau_k[M;h]=\tau_k[M;\psi^*h],\quad\mathrm{and}\quad\psi^*r_m[M;h]=r_m[M;\psi^*h]\,.
\end{align*}
\item for any $(M;h)\in\obj(\bkgg)$, the integral kernel of $\tau_k[M;h]$ decomposes as a finite sum of the form
\begin{align}\label{Eqn: decomposition of tau}
\tau_k[M;h](x,y)=\sum_I C_I(x)\alpha^*u_I(y),
\end{align}
where $I$ is a finite index set, while $C_I(x)\alpha^*u_I(y)\equiv(C_I(x))_{\mu_1\dots\mu_j}
(\alpha^*u_I(y))^{\mu_1\dots\mu_j}$. In addition 
each $C_I$ is built out of the components of suitable curvature tensors evaluated at $x$, \ie, sums of monomials in the metric $g$, in the Riemann tensor and in its covariant derivatives at most up to order $k-2$. Furthermore each $u_I$ is a tensor-valued $\mathrm{SO}(n\mathrm{D})$-covariant distribution on $\R^{n\mathrm{D}}\setminus\lbrace0\rbrace$, that is there exists a finite $j\in\mathbb{N}$ such that 
\begin{align*}
(u_I)_{a_1\dots a_j}(\Lambda\cdot)=\Lambda^{b_1}_{a_1}\dots\Lambda^{b_j}_{a_j}(u_I)_{b_1\dots b_j}(\cdot),\quad\forall\,\Lambda\in\mathrm{SO}(n\mathrm{D}),
\end{align*}
where the indices $(a_1,\dots,a_j)$ and $(b_1\dots b_j)$ refer to an expansion with respect to an arbitrary coordinate system on $T^*\mathbb{R}^{n\mathrm{D}}$. 
\item recalling the implicit dependence on $\underline{j}$, $\tau_k[M;h](x,\cdot)$ and $r_m[M;h](x,\cdot)$ scale almost homogeneously
	with dimension $\kappa_{\underline{j}}=\sum_{i=1}^{n+1}j_i\big(\frac{\mathrm{D}-2}{2}\big)$ and order $m_{\underline{j}}=\sum_{i=1}^{n+1}j_i$;
\item for any integer $k\geq 0$, the distributions $u_I$ in Equation \eqref{Eqn: decomposition of tau} scale almost homogeneously with dimension $\kappa_{\underline{j}}-k$
and finite order $N\in\mathbb{N}$ \textit{with respect to coordinate rescaling}
\footnote{Per definition this means that for all $\lambda>0$ the distributions whose integral kernel is given by $u_I(x), u_I(\lambda x)$ are related by $u(\lambda x)=\lambda^{\kappa_{\underline{j}}-k}\big(u(x)+\sum_{\ell=1}\log(\lambda)^\ell v_\ell(x)\big)$, $v_\ell(x)$ being a distribution which scales almost homogeneously of degree $\kappa_{\underline{j}}-k$ and order $N-\ell$ -- the definition is inductive and a distribution which scales almost homogeneously of degree $\kappa$ and order $0$ scales in fact homogeneously};

\item the scaling degree (sd) of the distribution $r_m[M;h](x,\cdot)$ is such that $\mathrm{sd}(r_m[M;h](x,\cdot))\leqslant |\underline{j}|-m-1$, {\it cf.} \cite{Brunetti:1999jn}.
\end{enumerate}
\end{theorem} 

As a by product of this last theorem extending $t_{\underline{j}}$ on $\mathrm{Diag}(M^{n+1})$ is tantamount to extending thereon $\tau_k$, $k=0,\dots,m$, for a given $m\in\N$ large enough, and $r_m$ as in \eqref{Eqn: scaling expansion for t outside the diagonal}. 

\vskip .2cm

\noindent{\em Step 1)} We start from the remainder and, in view of item v) of Theorem \ref{Theorem: properties of the scaling expansion}, choosing $m=|\underline{j}|-n\mathrm{D}$, the scaling degree of  $r_m$ is $n\mathrm{D}-1$. On account of \cite[Theorem 5.2]{Brunetti:1999jn}, $r_m$ admits a unique extension to the whole $O^n$ which can be constructed as follows. Let $\lbrace\vartheta^{(j)}\rbrace$ be smooth functions identically $1$ outside a neighbourhood $\mathcal{U}_{n+1}^{(j)}$ of $\mathrm{Diag}(M^{n+1})$ and supported in $O^{n+1}\setminus\mathrm{Diag}(M^{n+1})$ in such a way that the support of $1-\vartheta^{(j)}$ shrinks to $\mathrm{Diag}(M^{n+1})$ as $j\to\infty$. The extension of $r_m$, is defined as the distribution $\tilde{r}_m$ such that, for all $f\in \mathcal{D}(M^n)$, $\tilde{r}_m[M;h](x,f)\vcentcolon=\lim_{j\to\infty}r_m[M;h](x,\vartheta^{(j)}f)$.

\vskip .2cm

\noindent{\em Step 2)} If we focus on $\tau_k[M;h](x,\cdot)$, we can use the following lemma whose proof is identical to that of \cite[Lemma 4.1]{Hollands-Wald-02}. Most notably it guarantees the existence of an extension of the distributions whose integral kernel is $u_I(y)$ as in Equation \eqref{Eqn: decomposition of tau}.

\begin{lemma}\label{Lemma: extension}
Let $u$ be any tensor valued $\mathrm{SO}(n\mathrm{D})$-invariant distribution on $\R^{n\mathrm{D}}\setminus\lbrace0\rbrace$ whose components are $u_{a_1\dots a_\ell}$. If under coordinate rescaling $u$ scales almost homogeneously with dimension $\rho\in\mathbb{R}$, then it admits a $\mathrm{SO}(n\mathrm{D})$-invariant extension $\tilde{u}$ to $\R^{n\mathrm{D}}$ which scales almost homogeneously with dimension $\rho$.
Two different extensions $\tilde{u},\hat{u}$ are such that
\begin{align*}
	\tilde{u}-\hat{u}=\sum_{|\alpha|\leq\lfloor\rho\rfloor}a_\alpha\delta^{(\alpha)}\,,
\end{align*}
where $a_\alpha\in\mathbb{R}$ while $\lfloor\rho\rfloor$ denotes the integer part of $\rho$.
\end{lemma}

\noindent As a consequence, we can extend $\tau_k[M;h](x,\cdot)$ by taking 
\begin{align*}
\tilde{\tau}_k[M;h](x,y)=\sum\limits_I C_I(x)\alpha^*_x \tilde{u}_I(y),
\end{align*}
where $\tilde{u}_I$ is the extension of $u_I$ as per Lemma \ref{Lemma: extension}.

\vskip .2cm

\noindent{\em Step 3)} Combining together the two previous steps we have built $\tilde{t}_j[M;h]$, extension of $t_{\underline{j}}[M;h]$ such that
$$\tilde{t}_{\underline{j}}[M;h]=\sum_{k=0}^m\tilde{\tau}_k[M;h](x,\cdot)+\tilde{r}_m[M;h](x,\cdot)\,.$$
After symmetrization, $\tilde{t}_{\underline{j}}[M;h]$ satisfies the axioms $\mathsf{E}1-\mathsf{E}3$.
This is a direct consequence of the analysis \cite[Section 4.3]{Hollands-Wald-02} adapted to the case in hand and therefore we omit it.

\subsection{Uniqueness of E-Product of Wick Ordered Powers of Quantum Fields}\label{Section: Uniqueness of E-product of Wick Ordered Powers of Quantum Fields}

In this section we discuss whether there exist ambiguities in the construction of the E-product of Wick polynomials.
In the same spirit of Section \ref{Section: Uniqueness of Wick Ordered Powers of Quantum Fields}, we split the main result in two theorems.
In the first we show that the difference between two E-products can be fully encoded in terms of suitable coefficients, whose characterization is at the heart of the second theorem. 
\begin{remark}\label{Remark: coinciding point tensor product}
	In the following we shall adopt the following notation: given $u,v\in\Gamma(B)$ two sections of a vector bundle $B\to M$ we shall denote with $u[\otimes]v\vcentcolon=[u\otimes v]\in\Gamma(B\otimes B)$. Notice that $u\otimes v\in\Gamma(B\boxtimes B)$, while $[u\otimes v]$ denotes the coinciding point limit of $u\otimes v$, that is, $[u\otimes v](x):=(u\otimes v)(x,x)$.
\end{remark}

In the following we introduce additional structures which will allow us to discuss with the same notation both the case in hand and the Wick polynomials in presence of derivatives of the underlying field configurations, see Section \ref{Section: The Case of a Scalar Field with Derivatives}. In particular Definition \ref{Def: functor Gamma} has to be modified as follows. Let us now consider the jet bundle $J_\infty(M)$ over $M$, namely the inductive limit of the $n$-jet bundles $J_n(M)$, $n\in\mathbb{N}$ -- see \cite{Kolar-Michor-Slocvak-93} for further details. Moreover, we denote with $j_\infty\colon \mathcal{E}(M)\to \Gamma(J_\infty(M))$ the inductive limit of the $k$-jet prolongation maps $j_k\colon \mathcal{E}(M)\to\Gamma(J_k(M))$. 

Let $\text{\ul{k}}=(k_n)_n$ be a finite sequence of $\ell(\underline{\mathrm{k}})$ many strictly positive integers as in Section \ref{Section: E-Product of Wick Ordered Powers of Quantum Fields}.
To each $\text{\ul{k}}$ one associates a covariant functor $\Gamma_{\mathrm{c}}^{\text{\ul{k}}}\colon\bkgg\to\vect$ such that, for any $(M;h)\in\obj(\bkgg)$ and $\chi\in\arr(\bkgg)$, we set
\begin{align}
	\label{Eqn: compactly supported section bundle of fixed order}
	\Gamma_{\mathrm{c}}^{\text{\ul{k}}}[M;h]&:=
	\Gamma_{\operatorname{c}}(\boxtimes_{j=1}^{\ell(\underline{\mathrm{k}})} S^{ k_j}J_\infty(M)^*)\,,\qquad
	\Gamma_{\mathrm{c}}^{\text{\ul{k}}}(\chi):=\boxtimes_{j=1}^{\ell(\underline{\mathrm{k}})} S^{ k_j}\chi_*\,,\\
	\Gamma^{\text{\ul{k}}}[M;h]&:=
	\Gamma(\boxtimes_{j=1}^{\ell(\underline{\mathrm{k}})} S^{ k_j}J_\infty(M))\,,\qquad
	\Gamma^{\text{\ul{k}}}(\chi):=\boxtimes_{j=1}^{\ell(\underline{\mathrm{k}})} S^{ k_j}\chi^*\,,
\end{align}
Here $S^{k}$ denotes the $k$-th symmetric tensor product while $\boxtimes$ denotes the external tensor product.

\begin{Example}\label{Example: local and covariant linear observable with derivatives}
	To better clarify to a reader the previous discussion we repeat with the new structures the example of a standard, linear local and covariant observable as in Example \ref{Example: quantum fields}. Given $(M;h)\in\obj(\bkgg)$ and $f\in\Gamma_\mathrm{c}^{1}[M;h]=\Gamma_{\mathrm{c}}(J_\infty(M)^*)$, let $\Phi[M;h](f)$ be the element of $\mathcal{A}_{\operatorname{reg}}[M;h]$ 
	\begin{align*}
	\Phi[M;h](f,P,\varphi)\vcentcolon=\int_M\langle f,j_\infty\varphi\rangle\mu_g,
	\end{align*}
	where $P\in\operatorname{Par}[M;h]$ and $\varphi\in \mathcal{E}(M)$, while $\langle\,,\,\rangle$ denotes the dual pairing. Observe that $\langle f,j_\infty\varphi\rangle$ involves finitely many derivatives of the field configuration $\varphi$. Locality and covariance descend as in Example \ref{Example: quantum fields}.
\end{Example}

\noindent In view of Definition \ref{Def: E-product}, the following theorem holds true. 

\begin{theorem}\label{Theorem: uniqueness of E-product I}
	Let $\{\widehat{\Phi}^k\}_{k\in\mathbb{N}}$ and $\{\Phi^k\}_{k\in\mathbb{N}}$ be two families of Wick powers associated to $\Phi$ as per Definition \ref{Def: Wick polynomials}.
	In addition let $\{\widehat{\Phi}^{\text{\ul{k}}}\}_{\text{\ul{k}}},\{\Phi^{\text{\ul{k}}}\}_{\text{\ul{k}}}$ be two family of Wick monomials respectively associated to $\{\widehat{\Phi}^k\}_{k\in\mathbb{N}}$ and $\{\Phi^k\}_{k\in\mathbb{N}}$, as per Definition \ref{Def: E-product} -- here $\text{\ul{k}}=(k_n)_n$ denotes an arbitrary finite sequence of $\ell(\underline{\mathrm{k}})$ many non-negative integers.
	Then for any $(M;h)\in\obj(\bkgg)$, and $\omega_{k_1}\otimes\ldots\otimes\omega_{k_{\ell(\underline{\mathrm{k}})}}\in\Gamma_{\mathrm{eq}}^{\text{\ul{k}}}[M;h]$ it holds
	\begin{align}\label{Eqn: ambiguities of E-product}
		\widehat{\Phi}^{\text{\ul{k}}}[M;h](\omega_{k_1}\otimes\ldots\otimes\omega_{k_{\ell(\underline{\mathrm{k}})}})&=
		\Phi^{\text{\ul{k}}}[M;h](\omega_{k_1}\otimes\ldots\otimes\omega_{k_{\ell(\underline{\mathrm{k}})}})\\\nonumber&+
		\sum_{\underset{|\wp|<{\ell(\underline{\mathrm{k}})}}{\wp\in\mathsf{P}\{1,\ldots,\ell(\underline{\mathrm{k}})\}}}\sum_{\underline{j}\leq\text{\ul{k}}_\wp} {\underline{\mathrm{k}_\wp}\choose\underline{j}}
		\Phi^{\underline{j}}[M,h]\bigg(\bigotimes_{I\in\wp}\bigg(c_{k_I-j_I}[M,h]\lrcorner\underset{i\in I}{\big[\bigotimes\big]}\omega_{k_i}\bigg)\bigg)\,,
	\end{align}
	where $\mathsf{P}\{1,\ldots,\ell(\underline{\mathrm{k}})\}$ denotes the set of partitions of $\{1,\dots,\ell(\underline{\mathrm{k}})\}$ in non-empty subsets while $\text{\ul{k}}_\wp=(k_I)_{I\in\wp}$ where $k_I\vcentcolon=\sum_{i\in I}k_i$.
	Furthermore, given a sequence $\underline{j}=(j_I)_{I\in\wp}$, $\underline{j}\leq\text{\ul{k}}_\wp$ if and only if $j_I\leq k_I$ for all $I\in\wp$.
	\footnote{
	If $\wp=\{I_1,\ldots,I_{|\wp|}\}$ then $(\text{\ul{k}}_\wp)_s=\sum_{i\in I_s}k_i$ for $1\leq s\leq|\wp|$ while $\underline{j}=\{j_s\}_{s=1}^{|\wp|}$ is such that $\underline{j}\leq\text{\ul{k}}_\wp$ if and only if $j_s\leq\sum_{i\in I_s}k_i$ for all $1\leq s\leq|\wp|$.}
	Finally $c_{k_I-j_I}[M,h]\lrcorner\underset{i\in I}{\big[\bigotimes\big]}\omega_{k_i}\in\Gamma_{\mathrm{c}}^{j_I}[M,h]$ denotes the symmetrized contraction between $\underset{i\in I}{\big[\bigotimes\big]}\omega_{k_i}\in\Gamma_{\operatorname{c}}^{\underline{k}_I}[M,h]$ and $c_{k_I-j_I}[M,h]\in\Gamma^{k_I-j_I}[M;h]$.
	Moreover $c_{k_I-j_I}[M,h]$ is weakly regular as per Definition \ref{Def: weak regularity condition} and the assignment
	\begin{align}\label{Eqn: local and covariant observable induced by local and covariant tensor}
		C_{k_I-j_I}[M,h](\omega)\vcentcolon=\int_M\mu_g\;c_{k_I-j_I}[M;h]\lrcorner\omega\;\operatorname{Id}_{\Gamma_{\mathrm{eq}}[M;h]}\qquad
		\forall\omega\in\Gamma_{\mathrm{c}}^{k_I}[M;h]\,,
	\end{align}
	defines a local and covariant observable -- \textit{cf.} Definition \ref{Def: locally covariant observable} -- which scales almost homogeneously with dimension $\frac{\mathrm{D}-2}{2}(k_I-j_I)$ with respect to the transformation $h=(g,A,c)\mapsto h_\lambda=(\lambda^{-2}g,A,\lambda^2c)$.
\end{theorem}

\begin{proof}
	For later convenience let us notice that in equation \eqref{Eqn: ambiguities of E-product} the term corresponding to $\wp=\{\{1,\ldots,\ell(\underline{\mathrm{k}})\}\}$ -- \textit{i.e.} the term corresponding with the trivial partition -- is given by
	\begin{align}\label{Eqn: counterterm for trivial partition}
		\sum_{j\leq|\text{\ul{k}}|}{|\underline{\mathrm{k}}|\choose j}
		\Phi^{j}[M,h]\bigg(c_{|\text{\ul{k}}|-j}[M,h]\lrcorner\underset{i\in\{1,\ldots,\ell\}}{\big[\bigotimes\big]}\omega_{k_i}\bigg)\,,
	\end{align}
	where $c_{|\text{\ul{k}}|-j}[M;h]$ enjoys the same properties of the tensors appearing in Theorem \ref{Theorem: Unicity of Wick powers I}.
	
	We proceed inductively with respect to $\ell=\ell(\underline{\mathrm{k}})$ and to $\text{\ul{k}}$.
	Notice that the thesis holds true if $\ell=1$, independently of the value of $\text{\ul{k}}=k_1$, since this case reduces to Theorem \ref{Theorem: Unicity of Wick powers I}.
	In addition the statement becomes trivial for all values of $\ell$, if $|\text{\ul{k}}|=0$ or $|\text{\ul{k}}|=1$.  
	
	Let us start by assuming the theorem to hold true up to order $\ell-1$ and proving it to order $\ell$.
	To this end, let us consider

	\begin{align}\label{Eqn: defining equation for Phi_k}
		\Phi_{|\text{\ul{k}}|}[M;h](\omega_{k_1}\otimes\ldots\otimes\omega_{k_\ell})&\vcentcolon=
		\widehat{\Phi}^{\text{\ul{k}}}[M;h](\omega_{k_1}\otimes\ldots\otimes\omega_{k_\ell})-
		\Phi^{\text{\ul{k}}}[M;h](\omega_{k_1}\otimes\ldots\otimes\omega_{k_\ell})\\\nonumber&-
		\sum_{\underset{1<|\wp|<\ell}{\wp\in\mathsf{P}\{1,\ldots,\ell\}}}\sum_{\underline{j}\leq\text{\ul{k}}_\wp}{\underline{\mathrm{k}_\wp}\choose\underline{j}}
		\Phi^{\underline{j}}[M,h]\bigg(\bigotimes_{I\in\wp}c_{k_I-j_I}[M,h]\lrcorner\underset{i\in I}{\big[\bigotimes\big]}\omega_{k_i}\bigg)\,,
	\end{align}
	As usual $\Phi_{|\text{\ul{k}}|}[M;h]$ is local and covariant with appropriate regularity and scaling.
	Moreover, on account of the support factorization axiom in Definition \ref{Def: E-product} and of the inductive hypothesis on $\ell$,
	\begin{align*}
		\Phi_{|\text{\ul{k}}|}[M;h](\omega_{k_1}\otimes\ldots\otimes\omega_{k_\ell})=
		\Psi_{|\text{\ul{k}}|}[M;h](\omega_{k_1}[\otimes]\ldots[\otimes]\omega_{k_\ell})\,,
	\end{align*}
	where $\Psi_{|\text{\ul{k}}|}[M;h]\colon\Gamma_{\mathrm{c}}^{|\text{\ul{k}}|}[M;h]\to\Gamma_{\mathrm{eq}}[M;h]$ is local and covariant with almost homogeneous scaling of degree $|\text{\ul{k}}|\mathrm{D}_\varphi$.
	The inductive assumption over $\ell$ together with induction over $|\text{\ul{k}}|$ implies that $\Psi_{|\text{\ul{k}}|}[M;k]$ can be written as
	\begin{align*}
		\Psi_{|\text{\ul{k}}|}[M;h](\omega)=
		\sum_{j\leq|\text{\ul{k}}|}{|\underline{\mathrm{k}}|\choose j}
		\Phi^{j}[M,h]\big(c_{|\text{\ul{k}}|-j}[M,h]\lrcorner\omega\big)\qquad
		\forall\omega\in\Gamma_{\mathrm{c}}^{\text{\ul{k}}}[M;h]\,,
	\end{align*}
	where $c_{|\text{\ul{k}}|-j}\in\Gamma^{|\text{\ul{k}}|-j}[M;h]$.
	Considering the locally covariant observables $C_{|\text{\ul{k}}|-j}[M;h]$ defined from $c_{|\text{\ul{k}}|-j}[M;h]$ as per equation \eqref{Eqn: local and covariant observable induced by local and covariant tensor} the proof is completed along the same lines of Theorem \ref{Theorem: Unicity of Wick powers I}.
\end{proof}


We conclude the section by stating a theorem, similar in spirit to Theorem \ref{Theorem: Unicity of Wick powers II}, which characterizes the form of the coefficients $c_{k_I-j_I}[M;h]$.

\begin{theorem}\label{Theorem: uniqueness of E-product II}
	Adopting the same notation of Theorem \ref{Theorem: uniqueness of E-product I}, for every $(M;h)\in\obj(\bkgg)$ the coefficients $c_{k_I-j_I}[M;h](x)$ appearing in equation \eqref{Eqn: ambiguities of E-product} are differential operators taking the form
	\begin{align*}
		c_{k_I-j_I}[M;h](x)=
		c_{k_I-j_I}[M;h]&[g^{ab}(x), \epsilon^{a_1\dots a_n}(x), R_{abcd}(x), \dots,\\&\nabla_{e_1}\dots\nabla_{e_n}R_{abcd}(x),
		A(x),\ldots,\nabla_{e_1}\dots\nabla_{e_n}A,
		c(x), \nabla_{e_1}\dots\nabla_{e_n}c(x)]\,,
	\end{align*}
	where $\epsilon^{a_1\dots a_n}(x)$ and $R_{abcd}(x)$ denote, respectively, the Levi-Civita and the Riemann curvature tensors built out of $g$ at $x\in M$, while $\nabla_{e_i}$ denotes the covariant derivative along the direction $e_i$.
	Furthermore each $c_k$ is a polynomial scalar function, covariantly constructed out of its arguments.
	In addition, every $c_{k_I-j_I}[M;h]$ scales homogeneously with dimension $k\big(\frac{\mathrm{D}-2}{2}\big)-\mathrm{D}(\ell-1)$ under the transformation $h=(g,A,c)\mapsto h_\lambda=(\lambda^{-2}g,A,\lambda^2c)$.
\end{theorem}

The theorem is a direct consequence of Theorem \ref{Theorem: Unicity of Wick powers II} and of \cite[Theorem 3.1]{Khavkine:2014zsa}, since the proof is based only on the properties of the coefficients $c_{k_I-j_I}[M;h](x)$, which are of the same type of those of the Wick polynomials on account of Theorem \ref{Theorem: uniqueness of E-product I}, {\it cf.} Theorem \ref{Theorem: Unicity of Wick powers I} .

\section{The Case of a Scalar Field with Derivatives}\label{Section: The Case of a Scalar Field with Derivatives}

So far our discussion of the locally covariant algebra $\mathcal{A}$ and of its Wick powers and monomials has been confined to the case of polynomial functionals $F\in\mathcal{P}$ which do not contain derivatives of the field configurations $\varphi$ -- \textit{cf.} Remark \ref{Remark: functional without dependence on derivative of the field configuration}.
For example, we did not consider functionals of the form $F(\varphi)\vcentcolon=\int_M\mu_g\,\omega^{ab}\varphi\nabla_a\varphi\nabla_b\varphi$ where $\omega\in \Gamma_{\operatorname{c}}(TM)$.
There is no issue a priori in extending the previous discussion so to account for arbitrary derivatives of the field configurations. Yet, as pointed out in \cite{Hollands-Wald-05}, one needs to add to the axioms for Wick monomials $\Phi^{\text{\ul{k}}}$ two additional requirements -- \textit{cf.} section \ref{Subsection: Additional axioms: Leibniz rule and Principle of Perturbative Agreement}. It is important to stress that the extension to this larger class of configurations is of paramount relevance in many concrete applications, as one can infer from the Lorentzian scenario, see {\it e.g.} \cite{FR12,FR13}.

In Section \ref{Subsection: Wick polynomials with derivatives} we discuss succinctly how to adapt Definitions \ref{Def: Wick polynomials}-\ref{Def: E-product} to the case of functionals which contain derivatives of the field configuration $\varphi$. This part of our work benefits from \cite{Khavkine-Melati-Moretti-17}, where Wick powers are thoroughly studied for tensor fields on globally hyperbolic spacetimes. In Section \ref{Subsection: Additional axioms: Leibniz rule and Principle of Perturbative Agreement} we outline instead the additional requirements to be added to the axioms for the Wick monomials, following the analysis for the Lorentzian counterpart in \cite{Hollands-Wald-05}. 

Since many statements and proofs are similar to those already discussed in the previous parts of this paper, we will limit ourselves to pointing out the main differences avoiding wherever possible unnecessary repetitions.

\subsection{Wick polynomials with derivatives}\label{Subsection: Wick polynomials with derivatives}
Goal of this section will be to extend Definition \ref{Def: Wick polynomials}-\ref{Def: E-product} to include also derivatives of the field configurations $\varphi\in\mathcal{E}(M)$. To this end, we need to generalize the structures considered in Section \ref{Section: General Setting}. Hence, for any smooth vector bundle $B\to M$ over $M$ we consider as kinematic configurations $\Gamma(B)$, the space of smooth sections of $B$. As a consequence, smooth, local functionals $\widehat{F}:\Gamma(B)\to\mathbb{C}$ are defined analogously to Definition \ref{Def: functionals}, with the difference that the $k$-th functional derivative $\widehat{F}^{(k)}[\alpha]\in\Gamma^\prime(S^{\boxtimes k}B)$, $k\in\mathbb{N}$ and $\alpha\in\Gamma(B)$.

According to Definition \ref{Def: functionals} if $F\in\mathcal{P}[M;h]$ then there exists $N\in\mathbb{N}$ such that $F^{(k)}=0$ for all $k\leq N$.
It follows that there exists $\omega^{(F)}:=\sum_{\ell=1}^N\omega^{(F)}_\ell$, where $\omega_\ell\in\Gamma_{\mathrm{c}}(S^{\boxtimes \ell}J_\infty^*(M)\otimes \wedge^{\mathrm{top}}T^*M)$, such that $F(\varphi)=\sum_{\ell=1}^N\int_M\langle j_\infty\varphi^{\otimes \ell},\omega^{(F)}_\ell\rangle$.
Here $\langle j_\infty\varphi^{\otimes \ell},\omega^{(F)}_\ell\rangle$ denotes the top-density on $M$ obtained contracting $\omega_\ell$ with $j_\infty\varphi^{\otimes \ell}$, where $j_\infty\varphi\in\Gamma(J_n(M))$ is the $\infty$-th jet extension of $\varphi$, $j_\infty\colon\mathcal{E}(M)\to\Gamma(J_\infty(M))$.
\begin{Definition}\label{Def: functional depending on derivative of field configuration}
	A smooth polynomial functional $F\colon \mathcal{E}(M)\to\mathbb{C}$ is said to depend on the derivatives of $\varphi$ up to order $n\in\mathbb{N}$ if the associated density-valued form $\omega^{(F)}=\sum_{\ell=1}^N\omega^{(F)}_\ell$ is such that $\omega^{(F)}_\ell\in\Gamma_{\mathrm{c}}(S^{\boxtimes \ell}J_{r_\ell}^*(M)\otimes \wedge^{\mathrm{top}}T^*M)$ for $r_\ell\leq n$.
\end{Definition}

\begin{remark}
	In the following we will denote with $\mathcal{P}_{\mathrm{reg}}[M;h]$ ({\it resp.} $\mathcal{P}_{\operatorname{loc}}[M;h]$) the space of smooth, regular ({\it resp.} smooth, local) polynomial functionals depending on the derivative of $\varphi$ up to an arbitrary but finite order. Recall that, thanks to \cite[Prop. 2.3.12]{Brunetti-Fredenhagen-Ribeiro-12} -- see also \cite{Brouder-Dang-Gengoux-Rejzner-18} -- all smooth, local, polynomial functionals depend on a finite number of derivatives of the field configuration $\varphi$.
\end{remark}

\begin{remark}
	We observe that the definitions of the functors $\Gamma_{\mathrm{eq}}\colon\bkgg\to\vect$ and $\mathcal{A}_{\operatorname{reg}}\colon\bkgg\to\alg$ generalize slavishly to the case of functionals which depend on the derivatives of the fields -- \textit{cf.} Definitions \ref{Def: locally covariant algebra of regular observables}-\ref{Def: functor Gamma}. In particular Theorem \ref{Theorem: locally covariant algebra of interest} holds true in this setting.
\end{remark}	

\vskip .2cm

Similarly to the case without derivatives -- \textit{cf.} Definition \ref{Def: locally covariant observable} -- a \emph{locally covariant observable} is a natural transformation from a functor of compactly supported sections $\Gamma_{\mathrm{c}}$ to an Euclidean locally covariant theory $\mathcal{A}$ -- \textit{cf.} Definition \ref{Def: locally covariant observable} and Remark \ref{Remark: generalized locally covariant observable}. We give thus the definition of Wick powers and Wick monomials along the same lines of Definitions \ref{Def: Wick polynomials}-\ref{Def: E-product}. In what follows $\Phi$ will denote always the local and covariant observable as in Example \ref{Example: local and covariant linear observable with derivatives}.
\begin{Definition}\label{Def: Wick powers with derivatives}
	We define a family of \emph{Wick powers}, associated with $\Phi$, as a collection of natural transformations $\lbrace\Phi^k\rbrace_{k\in\N}$, with $\Phi^k\vcentcolon\Gamma^{k}_\mathrm{c}\to\Gamma_{\operatorname{eq}}$, such that axioms (1),(2) and (5) of Definition \ref{Def: Wick polynomials} hold true and in addition
	\begin{enumerate}
		\item[(3)] $\forall k\in\N\cup\lbrace0\rbrace$, $(M;h)\in\obj(\bkgg)$, $\omega\in \Gamma_{\mathrm{c}}^{k}[M;h]$, $P\in\operatorname{Par}[M;h]$ and $\varphi_1,\varphi_2\in \mathcal{E}(M)$,
		\begin{align}\label{Eqn: derivative of Wick powers with derivative}
			\langle\Phi^k[M;h](\omega,P)^{(1)}[\varphi_1],\varphi_2\rangle=k\,\Phi^{k-1}[M;h](j_\infty\varphi_2\lrcorner\omega,P,\varphi_1),
		\end{align}
		where $j_\infty\varphi_2\lrcorner\omega\in\Gamma_{\mathrm{c}}^{k-1}[M;h]$ denotes the contraction between $j_\infty\varphi_2$ and $\omega$.

		\item[(4)] let $n\in\N$ and let $(M;h_s)\in\obj(\bkgg)$, with $\lbrace h_s\rbrace_{s\in\R^n}$ a smooth and compactly supported $n$-dimensional family of variations of $h$ as per Definition \ref{Def: smooth compactly supported d-dimensional family of variations}. For any smooth family $\lbrace P_s\rbrace_{s\in\R^n}$ with $P_s\in\mathrm{Par}(M;h_s)$ and for any $s\in\R^n$, let $\mathcal{U}_{k}\in\Gamma_\mathrm{c}^\prime(\pi_n^*S^{\otimes k}J_\infty^*(M))$ be the distribution on the pull-back bundle $\pi_n^*S^{\otimes k}J_\infty^*(M)$ with base space $\R^n\times M$ such that, for any $\omega\in\Gamma_{\mathrm{c}}^{k}[M;h]$, 
		\begin{align*}
			\mathcal{U}_{k}(\chi,\omega)\vcentcolon=\int_{\R^n}\mathrm{d}s\;\Phi^k[M;h_s](\omega,P_s,0)\chi(s)\,,\quad
			\forall\chi\in C^\infty_\mathrm{c}(\R^n).
		\end{align*}
		We require that, $\forall k\in\N$,
		\begin{align*}
			\mathrm{WF}(\mathcal{U}_{k})=\emptyset.
		\end{align*}
	\end{enumerate}
\end{Definition}

A straightforward generalization of Equation \eqref{Eqn: definition of Wick powers} provides an example of a family of Wick powers $\{\wick{\Phi^k}_H\}_k$ which satisfies Definition \ref{Def: Wick powers with derivatives}.
The results on existence and uniqueness of Wick powers can be read as the vector-valued generalization of Theorems \ref{Theorem: Unicity of Wick powers I}-\ref{Theorem: Unicity of Wick powers II}, see also \cite[Section 6]{Khavkine-Melati-Moretti-17}.
In particular Equation \eqref{Eqn: ambiguities for Wick powers} holds true.

\begin{Definition}\label{Def: Wick monomials with derivatives}
	Let $\lbrace\Phi^k\rbrace_{k\in\N}$ be a family of Wick powers associated with the quantum field $\Phi$, as per Definition \ref{Def: Wick powers with derivatives} and let $\text{\ul{k}}=(k_n)_{n}$ be a finite sequence of $\ell(\underline{\mathrm{k}})$ many non-negative integers.
	We call family of Wick monomials $\{\Phi^{\text{\ul{k}}}\}_{\text{\ul{k}}}$ associated with that of Wick powers $\{\Phi^k\}_k$ to be a collection of natural transformations $\Phi^{\text{\ul{k}}}\colon\Gamma_{\mathrm{c}}^{\text{\ul{k}}}\to\Gamma_{\mathrm{eq}}$, one for each sequence $\text{\ul{k}}$, with the following properties:
	\begin{enumerate}
		\item
		for every finite sequence $\text{\ul{k}}$, $\Phi^k\colon\Gamma_{\operatorname{c}}^{\text{\ul{k}}}\to\Gamma_{\operatorname{eq}}$ scales almost homogeneously with dimension $\sum_{i=1}^\ell k_i\mathrm{d}_\varphi$;
		
		\item
		for every finite sequence $\text{\ul{k}}$ we have $\Phi^{\text{\ul{k}}}=\Phi^{k_1}$;
		
		\item let $\text{\ul{k}}=(k_1,\ldots,k_\ell)$ be an arbitrary sequence of non-negative integers, $\ell\in\N$ and $(M;h)\in\obj(\bkgg)$, $\omega_{k_j}\in\Gamma_{\mathrm{c}}(S^{k_j}J_\infty^*M)$ for $j=1,\ldots,\ell$.
		Let $I\subsetneq\lbrace1,\dots,\ell\rbrace$ be a proper subset and denote with $I^\mathrm{c}$ the complement of $I$ with respect to $\lbrace1,\dots,\ell\rbrace$.
		We require that, if
		\begin{align*}
			\bigcup_{i\in I}\mathrm{supp}(\omega_{k_i})\cap\bigcup_{j\in I^\mathrm{c}}\mathrm{supp}(\omega_{k_j})=\emptyset,
		\end{align*}
		then
		\begin{align*}
			\Phi^{\text{\ul{k}}}[M;h](\omega_{k_1}\otimes\dots\otimes\omega_{k_\ell})=
			\Phi^{\text{\ul{k}}_I}[M;h]\bigg(\bigotimes_{i\in I}\omega_{k_i}\bigg)
			\cdot\Phi^{\text{\ul{k}}_{I^\mathrm{c}}}[M;h]\bigg(\bigotimes_{j\in I^\mathrm{c}}\omega_{k_j}\bigg)\,;
		\end{align*}
		
		\item
		for all finite sequence $\text{\ul{k}}=(k_1,\ldots,k_\ell)$, $(M;h)\in\obj(\bkgg)$, $\omega\in \Gamma_{\mathrm{c}}^{\text{\ul{k}}}[M;h]$, $P\in\operatorname{Par}[M;h]$ and $\varphi,\psi\in \mathcal{E}(M)$,
		\begin{align}\label{Eqn: derivative of Wick monomials with derivatives}
			\langle\Phi^{\text{\ul{k}}}[M;h](\omega,P)^{(1)}[\varphi],\psi\rangle=
			\sum_{j=1}^\ell k_j\Phi^{\widehat{\text{\ul{k}}}_j}[M;h](j_\infty\psi\lrcorner\omega,P,\varphi)\,,
		\end{align}
		where $j_\infty\psi\lrcorner\omega\in\Gamma_{\mathrm{c}}^{\widehat{\text{\ul{k}}}_j}[M;h]$ is the $j$-th contraction -- that is the contraction $\Gamma(J_\infty(M))\times\Gamma_{\mathrm{c}}(S^{k_j}J_\infty^*(M))\to\Gamma_{\mathrm{c}}(S^{k_j-1}J_\infty^*(M))$ -- of $j_\infty\psi\in\Gamma(J_\infty(M))$ with $\omega\in\Gamma_{\mathrm{c}}^{\text{\ul{k}}}[M;h]$ -- we set $\Phi^{\widehat{\underline{\mathrm{k}}}_j}=0$ whenever $k_j-1<0$.
	
		\item
		for all sequences $\text{\ul{k}}$ and $\ell\in\mathbb{N}$, let $(M;h_s)\in\obj(\bkgg)$ with $\lbrace h_s\rbrace_{s\in\R^n}$ smooth and compactly supported family of variations of $h$, as per Definition \ref{Def: smooth compactly supported d-dimensional family of variations}, for any smooth family $\lbrace P_s\rbrace_{s\in\R^n}$ with $P_s\in\mathrm{Par}(M;h_s)$ and for any $s\in\R^n$ (see Remark \ref{Remark: dependence on s of the parametrix}), let $\mathcal{U}_{\text{\ul{k}}}\in\Gamma_\mathrm{c}(\pi^*_n\boxtimes_{j=1}^\ell S^{ k_j}J_\infty^*(M))^\prime$ be the distribution on the pull-back bundle $\pi^*_n\boxtimes_{j=1}^\ell S^{ k_j}J_\infty^*(M)\to M\times\mathbb{R}^n$ such that, for any $\omega\in \Gamma_{\mathrm{eq}}^{\text{\ul{k}}}[M;h]$ and for any $\chi\in C^\infty_\mathrm{c}(\R^n)$,
		\begin{align*}
			\mathcal{U}_{\text{\ul{k}}}(\chi,\omega)\vcentcolon=
			\int_{\R^n}\mathrm{d}s\;\Phi^{\text{\ul{k}}}[M;h_s](\omega,P_s,0)\chi(s)\,.
		\end{align*}
		We require that the wave front set $\mathrm{WF}(\mathcal{U}_{\text{\ul{k}}})$ lies in 
		\begin{align*}
			\nonumber
			C_{(\ell)}(M)=&\bigg\lbrace(x_1,p_1;\dots;x_\ell,p_\ell,s;\tau)\in T^*(\pi^*_n\boxtimes_{j=1}^\ell S^{ k_j}J_\infty^*(M))^\ell\setminus\lbrace0\rbrace\,|\\&\exists\,I=\lbrace i_1,\dots, i_{|I|}\rbrace\subset\lbrace1,\dots,\ell\rbrace,
			|I|\geqslant2 : (x_{i_1},\dots, x_{i_{|I|}})\in\mathrm{Diag}(M^{|I|}), \sum_{i\in I}p_i=0\bigg\rbrace;
		\end{align*}
	\end{enumerate}
\end{Definition}
Following the same arguments of Proposition \ref{Prop: local and covariant algebra of local polynomials}, given a family of Wick powers and Wick monomials we can identify an Euclidean local and covariant field theory $\mathcal{A}$.
Moreover the uniqueness theorems \ref{Theorem: uniqueness of E-product I}-\ref{Theorem: uniqueness of E-product II} still hold true. In particular Equation \eqref{Eqn: ambiguities of E-product} is valid in this context.

\begin{remark}\label{Remark: interpretation of vector fields as element in the dual of the jet bundle}
	Notice that any multivector field $\omega_k\in\Gamma_{\mathrm{c}}(S^kTM)$ leads to a unique section $\widehat{\omega}_k\in\Gamma_{\mathrm{c}}(S^kJ_1(M))\subseteq\Gamma_{\mathrm{c}}^{k}[M;h]$ defined by
	$\langle\widehat{\omega}_k,j_\infty\varphi[\otimes]\ldots[\otimes] j_\infty\varphi\rangle\vcentcolon=\langle\omega_k,\mathrm{d}\varphi[\otimes]\ldots[\otimes] \mathrm{d}\varphi\rangle$.
	For $k=0$ we recover the identification between $\mathcal{D}(M)$ and $\Gamma_{\mathrm{c}}^{0}[M;h]$.
	In the following we identify $\omega_k$ and $\widehat{\omega}_k$.
	Similarly a multivector field $\Omega\vcentcolon=\omega_{k_1}\otimes\ldots\otimes\omega_{k_\ell}\in\Gamma_{\mathrm{c}}(\boxtimes_{j=1}^\ell S^{k_j} TM)$ identifies a unique $\widehat{\Omega}\in\Gamma_{\mathrm{c}}(\boxtimes_{j=1}^\ell S^{k_j}J_1(M))\subset\Gamma_{\mathrm{c}}^{(k_1,\ldots,k_\ell)}[M,h]$.
\end{remark}

\subsection{Additional axioms: Leibniz rule and Principle of Perturbative Agreement}\label{Subsection: Additional axioms: Leibniz rule and Principle of Perturbative Agreement}
In this section we discuss two additional requirements which provide further structural constraints to Wick monomials: the Leibniz rule and the principle of perturbative agreement (PPA). These axioms have been introduced in \cite{Hollands-Wald-05} -- see also \cite{Drago-Hack-Pinamonti-2016,Zahn:2013ywa} -- as a requirement for internal consistency of Wick monomials. The PPA in particular is necessary to ensure that any term in the Lagrangian, which has a quadratic dependence on the fields of the underlying theory, can be equivalently included in the free or in the interacting part of the Lagrangian without changing the prediction of the model.

From a technical point of view, these new axioms have the merit of further restricting the ambiguities present in the definition of Wick powers and of Wick monomials. For this reason this prompts the question whether there exists a family of Wick powers and of Wick monomials, adhering to Definitions \ref{Def: Wick powers with derivatives} and \ref{Def: Wick monomials with derivatives}, which satisfies all axioms. Similarly, the proofs of Theorem \ref{Theorem: Unicity of Wick powers I} and \ref{Theorem: Unicity of Wick powers II} are no longer valid slavishly and they should be generalized to the case in hand. Luckily, these problems have been already tackled in \cite{Hollands-Wald-05} in the Lorentzian case and this allows us to avoid giving all the details, highlighting instead the main differences between Riemannian and Lorentzian theories.

We divide the analysis in two steps. In the first we state the so-called Leibniz rule and our  main result in this direction is contained in Proposition \ref{Prop: existence of Wick monomials with Leibniz rule}. Herein we show that there exists always a prescription of Wick monomials which satisfies both Definition \ref{Def: Wick monomials with derivatives} and \ref{Def: Leibniz rule}. In the second step, instead, we formulate the PPA and we investigate its implications, which are discussed mainly in Theorem \ref{Theorem: existence of Wick monomials with satisfy Leibniz rule and PPA}.

\subsubsection{Leibniz rule}\label{Subsubsection: Leibniz rule}
Definitions \ref{Def: Wick powers with derivatives}-\ref{Def: Wick monomials with derivatives} establish a list of properties on the families of Wick powers $\Phi^k$ and of Wick monomials $\Phi^{\text{\ul{k}}}$. Yet, there is no condition which links together polynomial expressions of the fields which are not functionally independent. As an example, consider the family of Wick powers $\{\wick{\Phi^k}_H\}_k$ defined in Section \ref{Section: Existence of Wick powers}.
Let $(M;h)\in\mathsf{Obj}(\bkgg)$ and let $\omega_2\in\Gamma_{\mathrm{c}}^{2}[M;h]$ -- \textit{cf.} equation \eqref{Eqn: compactly supported section bundle of fixed order}.
Moreover let $X\in\Gamma(TM)$ and consider $\operatorname{div}(\omega_2\otimes X)=\nabla_X\omega_2+\operatorname{div}(X)\omega_2\in\Gamma_{\operatorname{c}}^{2}[M;h]$.
Setting $\psi:=j_\infty\varphi$, it reads locally $\langle j_\infty\varphi^{[\otimes]2},\operatorname{div}(\omega_2\otimes X)\rangle=\psi^\alpha\psi^\beta\nabla_a(X^a(\omega_{2})_{\alpha\beta})$.
A direct computation gives
\begin{align*}
	\wick{\Phi^2}_H[M;h](-\operatorname{div}_g(\omega_2 \otimes X),P,\varphi)&=
	-\exp\big[\Upsilon_{W_P}\big]\int_M\mu_g\langle j_\infty\varphi^{[\otimes]2},\operatorname{div}_g(\omega_2 \otimes X)\rangle\\&=
	-\int_M\mu_g\,\langle j_\infty\varphi^{[\otimes 2]}+
	\textrm{$j_\infty$}
	[W_P],\operatorname{div}_g(\omega_2 \otimes X)\rangle\,.
\end{align*}
where $P\in\operatorname{Par}[M;h]$ and $\varphi\in \mathcal{E}(M)$.
Applying Stokes' theorem, one obtains
\begin{align}
	\nonumber
	\wick{\Phi^2}_H[M;h](-\operatorname{div}_g(\omega_2 \otimes X),P,\varphi)&=
	\int_M\mu_g\,\langle2j_\infty\varphi[\otimes]\nabla_Xj_\infty\varphi+2
	\textrm{$j_\infty$}
	[\nabla^{(1)}_XW_P],\omega_2\rangle\\\nonumber &=
	\exp\big[\Upsilon_{W_P}\big]\int_M\mu_g\langle 2j_\infty\varphi[\otimes]\nabla_Xj_\infty\varphi,\omega_2\rangle=
	2\Phi[M;h](\nabla_X\varphi\lrcorner \omega_2,P,\varphi)\\&=
	\label{Eqn: preliminary example for Leibniz rule}
	\big\langle\Phi^2[M;h](\omega_2,P)^{(1)}[\varphi],\nabla_X\varphi\big\rangle\,,
\end{align}
where the last equality is a consequence of Equation \eqref{Eqn: derivative of Wick powers with derivative}.
With a slight abuse of notation we denoted with $\nabla_Xj_\infty\varphi$ the covariant derivative along $X$ of $j_\infty\varphi$ with respect to the unique connection obtained by lifting to $J_\infty(M)$ that of Levi-Civita over $M$.
The symbol $\nabla^{(1)}\vcentcolon=\nabla\otimes\mathrm{Id}$ denotes the Levi-Civita connection acting on the first base point of $W_P$. Since $W_P$ is symmetric, it holds $\nabla_X[W_P]=2[\nabla_X^{(1)}W_P]$.

From Equation \eqref{Eqn: preliminary example for Leibniz rule}, one can infer that the Wick ordered expressions $\wick{\varphi^2}_H$ and $\wick{\varphi\nabla_a\varphi}_H$ are not independent, rather $\nabla_a\wick{\varphi^2}_H=2\wick{\varphi\nabla_a\varphi}_H$.
On account of Theorem \ref{Theorem: Unicity of Wick powers I} this constraint may not be implemented in a general family of Wick powers $\{\Phi^k\}_k$. The Leibniz rule discards these scenarios.
\begin{Definition}[Leibniz rule]\label{Def: Leibniz rule}
	A family of Wick powers $\{\Phi^k\}_{k\in\mathbb{N}}$ is said to satisfy the {\em Leibniz rule} if, for all $(M,h)\in\mathsf{Obj}(\bkgg)$, $\omega_k\in\Gamma_{\mathrm{c}}^{k}[M;h]$, $X\in\Gamma(TM)$ it holds
	\begin{align}\label{Eqn: Leibniz rule for Wick powers}
		\Phi^k[M;h](-\operatorname{div_g}(\omega_k\otimes X),P,\varphi)=
		\big\langle\Phi^k[M,h](\omega_k,P)^{(1)}[\varphi],\nabla_X\varphi\big\rangle\,,
	\end{align}
	for all $P\in\operatorname{Par}[M;h]$ and for all $\varphi\in \mathcal{E}(M)$.
	Here $\operatorname{div}_g(\omega_k\otimes X)\vcentcolon=\nabla_X\omega_k+\operatorname{div}_g(X)\omega_k\in\Gamma_{c}^{k}[M;h]$.
	Similarly, a family of Wick monomial $\{\Phi^{\text{\ul{k}}}\}_{\text{\ul{k}}}$ is said to satisfy the {\em Leibniz rule} if, for all $(M;h)\in\mathsf{Obj}(\bkgg)$, $\ell\in\mathbb{N}$, $\underline{\mathrm{k}}=(k_1,\ldots,k_\ell)\in\mathbb{Z}_+^\ell$, $\omega_{k_j}\in\Gamma_{\mathrm{c}}^{k_j}[M;h]$ for $j\in\{1,\ldots,\ell\}$, $X\in\Gamma(TM)$ it holds
	\begin{multline}\label{Eqn: Leibniz rule for Wick monomials}
		\Phi^{\text{\ul{k}}}[M;h](\omega_{k_1}\otimes\ldots\otimes-\operatorname{div}_g(X\otimes\omega_{k_j})\otimes\ldots\otimes\omega_{k_\ell},P,\varphi)\\
		=\big\langle\Phi^{\text{\ul{k}}}[M;h](\omega_{k_1}\otimes\ldots\otimes\omega_{k_\ell},P)^{(1)}[\varphi],\nabla_X\varphi\big\rangle\,.
	\end{multline}
	for all $P\in\operatorname{Par}[M;h]$ and for all $\varphi\in \mathcal{E}(M)$.
\end{Definition}
A straightforward application of Equation \eqref{Eqn: preliminary example for Leibniz rule} shows that the family of Wick powers $\{\wick{\Phi^k}_H\}_k$ satisfies Definition \ref{Def: Leibniz rule}.
Nevertheless the construction discussed in Section \ref{Section: Existence E-product of Wick of Wick powers} is less explicit and, in principle, we should repeat the whole argument in order to show that, at each order in the iterative process, we can adjust the construction so that the corresponding family of Wick monomials $\{\Phi^{\text{\ul{k}}}\}_{\text{\ul{k}}}$ satisfies the Leibniz rule as per Definition \ref{Def: Leibniz rule}. Yet the same problem in the Lorentzian case has been tackled in \cite[Prop. 3.1]{Hollands-Wald-05} and, since switching to the Riemannian case, lead to no changes, we omit the proof.
\begin{proposition}\label{Prop: existence of Wick monomials with Leibniz rule}
	Let $\{\Phi^k\}_k$ be a family of Wick powers which satisfies the Leibniz rule as per Definition \ref{Def: Leibniz rule}.
	Then there exists a family of Wick monomials $\{\Phi^{\text{\ul{k}}}\}_{\text{\ul{k}}}$ associated to $\{\Phi^k\}_k$ which satisfies the Leibniz rule as well.
\end{proposition}

\noindent To conclude we focus on the extension of Theorem \ref{Theorem: Unicity of Wick powers II}.
\begin{proposition}\label{Prop: constraint on Wick monomials ambiguities due to Leibniz rule}
	Let $\{\Phi^{\text{\ul{k}}}\}_{\text{\ul{k}}}$ be a family of Wick monomials which satisfies the Leibniz rule -- \textit{cf.} Definitions \ref{Def: Wick monomials with derivatives}-\ref{Def: Leibniz rule}.
	Let $c_{k_I-j_I}[M;h]\in\Gamma^{k_I-j_I}[M,h]$ be the tensor coefficients introduced in Theorems \ref{Theorem: uniqueness of E-product I}-\ref{Theorem: uniqueness of E-product II}.
	Then $c_{k_I-j_I}[M;h]$ is covariantly constant, that is $\nabla c_{k_I-j_I}[M;h]=0$.
\end{proposition}
\begin{proof}
	The proof goes by induction with respect to the indices $\ell, k_1,\ldots,k_\ell$ appearing in Equation \eqref{Eqn: ambiguities of E-product}. For simplicity in the notation we consider the case $\ell=1$, $k=2$, all others following suit. Equation \eqref{Eqn: ambiguities of E-product} reduces to \eqref{Eqn: ambiguities for Wick powers}, namely
	\begin{align*}
		\Phi^2[M;h](\omega)=\wick{\Phi^2}_H[M;h](\omega)+C_2[M;h](\omega)\,,
	\end{align*}
	for all $\omega\in\Gamma_{\operatorname{c}}^{k}[M;h]$, where $C_2[M;h](\omega)=\int_M\mu_g\, c_2[M,h]\lrcorner\omega$, being $c_2[M,h]\in\Gamma^{k}[M;h]$.
	Imposing the Leibniz rule \eqref{Eqn: Leibniz rule for Wick powers} and using Equation \eqref{Eqn: derivative of Wick powers with derivative} as well as $\wick{\Phi^1}=\Phi^1=\Phi$ we find that, for all $\omega\in\Gamma_{\mathrm{c}}^{k}[M,h]$ and for all $X\in\Gamma(TM)$,
	\begin{align*}
		0=C_2[M;h](\operatorname{div}_g(X\otimes\omega))=
		\int_M\mu_gc_2[M;h]\lrcorner\operatorname{div}_g(X\otimes\omega)=-
		\int_M\mu_g\nabla_X(c_2[M,h])\lrcorner\omega\,.
	\end{align*}
	Since $\omega$ is arbitrary, it descends $\nabla_Xc_2[M,h]$ for all $X$, that is $\nabla c_2[M;h]=0$, which is the sought statement.
\end{proof}

\subsubsection{Principle of Perturbative Agreement}\label{Subsubsection: PPA}
The second axiom we impose in addition to those in Definition \ref{Def: Wick powers with derivatives} and \ref{Def: Wick monomials with derivatives} goes under the name of principle of perturbative agreement (PPA).
This has been introduced in \cite{Hollands-Wald-05}, see also \cite{Drago-Hack-Pinamonti-2016,Zahn:2013ywa} and it is essential to guarantee that, in the construction of the algebra of Wick polynomials, one can include equivalently any term in the Lagrangian, which has a quadratic dependence on the underlying fields, either in the free or in the interacting part of the Lagrangian.

The original formulation of the PPA on Lorentzian backgrounds exploits the perturbative approach to interacting field theories -- \textit{cf.} \cite{Hollands-Wald-05}.
The same formulation in the Riemannian setting is not immediately available because, as we shall see in Section \ref{Section: Interacting models}, the formulation of the perturbative approach to interacting theories seems to require additional structures.
Nevertheless in \cite{Drago-Hack-Pinamonti-2016} an equivalent formulation to the PPA has been given and this turns out to be more suitable to be adapted to the Riemannian setting. In this framework the PPA becomes a natural requirement which strengthens the covariance axiom -- \textit{cf.} Definition \ref{Def: Euclidean locally covariant theory}.

In particular, let $(M,h),(M,h_s)\in\mathsf{Obj}(\bkgg)$ be such that $\{h_s\}_{s\in\mathbb{R}}$ is a smooth compactly supported family of variations of $h$ as per definition \ref{Def: smooth compactly supported d-dimensional family of variations}. In this situation we may consider the algebras $\mathcal{A}[M;h],\mathcal{A}[M;h_s]$ as per definition \ref{Def: locally covariant algebra of regular observables} and proposition \ref{Prop: local and covariant algebra of local polynomials}.
For perturbations $(M,h_s)$ of $(M,h)$ arising from a diffeomorphism $\chi\colon M\to M$ the requirement of covariance on $\mathcal{A}$ -- \textit{cf.} Definition \ref{Def: Euclidean locally covariant theory} -- yields a $\ast$-isomorphism between $\mathcal{A}[M;h]$ and $\mathcal{A}[M;h_s]$. Heuristically speaking, the PPA requires that a similar $\ast$-isomorphism exists also in the case of an arbitrary compactly supported perturbation $(M,h_s)$ of $(M,h)$ -- \textit{cf.} Definition \ref{Def: PPA}.
This implies in particular that, whenever the ambiguities in defining the algebra $\mathcal{A}[M,h]$  have been fixed -- \textit{cf.} Proposition \ref{Prop: local and covariant algebra of local polynomials} -- the same happens for those arising in the definition of $\mathcal{A}[M,h_s]$. This is a rather strong requirement because the Hadamard parametrices $H_s$, $H$ associated with the elliptic operators $E_s$ and $E$ are different -- \textit{cf.} Remark \ref{Remark: Hadamard representation}. Therefore, the PPA cannot be imposed naively, meaning that it is not possible to compare directly the algebras $\mathcal{A}[M;h]$ and $\mathcal{A}[M;h_s]$. On the contrary one has to consider a Taylor expansion in $s$ of $h_s$, regarding the parameter $s$ as formal and the PPA can be formulated as a requirement on the Wick monomial $\Phi^{\text{\ul{k}}}[M;h]$ which generate $\mathcal{A}[M;h]$ -- \textit{cf.} Definition \ref{Def: PPA}.

In the following we discuss the PPA for a scalar field theory on a Riemannian manifold along the lines of \cite{Hollands-Wald-05, Drago-Hack-Pinamonti-2016}.

\paragraph{The PPA for the regular algebra.}

For definiteness, let $(M,h)\in\mathsf{Obj}(\bkgg)$. In the following $h_s$, $s\in\mathbb{R}$ denotes a smooth and compactly supported family of variations of $h$ -- \textit{cf.} Definition \ref{Def: smooth compactly supported d-dimensional family of variations}.

The PPA calls for a comparison between the algebras $\mathcal{A}[M;h]$, $\mathcal{A}[M,h_s]$.
To this end, let us start from the regular counterpart, $\mathcal{A}_{\mathrm{reg}}[M,h]$, $\mathcal{A}_{\mathrm{reg}}[M,h_s]$ respectively.
The spaces of parametrices $\operatorname{Par}[M;h]$ and $\operatorname{Par}[M;h_s]$ associated with $E$ and $E_s$ turn out to be isomorphic.
This is a consequence of the following result, which is the Euclidean counterpart of a well-known construction in the Lorentzian framework -- \textit{cf.} \cite{Dappiaggi-Drago-16,Drago-Hack-Pinamonti-2016,Hollands-Wald-05}.
\begin{proposition}\label{Prop: 1-1 correspondence between E-parametrices and Eg-parametrices}
	There exists an isomorphism $R_s\colon\operatorname{Par}[M;h]\to\operatorname{Par}[M;h_s]$ of affine spaces between $\operatorname{Par}[M;h]$ and $\operatorname{Par}[M;h_s]$.
\end{proposition}
\begin{proof}
	Let $\widehat{P}\in\operatorname{Par}[M;h]$ and $\widehat{P}_s\in\operatorname{Par}[M;h_s]$.
	We define an isomorphism $R_s\colon\operatorname{Par}[M;h]\to\operatorname{Par}[M;h_s]$ by setting $R_sP:=\widehat{P}_s+(P-\widehat{P})$.
	Since $P,\widehat{P}\in\operatorname{Par}[M;h]$ we have $P-\widehat{P}\in\mathcal{E}(M\times M)$, therefore $R_sP\in\operatorname{Par}[M;h_s]$.
	Moreover $R_s$ is injective because, $\forall P,Q\in\operatorname{Par}[M;h]$ such that $R_sP=R_sQ$, it holds $\widehat{P}_s+P-\widehat{P}=\widehat{P}_s+Q-\widehat{P}$.
	Furthermore $R_s$ is surjective. For all $P_s\in\operatorname{Par}[M;h_s]$, one can write $P:=\widehat{P}+(P_s-\widehat{P}_s)\in\operatorname{Par}[M;h]$ and $R_sP=P_s$.
\end{proof}

On account of Remark \ref{Remark: dependence on s of the parametrix} we can choose $\widehat{P}_s$ in the proof of Proposition \ref{Prop: 1-1 correspondence between E-parametrices and Eg-parametrices} so that $\{P_s\}_s$ is a smooth family of parametrices -- \textit{cf.} Remark \ref{Remark: dependence on s of the parametrix}. This entails analogous smoothness properties for the map $R_s$.
Henceforth we shall implicitly choose $P_s$ smoothly dependent from $s$.

Due to the regularity of their elements, the algebras $\mathcal{A}_{\mathrm{reg}}[M;h]$ and $\mathcal{A}_{\mathrm{reg}}[M;h_s]$ are $\ast$-isomorphic as we establish in the following proposition.
\begin{proposition}\label{Prop: PPA for regular algebras}
	The algebras $\mathcal{A}_{\mathrm{reg}}[M;h]$, $\mathcal{A}_{\mathrm{reg}}[M,h_s]$ are $\ast$-isomorphic, the $\ast$-isomorphism being realized by $\beta_s\colon\mathcal{A}_{\mathrm{reg}}[M;h]\to\mathcal{A}_{\mathrm{reg}}[M,h_s]$ where, for all $F\in\mathcal{A}_{\mathrm{reg}}[M;h]$ and for all $P\in\operatorname{Par}[M;h]$,
	\begin{align}\label{Eqn: PPA maps for regular functionals}
		(\beta_s F)[P_s]=\exp\big[\Upsilon_{P_s-P}\big]F[P]\,,
	\end{align}
	where $P_s=R_sP$ has been defined in Proposition \ref{Prop: 1-1 correspondence between E-parametrices and Eg-parametrices}.
\end{proposition}
\begin{proof}
	For all $F\in\mathcal{A}_{\mathrm{reg}}[M;h]$, the functional $(\beta_sF)(P_s)$ introduced in \eqref{Eqn: PPA maps for regular functionals} is well-defined on account of the regularity of $F$. In addition the map $P_s\mapsto(\beta_sF)[P_s]$ is equivariant.
	As a matter of fact, for $P_s=R_sP,Q_s=R_sQ\in\operatorname{Par}[M;h_s]$ we have
	\begin{align*}
		\alpha_{P_s}^{Q_s}(\beta_sF)(Q_s)=
		\exp\big[\Upsilon_{P_s-Q_s}\big]\exp\big[\Upsilon_{Q_s-Q}\big]F[Q]&=
		\exp\big[\Upsilon_{P_s-P+P-Q}\big]F[Q]\\&=
		\exp\big[\Upsilon_{P_s-P}\big]\alpha^Q_PF[Q]=
		(\beta_sF)(P_s)\,,
	\end{align*}
	where in the last equality we used the equivariance property of $F$, namely $F[P]=\alpha^Q_PF[Q]$ -- \textit{cf.} Definition \ref{Def: locally covariant algebra of regular observables}.
	The $\ast$-isomorphism can be proven adapting to the case in hand the analysis of Proposition \ref{Prop: star-isomorphism between different algebras}.
\end{proof}

\begin{remark}\label{Remark: PPA and regular local and covariant observable}
	Since $\beta_s$ is a $\ast$-isomorphism between $\mathcal{A}_{\mathrm{reg}}[M;h]$ and $\mathcal{A}_{\mathrm{reg}}[M;h_s]$ one may wonder whether it preserves local and covariant observables as per Definition \ref{Def: locally covariant observable}.
	This holds true in the following sense. Let $\mathcal{O}$ be a local and covariant observable -- \textit{cf.} Definition \ref{Def: locally covariant observable} -- such that $\mathcal{O}[M;h]\in\mathcal{A}_{\mathrm{reg}}[M;h]$ for all $(M;h)\in\mathsf{Obj}(\bkgg)$. A canonical example is $\mathcal{O}=\Phi$ as defined in Example \ref{Example: quantum fields} and \ref{Example: local and covariant linear observable with derivatives}. Considering the  same setting of Proposition \ref{Prop: 1-1 correspondence between E-parametrices and Eg-parametrices} an explicit computation yields
	\begin{align}\label{Eqn: PPA and regular local and covariant observable}
		\mathcal{O}[M;h_s](\omega_{(s)},P_s)=\beta_s\big[\mathcal{O}[M;h](\omega)\big](P_s)=\exp\big[\Upsilon_{P_s-P}\big]\mathcal{O}[M;h](\omega,P)\,,
	\end{align}
	where $\omega,\omega_{(s)}\in\Gamma_{\mathrm{c}}^{1}[M;h]$ are such that $\omega_{(s)}\mu_{g_s}=\omega\mu_g$. Thus, up to the change of volume measure, $\beta_s$ preserves local and covariant, regular observables.
	This example suggests also that one might consider working directly with densitized observables so to account for the change in the volume measure.
\end{remark}

\paragraph{The PPA for the full algebra.}

According to Proposition \ref{Prop: PPA for regular algebras}, $\beta_s$ is a $\ast$-isomorphism between $\mathcal{A}_{\mathrm{reg}}[M;h]$ and $\mathcal{A}_{\mathrm{reg}}[M;h_s]$, but it does not lift to a counterpart between $\mathcal{A}[M;h]$ and $\mathcal{A}[M,h_s]$.
This can be realized by a close scrutiny of the local Hadamard representation of $P$ and of $P_s=R_sP$ which shows that $P_s-P$ is not smooth, \textit{cf.} Remark \ref{Remark: Hadamard representation}.
Therefore $\Upsilon_{P_s-P}$ cannot be applied to a local and polynomial functional unless it lies $\mathcal{P}_{\mathrm{loc}}[M;h]\cap\mathcal{P}_{\mathrm{reg}}[M;h]$. Consequently $\beta_s$ cannot be lifted to $\mathcal{A}[M;h]$ or to $\mathcal{A}[M,h_s]$.

Notwithstanding, we can still require Equation \eqref{Eqn: PPA and regular local and covariant observable} to hold  true for $\mathcal{O}=\Phi^{\text{\ul{k}}}$.
Since this cannot be achieved exactly, the strategy is to expand Equation \eqref{Eqn: PPA and regular local and covariant observable} as a formal power series in $s$.
This leads to a hierarchy of equations which constraint $\Phi^{\text{\ul{k}}}$ -- \textit{cf.} Definition \ref{Def: PPA}.

To follow this line of thought, we need to prove that the expansion of $\beta_s$ as a perturbative series in $s$ is well-defined as a map $\beta_s\colon\mathcal{A}[M;h]\to\Gamma(\mathsf{E}[M;h])[[s]]$.
Here $\Gamma(\mathsf{E}[M;h])[[s]]$ denotes the $\ast$-algebra of formal power series in $s$ with coefficients in $\Gamma(\mathsf{E}[M;h])$ -- \textit{cf.} definitions \ref{Def: bundle over parametrices}-\ref{Def: locally covariant algebra of regular observables}.

\begin{proposition}\label{Prop: well-definitenss of perturbative expansion of beta}
	Let $\beta_s\colon\colon\mathcal{A}_{\mathrm{reg}}[M;h]\to\Gamma_{\mathrm{reg}}(\mathsf{E}[M;h])[[s]]$ be the linear operator obtained by expanding $\beta_s\colon\mathcal{A}_{\mathrm{reg}}[M;h]\to\mathcal{A}[M;h_s]$ as formal power series in $s$.
	Then the map $\beta_s$ can be extended to a counterpart, still denoted with $\beta_s$, from $\mathcal{A}[M;h]$ to $\Gamma(\mathsf{E}[M;h])[[s]]$.
\end{proposition}
\begin{proof}
	Since we are interested in proving that the expansion in formal power series in $s$ of $\beta_s$ is well-defined we can discard smooth contributions from the expansion in $s$, focusing only on the singular contributions.
	This procedure will yield a map $\beta_{[[s]]}$ such that $\beta_{s}-\beta_{[[s]]}$ is smooth at all orders in $s$.
	We shall discuss, moreover, whether $\beta_{[[s]]}$ and thus also $\beta_{s}$ are extensible.
	
	We expand each parametrix $P_s=R_sP$ of $E_s$ as a formal power series in $s$ built out of the corresponding counterpart $P$ of $E$.
	Observe that, by Definition \ref{Def: smooth compactly supported d-dimensional family of variations}, $h_s-h=O(s)$ as well as $G_s\vcentcolon=E_s-E=O(s)$.	Since $G_s$ is compactly supported, we may write
	\begin{align*}
	E_s=E+G_s=E(I+PG_s)-S_PG_s\,,
	\end{align*}
	where $S_P\colon\mathcal{D}(M)\to\mathcal{E}(M)$ is such that $EP=\mathrm{Id}_{\mathcal{D}(M)}+S_P$.
	Let us consider the operator $P_{[[s]]}\colon\mathcal{D}(M)\to\mathcal{E}(M)[[s]]$ defined by
	\begin{align}\label{Eqn: Moeller operator for parametrices}
		P_{[[s]]}\vcentcolon=\sum_{n\geq 0}(-PG_s)^nP\,.
	\end{align}
	A direct computation shows that $P_{[[s]]}$ satisfies, at each perturbative order in $G_s=O(s)$, $P_{[[s]]}E_sf=E_sP_{[[s]]}f=f+S_{[[s]]}f$ for all $f\in\mathcal{D}(M)$, where $S_{[[s]]}$ is a perturbative, smoothing remainder.
	Moreover $P_{[[s]]}$ is formally symmetric because
	\begin{align*}
	P_{[[s]]}^*=P^*\sum_{n\geq 0}(-G_s^*P^*)^n=P_{[[s]]}\,,
	\end{align*}
	where we exploited that $P=P^*$ and $G_s=G_s^*$.
	It follows that $P_{[[s]]}-P_{s}$ is smooth at each perturbative order in $s$. Therefore, we may consider the formal map $P_{[[s]]}$ as the perturbative expansion in $s$ of $P_s$ up to a smooth remainder.
	
	The perturbative expansion of $\beta_s$ up to a smooth remainder is obtained by replacing $P_s$ with $P_{[[s]]}$.
	In particular it holds that, for all $F\in\mathcal{A}_{\mathrm{reg}}[M;h]$ and $P\in\operatorname{Par}[M;h]$,
	\begin{align*}
	(\beta_{[[s]]}F)(P_{[[s]]})=\exp\big[\Upsilon_{P_{[[s]]}-P}\big]F[P]\,,
	\end{align*}
	so that $\beta_{[[s]]}F\in\Gamma(\mathsf{E}[M;h])[[s]]$.
	To make $\beta_{[[s]]}$ and thus $\beta_s$ well-defined on $\mathcal{A}[M;h]$ we need to study the coinciding point limit of $P_{[[s]]}-P$.
	This amounts to observe that
	\begin{align*}
	P_{[[s]]}-P=\sum_{n\geq 1}(-PG_s)^nP\,,
	\end{align*}
	so that the coinciding point limit is well-defined at each order in $s$ provided that the renormalization freedoms of $(PG_s)^nP$ have been accounted for. However, this is a consequence of the construction of the algebra $\mathcal{A}[M;h]$. Therefore, $\beta_s\colon\mathcal{A}[M;h]\to\Gamma(\mathsf{E}[M;h])[[s]]$ is well-defined.
\end{proof}

\begin{remark}
	As observed in \cite[Remark 3.26]{Drago-Hack-Pinamonti-2016}, the map $\beta_s$ requires to be renormalized at a perturbative level in each order in $s$.
	This may seem unsatisfactory at first glance; however, it can be shown that, in particular circumstances, for each $F\in\mathcal{A}[M;h]$, one needs to renormalize a finite number of terms of the form $(PG_s)^nP$.
	Indeed, observe that, since $F\in\mathcal{A}[M;h]$ is a polynomial functional, the exponential series which defines $\exp\big[\Upsilon_{P_{[[s]]}-P}\big]F[P]$ is finite.
	Moreover, notice that $G_s=E_s-E$ is a differential operator of degree $\deg(G_s)\leq 2$ with smooth compactly supported coefficients.
	Then for each $n\geq 1$, the distribution $(PG_s)^nP$ acts on a $n\mathrm{D}$-dimensional space with scaling degree 
	\begin{align}
		(\mathrm{D}-2)(n+1)+n\deg(G_s)=n\mathrm{D}+(\deg(G)-2)n+\mathrm{D}-2\,.
	\end{align}
	It descends that, if $h_s=(g,A_s,c_s)$ does not involve a variation of the background metric, then $\deg(G_s)\leq 1$ and the scaling degree is strictly lower than $n\mathrm{D}$ for $n>\mathrm{D}-2$.
	In this case we may apply \cite[Thm. 5.2]{Brunetti:1999jn} to conclude that the distribution $(PG_s)^nP$ has a unique extension to the whole space. Hence there are no renormalization ambiguities when $n$ is large enough.
\end{remark}

\begin{Definition}[PPA]\label{Def: PPA}
	A family of Wick monomials $\{\Phi^{\text{\ul{k}}}\}_{\text{\ul{k}}}$ is said to satisfy the principle of perturbative agreement (PPA) if, for all $(M;h)\in\mathsf{Obj}(\bkgg)$ and for any smooth compactly supported family of variations $\{h_s\}_{s\in\mathbb{R}}$ of $h$ -- \textit{cf.} Definition \ref{Def: smooth compactly supported d-dimensional family of variations} --, it holds
	\begin{align}\label{Eqn: PPA}
		\frac{\mathrm{d}^n}{\mathrm{d}s^n}\Phi^{\text{\ul{k}}}[M;h_s](\omega_{(s)},P_s)\bigg|_{s=0}
		=\frac{\mathrm{d}^n}{\mathrm{d}s^n}\beta_{s}\big[\Phi^{\text{\ul{k}}}[M;h](\omega)\big](P_s)\bigg|_{s=0}\,,
	\end{align}
	where $\omega,\omega_{(s)}\in\Gamma_{\mathrm{c}}^{\text{\ul{k}},\ell}[M;h]$ are such that $\omega_{(s)}\mu_{g_s}^{\otimes\ell}=\omega\mu_g^{\otimes \ell}$.
	Notice that the right-hand side of equation \eqref{Eqn: PPA} is well-defined on account of proposition \ref{Prop: well-definitenss of perturbative expansion of beta}.
\end{Definition}

\begin{remark}\label{Remark: comments on PPA}
	As observed in \cite{Hollands-Wald-05}, the PPA is satisfied for all $n\in\mathbb{N}$ whenever it holds true for $n=1$.
	This can be proved by induction. Let us assume that equation \eqref{Eqn: PPA} holds for all $n\leq k$ where $k\geq 1$. We can prove that equation \eqref{Eqn: PPA} is verified for $n=k$. To begin with we observe that
	\begin{align*}
		\frac{\mathrm{d}^k}{\mathrm{d}s^k}\Phi^k[M;h_s](\omega_s,P_s)\bigg|_{s=0}&=
		\frac{\mathrm{d}}{\mathrm{d}t}\frac{\mathrm{d}^{k-1}}{\mathrm{d}s^{k-1}}
		\Phi^k[M;h_{t+s}](\omega_{t+s},P_{t+s})\bigg|_{t,s=0}\\&=
		\frac{\mathrm{d}}{\mathrm{d}t}\frac{\mathrm{d}^{k-1}}{\mathrm{d}s^{k-1}}
		\beta_s\big[\Phi^k[M;h_t](\omega_t)\big](P_{t+s})\bigg|_{t,s=0}\,,
	\end{align*}
	where in the last equality we used equation \eqref{Eqn: PPA} for $n=k-1$. Notice that, at this point, $h_{t+s}$ has been regarded as a smooth, compactly supported $1$-dimensional family of variations of $h_t$. In view of the definition of $\beta_s$ we find
	\begin{align*}
		\frac{\mathrm{d}}{\mathrm{d}t}\frac{\mathrm{d}^{k-1}}{\mathrm{d}s^{k-1}}
		\beta_s\big[\Phi^k[M;h_t](\omega_t)\big](P_{t+s})\bigg|_{t,s=0}&=
		\frac{\mathrm{d}}{\mathrm{d}t}\frac{\mathrm{d}^{k-1}}{\mathrm{d}s^{k-1}}
		\exp\big[\Upsilon_{P_{t+s}-P_t}\big]\Phi^k[M;h_t](\omega_t,P_t)\bigg|_{t,s=0}\\&=
		\frac{\mathrm{d}^{k-1}}{\mathrm{d}s^{k-1}}
		\frac{\mathrm{d}}{\mathrm{d}t}\exp\big[\Upsilon_{P_{t+s}-P_t}\big]\bigg|_{t,s=0}\Phi^k[M;h](\omega,P)\\&+
		\frac{\mathrm{d}^{k-1}}{\mathrm{d}s^{k-1}}
		\exp\big[\Upsilon_{P_{t+s}-P_t}\big]\frac{\mathrm{d}}{\mathrm{d}t}\Phi^k[M;h_t](\omega_t,P_t)\bigg|_{t,s=0}\\&=
		\frac{\mathrm{d}^{k-1}}{\mathrm{d}s^{k-1}}
		\frac{\mathrm{d}}{\mathrm{d}t}\exp\big[\Upsilon_{P_{t+s}-P_t}\big]\bigg|_{t,s=0}\Phi^k[M;h](\omega,P)\\&+
		\frac{\mathrm{d}^{k-1}}{\mathrm{d}s^{k-1}}
		\exp\big[\Upsilon_{P_{t+s}-P_t}\big]\frac{\mathrm{d}}{\mathrm{d}t}\beta_t\big[\Phi^k[M;h](\omega)\big](P_t)\bigg|_{t,s=0}\,.
	\end{align*}
	In the last equality we used equation \eqref{Eqn: PPA} for $k=1$. In view of the definition of $\beta_t$ we find
	\begin{align*}
		\frac{\mathrm{d}^{k-1}}{\mathrm{d}s^{k-1}}
		\frac{\mathrm{d}}{\mathrm{d}t}\exp\big[\Upsilon_{P_{t+s}-P_t}\big]\bigg|_{t,s=0}\Phi^k[M;h](\omega,P)&+
		\frac{\mathrm{d}^{k-1}}{\mathrm{d}s^{k-1}}
		\exp\big[\Upsilon_{P_{t+s}-P_t}\big]\frac{\mathrm{d}}{\mathrm{d}t}\beta_t\big[\Phi^k[M;h](\omega)\big](P_t)\bigg|_{t,s=0}\\&=
		\frac{\mathrm{d}^{k-1}}{\mathrm{d}s^{k-1}}
		\frac{\mathrm{d}}{\mathrm{d}t}\exp\big[\Upsilon_{P_{t+s}-P_t}\big]\bigg|_{t,s=0}\Phi^k[M;h](\omega,P)\\&+
		\frac{\mathrm{d}^{k-1}}{\mathrm{d}s^{k-1}}
		\exp\big[\Upsilon_{P_{t+s}-P_t}\big]\frac{\mathrm{d}}{\mathrm{d}t}
		\exp\big[\Upsilon_{P_t-P}\big]\Phi^k[M;h](\omega,P)\bigg|_{t,s=0}\\&=
		\frac{\mathrm{d}^{k-1}}{\mathrm{d}s^{k-1}}
		\frac{\mathrm{d}}{\mathrm{d}t}\exp\big[\Upsilon_{P_{t+s}-P}\big]\Phi^k[M;h](\omega,P)\bigg|_{t,s=0}\\&=
		\frac{\mathrm{d}^k}{\mathrm{d}s^k}\beta_{s}\big[\Phi^k[M;h](\omega)\big](P_s)\bigg|_{s=0}\,.
	\end{align*}
	This entails the result sought.
\end{remark}

\noindent We end this section with the following result, whose proof can be adapted mutatis mutandis from the counterpart in \cite[Section 6]{Hollands-Wald-05}.

\begin{theorem}\label{Theorem: existence of Wick monomials with satisfy Leibniz rule and PPA}
	If $\mathrm{D}>2$ there exists a family of Wick monomials $\lbrace\Phi^{\text{\ul{k}}}\rbrace_{\text{\ul{k}}}$ as per Definition \ref{Def: Wick monomials with derivatives} which satisfies both the Leibniz rule and the PPA as per Definition \ref{Def: Leibniz rule} and \ref{Def: PPA}.
\end{theorem}

\section{Interacting models}\label{Section: Interacting models}

Up to this point, we have considered the $\ast$-algebra of observables $\mathcal{A}[M;h]$ associated with the quadratic Lagrangian $\mathcal{L}$ defined in Equation \eqref{Eqn: Lagrangian Density}. In this last section we outline the construction of a $\ast$-algebra $\mathcal{A}_V[M;h]$ of observables instead associated with a local perturbation of $\mathcal{L}$, that is $\mathcal{L}_V\vcentcolon=\mathcal{L}+V$, where $V=V^*\in\mathcal{A}[M;h]$ plays the r\^ole of an interaction term. Note that we could relax the requirement of $V$ being local provided that covariance is preserved; yet, we will not discuss further this option. In the following our analysis will rely on the perturbative approach to interacting AQFT \cite{Brunetti:2015vmh,Rejzner:2016hdj}, see also \cite{Hollands-Wald-03}. In this framework $V$, the perturbation, is multiplied by a formal parameter $\lambda$ with respect to which observables are expanded as a formal power series.
Convergence of the such series will not be discussed, since this problem can be dealt with only in special cases \cite{Bahns-Rejzner-18,Bahns-Fredenhagen-Rejzner-17,Drago-19}.

More precisely our goal is the following. We consider a local and covariant algebra built via the functor $\mathcal{A}$ defined in Proposition \ref{Prop: local and covariant algebra of local polynomials}. Moreover we call $\mathsf{P}_{\mathrm{loc}}\colon\bkgg\to\vect$ the functor such that, for any $[M;h]\in\mathsf{Obj}(\bkgg)$,  $\mathsf{P}_{\mathrm{loc}}[M;h]\subseteq\mathcal{P}_{\mathrm{loc}}[M;h]$ is the vector space generated by a family of Wick powers $\lbrace\Phi^k\rbrace_{k\in\mathbb{N}}$.
We shall construct a linear map $\mathsf{R}_V\colon\mathsf{P}_{\mathrm{loc}}[M;h]\to\mathcal{A}[M;h][[\lambda]]$ such that, given any local and covariant observable $F$ such that $F[M;h](\Gamma_{\mathrm{c}}^{1}[M;h])\subseteq\mathsf{P}_{\mathrm{loc}}[M;h]$, $\mathsf{R}_V(F)$ is the expansion of $F$ as a formal power series with respect to $\lambda$ with coefficients being local and covariant observables.
Heuristically $F$ should be thought as an element in the algebra of interacting observables associated with $\mathcal{L}_V$, while $\mathsf{R}_V(F)$ represents its expansion in terms of observables of the free algebra $\mathcal{A}[M;h]$. Hence the image of $\mathsf{R}_V$, that is $\mathsf{R}_V(\mathsf{P}_{\mathrm{loc}}[M;h])\subseteq\mathcal{A}[M;h][[\lambda]]$, yields a perturbative representation of the $\ast$-algebra of interacting observables.

\paragraph{Construction of the map $\mathsf{R}_V$.}
Our starting point is the work of \cite{Keller-09}. In this paper it is argued that the map $\mathsf{R}_V$ should be realized as the algebraic version of the formal path integral formula
\begin{align}\label{Eqn: naive formula for Moeller}
	\mathsf{R}_V(F)(\psi)\sim\mathcal{N}^{-1}\int_{\mathcal{C}[\varphi]}\mathrm{d}[\varphi]\, e^{-\mathcal{L}(\varphi)-\lambda V(\varphi)}F(\psi-\varphi)=
	(Z_V[P])^{-1}\big(Z_V[P]\cdot_P F[P]\big)(\psi)\,,
\end{align}
where $P$ is any parametrix of the underlying elliptic operator $E$, while $e^{-\mathcal{L}(\varphi)}\mathrm{d}[\varphi]$ represents a Gaussian measure on $\mathcal{C}[\varphi]$ a space of chosen kinematic configurations. In addition, $F\in\mathsf{P}_{\mathrm{loc}}[M;h]$, $\mathcal{N}$ is a normalization factor, while
$Z_V(\psi)\vcentcolon=\exp_\cdot(\lambda V)$ is defined as a formal power series in $\lambda$ and it represents the algebraic version of the partition function of statistical field theory. Thus, it is tempting to interpret \eqref{Eqn: naive formula for Moeller} as a Bogoliubov-like formula. Yet a closer scrutiny of \eqref{Eqn: naive formula for Moeller} unveils that is neither a local nor a covariant expression since a change of parametrix $P\in\operatorname{Par}[M;h]$ yields
\begin{align}\label{Eqn: naive formula for Moeller for generic parametrix}
	\alpha^P_Q\big[(Z_V[P])^{-1}\big(Z_V[P]\cdot_P F[P]\big)\big]=
	Z_V[Q]^{\cdot_{P-Q} -1}\cdot_{P-Q}\big(Z_V[Q]\cdot_Q F[Q]\big)\,,
\end{align}
where $\cdot_{P-Q}$ is defined as in Equation \eqref{Eqn: Algebra product} while $Z_V[Q]^{\cdot_{P-Q} -1}$ denotes the inverse of $Z_V[Q]$ with respect to $\cdot_{P-Q}$.

A possibility to restore the interpretation as a local and covariant observable occurs if, in place of letting the parametrix vary, there would exist a fixed choice of $P_0\in\operatorname{Par}[M;h]$ which is both local and covariant. A close scrutiny of the Lorentzian scenario unveils that, in such a case, this is the solution adopted, see \cite{Brunetti:2015vmh}. As a matter of fact the r\^ole of $P_0$ is played by a fundamental solution of the underlying normally hyperbolic operator, {\it e.g.} the advanced or the retarded propagators. In the category of globally hyperbolic spacetimes, this procedure is manifestly local and covariant. 

On the contrary, working with $(M,g)\in\mathsf{Obj}(\bkgg)$, also this viewpoint is slightly problematic. Following the seminal work \cite{Li_Tam}, the existence of the fundamental solutions of $E$, as in Equation \eqref{Eqn: local_form_E}, is ruled by $\ker E$ on $C^\infty_0(M)$, while its uniqueness by $\ker E$ on $\mathcal{E}(M)$. It is thus not hard to choose $E$ so that one can construct $(M_1,g_1),(M_2,g_2)\in\mathsf{Obj}(\bkgg)$ with an orientation preserving isometric embedding $\chi:M_1\to M_2$, such that, given $G_1$ a fundamental solution of $E$ in $M_1$, there does not exist $G_2$ fundamental solution of $E$ in $M_2$ obeying $\chi^*[G_2|_{\chi[M_1]}]=G_1$.

In view of this last remark and of the preceding discussion, in order to give a local and covariant description of the map $R_V$, we need to hard code in the background data a local and covariant choice of a fundamental solution of the underlying elliptic operator $E$. 

\begin{Definition}\label{Def: Augmented background category}
We call $\bkgg_G$ the category such that
\begin{itemize}
	\item $\obj(\bkgg_G)$ is the collection of pairs $(M;h,G)$, where $M$ denotes a smooth, connected and oriented manifold with empty boundary and with $\dim M=\mathrm{D}\geq 2$. In addition $h\equiv (g, A, c)$ identifies the background data, that is $A\in\Gamma(T^*M)$, $c\in C^\infty(M)$ while $g\in \Gamma(S^2T^*M)$ is a Riemannian metric, while $G\in\mathcal{D}^\prime(M\times M)$ is a fundamental solution of $E$, as in \eqref{Eqn: local_form_E}, {\it i.e.} $GE=EG=\operatorname{Id}_{\mathcal{D}(M)}$.
	\item $\arr(\bkgg_G)$ is the collection of morphisms between $(M;h,G),(M^\prime;h^\prime,G^\prime)\in\obj(\bkgg_G)$ which are specified by an orientation preserving isometric embedding between $\chi:M\to M^\prime$ such that $h=\chi^*h^\prime$ where $h^\prime\equiv(g^\prime,A^\prime,c^\prime)$ and $\chi^*G^\prime = G$.
\end{itemize}
\end{Definition}

\noindent We observe that there exists a forgetful functor $\pi_G$ from $\bkgg_G$ to a subcategory $\pi_G(\bkgg_G)$ of $\bkgg$ which is defined as $\pi_G(M;h,G)=(M;h)$ for every $(M;h,G)\in\mathsf{Obj}(\bkgg)$, while it acts as the identity on the arrows.
As a consequence, for every Euclidean locally covariant theory $\mathcal{A}$ as per Definition \ref{Def: Euclidean locally covariant theory}, we define:
\begin{equation}\label{Eqn: Elcft with G}
\mathcal{A}_G:\bkgg_G\to\alg,\quad\textrm{such that}\quad\mathcal{A}_G=\mathcal{A}\circ\pi_G.
\end{equation}

\noindent With a slight abuse of notation we will refer to $\mathcal{A}_G$ still as an Euclidean locally covariant theory. In view of the new structures that we have introduced, we can now bypass the problem outlined at the beginning of the section as follows.
\begin{Definition}\label{Def: local and covariant algebra of interacting observables}
	Let $\mathcal{A}$ be a local and covariant algebra as in Proposition \ref{Prop: local and covariant algebra of local polynomials} and let $\mathcal{A}_{G}$ be the counterpart as per Equation \eqref{Eqn: Elcft with G}.
	For $(M;h,G)\in\mathsf{Obj}(\bkgg_G)$ let $V=V^*\in\mathcal{A}[M;h]$. 
	For all $F\in\mathcal{A}_G[M;h,G]$ we define $\mathsf{R}_V(F)\in\mathcal{A}_G[M;h,G][[\lambda]]$ as 
	\begin{align}\label{Eqn: Moeller operator for closed manifold}
		\mathsf{R}_V(F)[P]\vcentcolon=Z_V[P]^{\cdot_{P-G} -1}\cdot_{P-G}\big(Z_V[P]\cdot_P F[P]\big)\,,
	\end{align}
	where $Z_V[P]\vcentcolon=\exp_{\cdot_P}(\lambda V)$. We define the $\ast$-algebra of interacting observable on $(M,h)$ to be the $\ast$-subalgebra $\mathcal{A}_V[M;h,G]$ generated by $\mathsf{R}_V(\mathsf{P}_{\mathrm{loc}}[M;h,G])\subseteq\mathcal{A}_G[M;h,G][[\lambda]]$.
\end{Definition}
\noindent Observe that formula \eqref{Eqn: Moeller operator for closed manifold} defines a local and covariant observable as per definition \ref{Def: locally covariant observable}. As a direct consequence of the properties of the structures introduced, the following statement holds true:
\begin{proposition}\label{Prop: the algebra of interacting observable is a local and covariant theory}
	Under the hypothesis of Definition \ref{Def: local and covariant algebra of interacting observables}, $\mathcal{A}_V\colon\bkgg_G\to\alg$ is a Euclidean locally covariant theory in the sense of Definition \ref{Def: Euclidean locally covariant theory} and Equation \eqref{Eqn: Elcft with G}.
\end{proposition}

\begin{remark}
	The M\o ller operator introduced in Equation \eqref{Eqn: Moeller operator for closed manifold} is an intertwiner between $E$ and $E_V\varphi:=E\varphi+V^{(1)}[\varphi]$.
	Indeed, let consider $F[P,\varphi]:=\int fE\varphi$, where $f\in C^\infty_{\mathrm{c}}(M)$.
	A direct computation yields
	\begin{align*}
		\mathsf{R}_V(F+V^{(1)}[\varphi](f))[P]=F[P]\,.
	\end{align*}
\end{remark}

\begin{remark}
	We observe that the problem of a local and covariant choice of a fundamental solution bears a similarity to the failure of isotony in Abelian gauge theories when discussing general local covariance. In this scenario, it was observed in \cite[Sec. 6]{Benini:2013ita} that a possible way to circumvent this problem consists of choosing a subcategory of the background data which possesses a terminal object. At the level of algebras this specialization leads to the identification of a Haag-Kastler net of observables. We could have adopted such viewpoint also in the analysis of the case in hand choosing a terminal object in $\bkgg$ rather than defining an additional background datum as in $\bkgg_G$. It is not difficult to realize from Definition \ref{Def: Augmented background category} that our choice includes the first as a special case. Hence Haag-Kastler nets of observables are all realized in Definition \ref{Def: local and covariant algebra of interacting observables} and in Proposition \ref{Prop: the algebra of interacting observable is a local and covariant theory}.
\end{remark}

\appendix

\section{The Peetre-Slov\'ak Theorem}\label{Appendix: Peetre-Slovak}
In this section we briefly review the Peetre-Slov\'ak theorem together with all the ancillary notions. The interested reader may refer to \cite{Navarro-Sancho-14} and to \cite{Khavkine:2014zsa} for a more in detail discussion.

\begin{remark}
Let $E\overset{\pi_E}{\to}B$ be a bundle over the smooth manifold $B$. With $J_rE$, $r\in\N$, we denote the $r$-jet bundle over the base $B$ \cite{Kolar-Michor-Slocvak-93}.
\end{remark}

\begin{Definition}\label{Def: Differential operators}
Let $E\overset{\pi_E}{\to}B$ and $F\overset{\pi_F}{\to}B$ be bundles over the same smooth manifold $B$. Consider a map $D:\Gamma(E)\to\Gamma(F)$, we say that:
\begin{enumerate}
\item $D$ is a \emph{differential operator of globally bounded order} $r\in\N$ if there exists a smooth map $d\colon J_rE\to F$ such that $\pi_F\circ d=\pi_{J_rE}$ and 
	\begin{align}\label{Eqn: definition of differential operator of globally bounded order}
	D(\varepsilon)=d(j_r\varepsilon),\qquad\forall\varepsilon\in\Gamma(E)\,,
	\end{align}
	with $j_r\varepsilon\in\Gamma(j_rE)$ denoting the $r$-jet extension of $\varepsilon$;
\item $D$ is a \emph{differential operator of locally bounded order} if, for any $x\in B$ and $\varepsilon\in\Gamma(E)$ there exist:
	\begin{itemize}
	\item a relatively compact open set $U\subset B$ containing $x$;
	\item an integer $r\in\mathbb{N}$, as well as a neighbourhood $Z^r\subseteq J_rE$ of $j_r\varepsilon_0(U)$ such that $\pi_{J_rE}Z^r=U$,
		\item a smooth map $d\colon Z^r\to F$ such that $\pi_F\circ d=\pi_{J_rE}$
   so that
	\begin{align}\label{Eqn: definition of differential operator of locally bounded order}
	D(\varepsilon)(x)=d(j_r\varepsilon)(x)\,,
	\end{align}
	for any $x\in U$ and $\varepsilon\in\Gamma(E)$ with $j_r\varepsilon(U)\subseteq Z^r$.
	\end{itemize}
\end{enumerate}
\end{Definition}

In this setting, the Peetre-Slov\'ak Theorem is a result giving sufficient condition for a differential operator to be of locally bounded order.

\begin{remark}
Denoting with $\pi_d\colon B\times\mathbb{R}^d\to B$ the canonical projection to $B$, the pull-back bundle $\pi_d^*E\stackrel{\pi_{\pi_d^*E}}{\to}B\times\mathbb{R}^d$ is the smooth bundle defined by
\begin{align}\label{Eqn: pull-back bundle}
\pi^*E\vcentcolon=\{(s,x,e)\in\mathbb{R}^d\times B\times E|\;\pi_E(e)=\pi_d(s,x)\}\simeq\mathbb{R}^d\times E\,.
\end{align}
Denoting with $\pi_{d,E}$ the projection $\pi_{d,E}\colon\pi_d^*E\to E$, each smooth section $\zeta\in\Gamma(\pi_d^*E)$ induces a smooth family of sections $\{\zeta_s\}_{s\in\mathbb{R}^d}$ in $\Gamma(E)$ defined by $\zeta_s(x)\vcentcolon=\pi_{d,E}\zeta((s,x))$ which, in turn, depends smoothly on the parameter $s\in\mathbb{R}^d$.
\end{remark}
\begin{Definition}\label{Def: smooth compactly supported d-dimensional family of variations}
	Let $d\in\N$ and consider a smooth family of sections $\{\zeta_s\}_{s\in\mathbb{R}^d}$ in $\Gamma(E)$ induced by a smooth section $\zeta\in\Gamma(\pi_d^*E)$.
	We say that $\{\zeta_s\}_{s\in\mathbb{R}^d}$ is a smooth compactly supported $d$-dimensional family of variations if there exists a compact $K\subseteq B$ such that $\zeta(s,x)=\zeta(s',x)$ for all $x\notin K$ and for all $s,s'\in\mathbb{R}^d$.
\end{Definition}
\begin{Definition}\label{Def: weak regularity condition}
	A map $D\colon\Gamma(E)\to\Gamma(F)$ is called weakly-regular if, for any $d\in\N$ and for all smooth compactly supported $d$-dimensional families of variations $\{\zeta_s\}_{s\in\mathbb{R}^d}$ -- see Definition \ref{Def: smooth compactly supported d-dimensional family of variations} --, $\psi_s\vcentcolon=D\zeta_s$ is a smooth compactly supported $d$-dimensional family of variations.
\end{Definition}
\begin{theorem}[Peetre-Slov\'ak]\label{Theorem: Peetre-Slovak's theorem}
	Let $D\colon\Gamma(E)\to\Gamma(F)$ be a smooth map such that
	\begin{itemize}
		\item
		for all $\varepsilon\in\Gamma(E)$ and for all $x\in B$, $D\varepsilon(x)$ depends only on the germ of $\varepsilon$ at $x\in B$, i.e. $(D\varepsilon)(x)=(D\widetilde{\varepsilon})(x)$ for all $\widetilde{\varepsilon}\in\Gamma(E)$ which coincides with $\varepsilon$ in a neighbourhood of $x$;
		\item
		$D$ is weakly regular as per Definition \ref{Def: weak regularity condition}.
	\end{itemize}
	Then $D$ is a differential operator of locally bounded order as per Definition \ref{Def: Differential operators}.
\end{theorem}

\section*{Acknowledgments}

We are grateful to Matteo Capoferri, Federico Faldino, Klaus Fredenhagen, Igor Khavkine, Valter Moretti and Nicola Pinamonti for the valuable comments.


\begin{thebibliography}{} 
	
\bibitem[BR18]{Bahns-Rejzner-18}
D. Bahns, K. Rejzner,
\textit{The Quantum Sine Gordon model in perturbative AQFT},
Commun. Math. Phys. (2018) 357: 421.
	
\bibitem[BFK17]{Bahns-Fredenhagen-Rejzner-17}
D. Bahns, K. Fredenhagen, K.Rejzner,
\textit{Local nets of von Neumann algebras in the Sine-Gordon model},
arXiv:1712.02844 [math-ph].
	
\bibitem[BDHS14]{Benini:2013ita}
M.~Benini, C.~Dappiaggi, T.~P.~Hack and A.~Schenkel,
\textit{``A C*-algebra for quantized principal U(1)-connections on globally hyperbolic Lorentzian manifolds,''}
Commun.\ Math.\ Phys.\  {\bf 332} (2014) 477
[arXiv:1307.3052 [math-ph]].

\bibitem[BDH13]{Benini:2013fia}
M.~Benini, C.~Dappiaggi and T.~P.~Hack,
\emph{``Quantum Field Theory on Curved Backgrounds -- A Primer,''}
Int.\ J.\ Mod.\ Phys.\ A {\bf 28} (2013) 1330023
[arXiv:1306.0527 [gr-qc]].
	
\bibitem[BPS19]{Benini_Perin_Schenkel}
M.~Benini, M.~Perin and A.~Schenkel,
\textit{``Model-independent comparison between factorization algebras and algebraic quantum field theory on Lorentzian manifolds,''}
arXiv:1903.03396 [math-ph]
	
\bibitem[BDGR18]{Brouder-Dang-Gengoux-Rejzner-18}
C.~Brouder, N.~V.~Dang, C.~Laurent-Gengoux and K.~Rejzner,
\textit{``Properties of field functionals and characterization of local functionals,''}
J.\ Math.\ Phys.\  {\bf 59} (2018) no.2,  023508
[arXiv:1705.01937 [math-ph]].	

\bibitem[BDFY15]{Brunetti:2015vmh}
R.~Brunetti, C.~Dappiaggi, K.~Fredenhagen and J.~Yngvason,
\textit{``Advances in algebraic quantum field theory,''}
Math.\ Phys.\ Stud.\  (2015), Springer 453p.	
	
\bibitem[BDF09]{Brunetti-Duetsch-Fredenhagen-09}
R. Brunetti , M. Duetsch, K. Fredenhagen,
\textit{``Perturbative Algebraic Quantum Field Theory and the Renormalization Groups,'},
Adv.\ Theor.\ Math.\ Phys.\  {\bf 13} (2009) no.5, 1541
[arXiv:0901.2038 [math-ph]].
	
\bibitem[BFK96]{Brunetti:1995rf}
R.~Brunetti, K.~Fredenhagen and M.~Kohler,
\textit{``The Microlocal spectrum condition and Wick polynomials of free fields on curved space-times,''}
Commun.\ Math.\ Phys.\  {\bf 180} (1996) 633
[gr-qc/9510056].	
	
\bibitem[BFR12]{Brunetti-Fredenhagen-Ribeiro-12}	
R. Brunetti, K. Fredenhagen, P. Laurisen Ribeiro,
\textit{``Algebraic Structure of Classical Field Theory: Kinematics and Linearized Dynamics for Real Scalar Fields''},
Commun. Math. Phys. 368, 519–584 (2019).
	
\bibitem[BF00]{Brunetti:1999jn}
R.~Brunetti and K.~Fredenhagen,
\textit{``Microlocal analysis and interacting quantum field theories: Renormalization on physical backgrounds,''}
Commun.\ Math.\ Phys.\  {\bf 208} (2000) 623
[math-ph/9903028].	
	
\bibitem[BFV03]{Brunetti:2001dx}
R.~Brunetti, K.~Fredenhagen and R.~Verch,
\textit{``The Generally covariant locality principle: A New paradigm for local quantum field theory,''}
Commun.\ Math.\ Phys.\  {\bf 237} (2003) 31
[math-ph/0112041].	

\bibitem[CDDR18]{CDDR18}
M.~Carfora, C.~Dappiaggi, N.~Drago, and P.~Rinaldi,
\textit{``Ricci Flow from the Renormalization of Nonlinear Sigma Models in the Framework of Euclidean Algebraic Quantum Field Theory,''}
Commun. Math. Phys. (2019). 

\bibitem[Da19a]{Dang:2019byu}
N.~V.~Dang,
\textit{``Renormalization of determinant lines in Quantum Field Theory,''}
arXiv:1901.10542 [math-ph].

\bibitem[Da19b]{Dang-18}
N. V. Dang,
\textit{``Wick Squares of the Gaussian Free Field And Riemannian Rigidity.''},
arXiv:1902.07315 [math-ph].	

\bibitem[DD16]{Dappiaggi-Drago-16}
C. Dappiaggi, N. Drago,
\textit{``Constructing Hadamard States via an Extended Møller Operator''},
Lett.\ Math.\ Phys.\  {\bf 106} (2016) no.11,  1587
[arXiv:1506.09122 [math-ph]].

\bibitem[D19]{Drago-19}
N. Drago,
\textit{``Thermal state with quadratic interaction''},
Ann. Henri Poinc. {\bf 20} (2019) 905.
arXiv:1711.01072 [math-ph].
	
\bibitem[DHP16]{Drago-Hack-Pinamonti-2016}
N. Drago, T-P. Hack, N. Pinamonti,
\textit{``The generalized principle of Perturbative Agreement and the thermal mass"},
Ann. Henri Poinc. {\bf 18} (2017) no.3,  807
[arXiv:1502.02705 [math-ph]]. 

\bibitem[FR12]{FR12}
K.~Fredenhagen and K.~Rejzner,
\textit{``Batalin-Vilkovisky formalism in the functional approach to classical field theory,''}
Commun.\ Math.\ Phys.\  {\bf 314} (2012) 93
[arXiv:1101.5112 [math-ph]].

\bibitem[FR13]{FR13}
K.~Fredenhagen and K.~Rejzner,
\textit{``Batalin-Vilkovisky formalism in perturbative algebraic quantum field theory,''}
Commun.\ Math.\ Phys.\  {\bf 317} (2013) 697,
[arXiv:1110.5232 [math-ph]].

\bibitem[G98]{Garabedian-98}
P.R. Garabedian,
\textit{``Partial Differential Equations''},
(1998) Chelsea Pub Co, 672 p.

\bibitem[GR17]{Gwilliam-Rejzner-17}
O. Gwilliam, K. Rejzner,
\textit{``Relating nets and factorization algebras of observables: free field theories''},
arXiv:1711.06674 [math-ph].

\bibitem[HK63]{Haag:1963dh}
R.~Haag and D.~Kastler,
\textit{``An Algebraic approach to quantum field theory,''}
J.\ Math.\ Phys.\  {\bf 5} (1964) 848.
	
\bibitem[HW01]{Hollands-Wald-01}
S.~Hollands and R.~M.~Wald,
\textit{``Local Wick polynomials and time ordered products of quantum fields in curved space-time,''}
Commun.\ Math.\ Phys.\  {\bf 223} (2001) 289
[gr-qc/0103074].
	
\bibitem[HW02]{Hollands-Wald-02}
S.~Hollands and R.~M.~Wald,
\textit{``Existence of local covariant time ordered products of quantum fields in curved space-time,''}
Commun.\ Math.\ Phys.\  {\bf 231} (2002) 309
[gr-qc/0111108].
	
\bibitem[HW03]{Hollands-Wald-03}
S.~Hollands and R.~M.~Wald,
\textit{``On the renormalization group in curved space-time,''}
Commun.\ Math.\ Phys.\  {\bf 237} (2003) 123
[gr-qc/0209029].
	
\bibitem[HW05]{Hollands-Wald-05}
S.~Hollands and R.~M.~Wald,
\textit{``Conservation of the stress tensor in interacting quantum field theory in curved spacetimes,''}
Rev.\ Math.\ Phys.\  {\bf 17} (2005) 227
[gr-qc/0404074].

\bibitem[H\"o03]{Hormander-83}
L. H\"ormander,
\textit{The analysis of linear partial differential operators I - distribution theory and fourier analysis}, (2003) Springer, 440p.
	
\bibitem[Kel10]{Keller:2010xq}
K.~J.~Keller,
\textit{``Dimensional Regularization in Position Space and a Forest Formula for Regularized Epstein-Glaser Renormalization,''} (2009), PhD thesis, U. Hamburg,
arXiv:1006.2148 [math-ph].
	
\bibitem[Kel09]{Keller-09}
 K.~J.~Keller,
\textit{``Euclidean Epstein-Glaser Renormalization,''}
J.\ Math.\ Phys.\  {\bf 50} (2009) 103503
[arXiv:0902.4789 [math-ph]].
	
\bibitem[KMM17]{Khavkine-Melati-Moretti-17}
I. Khavkine, A. Melati, V. Moretti,
\textit{``On Wick polynomials of boson fields in locally covariant algebraic QFT''},
Ann. Henri Poinc. {\bf 20} (2019) 929,
arXiv:1710.01937 [math-ph].
	
\bibitem[KM16]{Khavkine:2014zsa}
I.~Khavkine and V.~Moretti,
\textit{``Analytic Dependence is an Unnecessary Requirement in Renormalization of Locally Covariant QFT,''}
Commun.\ Math.\ Phys.\  {\bf 344} (2016) no.2,  581
[arXiv:1411.1302 [gr-qc]].

\bibitem[KM14]{Khavkine:2014mta}
I.~Khavkine and V.~Moretti,
\emph{``Algebraic QFT in Curved Spacetime and quasifree Hadamard states: an introduction,''}
Chapter 5, Advances in Algebraic Quantum Field Theory, R. Brunetti
et al. (eds.), Springer, 2015
[arXiv:1412.5945 [math-ph]].
	
\bibitem[KMS93]{Kolar-Michor-Slocvak-93}
I. Kol\'a\v r, P. W. Michor, J. Slov\'ak,
\textit{Natural Operations in Differential Geometry}, (1993)
Springer, 434p.
	
\bibitem[LT87]{Li_Tam}
P. Li and L.-F. Tam,
\textit{``Symmetric Green's Functions on Complete Manifolds''}, 
Am. J. of Math. {\bf 109} (1987), 1129.
	
\bibitem[Lin13]{Lindner:2013ila}
F.~Lindner,
\textit{``Perturbative Algebraic Quantum Field Theory at Finite Temperature,''} (2013), PhD thesis, U. Hamburg.
	
\bibitem[NS14]{Navarro-Sancho-14}
 J.	Navarro, J. B. Sancho,
\textit{``Peetre-Slov\'ak theorem revisited''},
arXiv:1411.7499 [math.DG].	
		
\bibitem[OS73]{Osterwalder:1973dx}
K.~Osterwalder and R.~Schrader,
\textit{``Axioms For Euclidean Green's Functions,''}
Commun.\ Math.\ Phys.\  {\bf 31} (1973) 83.

\bibitem[OS75]{Osterwalder:1974tc}
K.~Osterwalder and R.~Schrader,
\textit{``Axioms for Euclidean Green's Functions. 2.,''}
Commun.\ Math.\ Phys.\  {\bf 42} (1975) 281.	

\bibitem[Rej16]{Rejzner:2016hdj}
K.~Rejzner,
\textit{``Perturbative Algebraic Quantum Field Theory : An Introduction for Mathematicians,''}
Math.\ Phys.\ Stud.\  (2016), Springer 180p.
	
\bibitem[Sch98]{Schlingemann:1998cw}
D.~Schlingemann,
\textit{``From Euclidean field theory to quantum field theory,''}
Rev.\ Math.\ Phys.\  {\bf 11} (1999) 1151
[hep-th/9802035].

\bibitem[Shu87]{Shubin}
M. A. Shubin,
\textit{Pseudodifferential operators and spectral theory},
Springer-Verlag Berlin Heidelberg (1987), 302p.
	
\bibitem[Wa79]{Wald-79}
 R.~M.~Wald,
\textit{``On The Euclidean Approach To Quantum Field Theory In Curved Space-time,''}
Comm.\ Math.\ Phys.\  {\bf 70} (1979) 221. 

\bibitem[Wel08]{Wells}
R. O. Wells,
\textit{Differential Analysis on Complex Manifolds}, (2008) Springer, 299p.

\bibitem[Za15]{Zahn:2013ywa}
J.~Zahn,
\textit{``Locally covariant charged fields and background independence,''}
Rev.\ Math.\ Phys.\  {\bf 27} (2015) no.07,  1550017
[arXiv:1311.7661 [math-ph]].

\end{thebibliography}
\end{document}